\makeatletter
\AddToHook{package/before/amssymb}{\let\Bbbk\@undefined}
\makeatother

\documentclass[sigconf,nonacm]{acmart}
\AtBeginDocument{%
  \providecommand\BibTeX{{%
    Bib\TeX}}}

\usepackage[dvipsnames]{xcolor}  
\usepackage{tikz}
\usetikzlibrary{matrix,positioning,calc,arrows.meta,fit}

\usepackage{mathtools}
\usepackage{enumitem}
\usepackage{graphicx}
\usepackage{filecontents}
\usepackage{amsmath,amsfonts,amsthm}
\usepackage{hyperref}
\usepackage[most]{tcolorbox}
\usepackage{booktabs,tabularx}
\usepackage{url}
\usepackage{algorithm}
\usepackage{algpseudocode}
\usepackage{extpfeil}
\usepackage{comment}
\usepackage{textcomp}
\usepackage{xcolor}
\usepackage{mathrsfs}
\usepackage{ulem}
\usepackage{xspace}

\def\BibTeX{{\rm B\kern-.05em{\sc i\kern-.025em b}\kern-.08em
    T\kern-.1667em\lower.7ex\hbox{E}\kern-.125emX}}

\usepackage{ulem}

\begin{document}

\title{Rethinking the Security of DP-SGD: A Corrected Analysis of Differentially Private Machine Learning}

\author{Wenhao Wang}
\affiliation{%
  \institution{Monash University}
  \country{Australia}
}
\email{wenhao.wang@monash.edu}

\author{Shujie Cui}
\affiliation{%
  \institution{Monash University}
  \country{Australia}
}
\email{shujie.cui@monash.edu}

\author{Hui Cui}
\affiliation{%
  \institution{Monash University}
  \country{Australia}
}
\email{hui.cui@monash.edu}

\author{Xingliang Yuan}
\affiliation{%
  \institution{University of Melbourne}
  \country{Australia}
}
\email{xingliang.yuan@unimelb.edu.au}

\begin{abstract}
Differentially Private Stochastic Gradient Descent (DP-SGD) has been widely adopted to protect training data in machine learning. 
The privacy guarantee of DP\text{-}SGD and the DP-mechanisms built upon it is usually analyzed through a formal security game, in which an adversary infers whether a particular 
individual data record is included in the training dataset based on the mechanism’s output. 
Privacy leakage is characterized by the adversary’s privacy curve, which reports the false negative rate (FNR) as a function of the false positive rate (FPR).
The privacy guarantee is defined as the lower bound of this curve.

We observe that the privacy guarantees claimed for these mechanisms in many existing papers are derived from a mismatched game setting. 
 
Specifically, they formalize the mechanisms as the Subsampled Gaussian Mechanism (SGM), where Gaussian noise is added to the sum of gradients computed from a Poisson-sampled batch of data.
Indeed, the training procedure in these mechanisms introduces an additional normalization step: the noisy sum is further normalized either by the expected batch size or by the sampled batch size. Thus, these mechanisms should be formalized as either the Expected-Averaged SGM (EASGM) or the Batch-Averaged SGM (ASGM). 
The privacy auditing of DP-SGD suffers from the same issue, as it assesses the privacy guarantee of DP-SGD by treating it as an SGM.

We therefore re-analyze the privacy guarantee of these mechanisms under the corresponding EASGM and ASGM formalizations. Our analysis shows that, in theory, these DP\text{-}SGD mechanisms can yield weaker privacy guarantees than the SGM-based guarantee, suggesting that, in some settings, the true privacy leakage can exceed the reported SGM-based guarantee. 
 
We also empirically audit the leakage of implementations of four state-of-the-art DP\text{-}SGD algorithms, including the implementation used in Meta’s Opacus library, and show that empirical leakage exceeds the SGM-based guarantees. Finally, we conduct a thorough code audit of Opacus versions v0.9.0--v1.5.4 and derive a privacy guarantee for the latest Opacus implementation.

\end{abstract}

\maketitle
\section{Introduction}

\begin{figure*}[t]
\centering

\begin{minipage}[t]{0.47\textwidth}\vspace{0pt}
\usetikzlibrary{matrix,positioning,calc,arrows.meta} 


\begin{tcolorbox}[
  enhanced,
  width=\linewidth,
  colback=white, colframe=black, boxrule=0.6pt,
  sharp corners, left=6pt, right=6pt, top=4pt, bottom=6pt  
]
\small
\textbf{Game Parameters:} $(d,\mathcal{D},\mathcal{M}(q,C,\sigma))$\\[1pt] 
\vspace{-2pt}                                                    

\noindent
\begin{tikzpicture}[
  >=Latex,
  every node/.style={outer sep=0pt}
]
  \matrix (m) [
    matrix of nodes,
    nodes in empty cells,
    nodes={anchor=base, text height=1.6ex, text depth=.25ex},
    column sep=16.5mm,
    row sep=1.2mm,                 
    anchor=west
  ]{
    \textbf{Adversary} & \textbf{Challenger} \\
    $D^- \sim \mathcal{D}^{N}$ & \phantom{$b \sim \{0,1\}$} \\
    $D_0 := D^-$, $D_1 := D^- \cup \{d\}$ & $b \sim \{0,1\}$ \\
    $\mathcal{A}(O,\mathcal{M},D_0,D_1)\rightarrow \hat b$ & $O=\textcolor{red}{\mathcal{M}(D_b)}$ \\
  }; 

  \draw[->, line width=0.5pt]
    ([xshift=0.6mm,yshift=1pt] m-3-1.base east)
      -- node[above=2pt] {$D_{0},\,D_{1}$}
    ([xshift=-1.4mm,yshift=1pt] m-3-2.base west);

  \draw[<-, line width=0.5pt]
    ([xshift=0.6mm,yshift=1pt] m-4-1.base east)
      -- node[above=2pt] {$O$}
    ([xshift=-1.4mm,yshift=1pt] m-4-2.base west);
\end{tikzpicture}
\end{tcolorbox}

\end{minipage}
\hfill
\begin{minipage}[t]{0.47\textwidth}\vspace{0pt}
\begin{tcolorbox}[
  enhanced, width=\textwidth,
  colback=white, colframe=black, boxrule=0.6pt,
  sharp corners, left=6pt, right=6pt, top=6pt, bottom=6pt
]
\small
\textbf{Correct Mechanism formalization:}\\[4pt]
1.\;EASGM (def \ref{def:EASGM}):\; $\mathcal{M}(D_b)
  = \Bigl(\sum_{i\in B}\bar g_i + G\Bigr)\big/ (N\!\cdot\! q)$\\
2.\;ASGM (def \ref{def:ASGM}):\; $\mathcal{M}(D_b)
  = \Bigl(\sum_{i\in B}\bar g_i + G\Bigr)\big/ |B|$\\[6pt]
\textbf{Incorrect Mechanism formalization:}\\[4pt]
3.\;SGM (def \ref{def:SGM}):\; $\textcolor{red}{\mathcal{M}(D_b)= \displaystyle\sum_{i\in B}\bar g_i + G}$
\end{tcolorbox}
\end{minipage}
\caption{The left figure illustrates the procedure of the DP game, while the right figure presents the correct mechanism formalization proposed in this paper and the incorrect mechanism formalizations used in prior works~\cite{abadi2016deep,bu2020deep,fu2024dpsur}.}
\label{fig:privacy-game-and-mech}
\end{figure*}
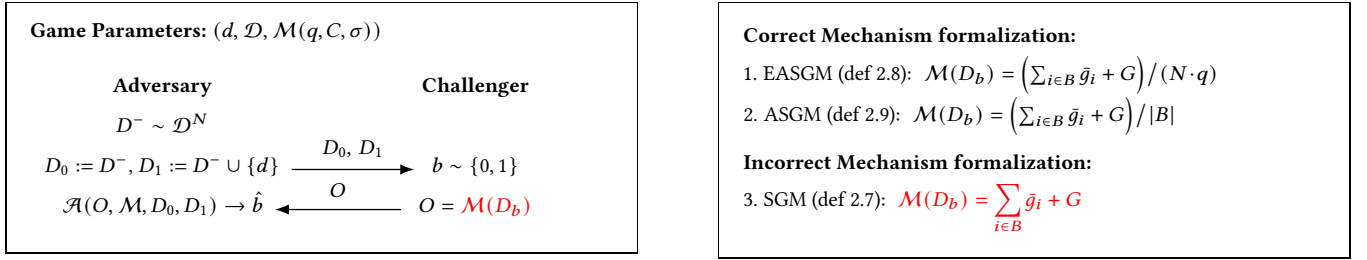

Deep learning models leak sensitive information of the training data \cite{du2025cascading,xiang2025tight}. Differential Privacy (DP) is the de facto standard and has been widely adopted for mitigating such risks~\cite{10.1145/3658644.3690194}.
 For example, Google recently introduced VaultGemma \cite{google2025vaultgemma}, a large language model (LLM) trained with DP to provide strong privacy guarantees, which they describe as “the world’s most capable differentially private LLM.”

 \noindent\textbf{DP-SGD brief and privacy analysis.} The state-of-the-art differentially private training method is Differentially Private Stochastic Gradient Descent (DP-SGD)~\cite{abadi2016deep}, and it has been used by companies including Google \cite{odp-dp-semantics}, Meta \cite{yousefpour2021opacus} for privacy protection in machine learning services. Meta also provides an open-source DP-SGD implementation through its Opacus library~\cite{yousefpour2021opacus}. 
 DP-SGD updates the model round by round, and each round consists of three basic steps: (i) Form a batch $B$ via Poisson sampling, including each data point with probability \(q\); (ii) Compute the gradient $g_i$ of each sample in the batch; and (iii) Clip each $g_i$ at the threshold $C$ into $\overline{g}_i$, get the aggregated value $\sum_{i\in B}\overline{g}_i$, and add calibrated Gaussian noise $G$ with standard deviation $\sigma$ to obtain the noisy sum $\mathcal{S}=\sum_{i\in B}\overline{g}_i + G$. These three steps are called \textit{Subsampled Gaussian Mechanism (SGM)}. In different DP-SGD mechanisms, $\mathcal{S}$ could be normalized and then used to update the model. For instance, in the initial DP\text{-}SGD proposed by Google~\cite{abadi2016deep}, the noisy sum of clipped gradients $\mathcal{S}$ is normalized by the expected batch size \(N \cdot q\), which we denote as the \textit{Expected-Averaged SGM (EASGM)}~\cite{abadi2016deep,11023490}, where $N$ is the size of the training dataset. 
In a recent DP\text{-}SGD variant~\cite{bu2020deep} (Algorithm 1 of \cite{bu2020deep_fullpaper}), $\mathcal{S}$ is normalized by the batch size \(|B|\), and is referred to as \textit{Batch-Averaged SGM (ASGM)}.

The privacy guarantee of the DP-SGD training process can be decomposed into per-round guarantees, which are then composed using standard composition techniques~\cite{dong2019gaussian}. 
DP-SGD’s per-round privacy guarantee is typically analyzed through a formalized game between a challenger and 
an adversary~\cite{bu2020deep,nasr2023tight}. In this game, the challenger applies one round of DP-SGD to generate a model update on one of two neighboring datasets, $D$ and $D'$, which differ by a single individual $d$, and then releases the update. The adversary has the knowledge of both neighboring datasets and the parameters of DP-SGD, including the sampling ratio $q$ and the noise scale $\sigma$, and aims to determine which dataset has been used to generate the model update. Privacy leakage is characterized by the adversary’s privacy curve, which reports the false negative rate (FNR) as a function of the false positive rate (FPR). The privacy guarantee is defined as the lower bound of this function \cite{dong2019gaussian}.
Depending on the tool or framework (e.g., RDP \cite{mironov2019r} and f-DP \cite{dong2019gaussian}) used for the analysis, the obtained privacy guarantee could be close or far away from the optimal lower bound. For instance, Edgeworth Expansion \cite{zheng2020sharp} can yield a tighter privacy guarantee than the Central Limit Theorem \cite{dong2019gaussian}. 
The closer to the optimal lower bound, the tighter the privacy guarantee. Tighter guarantees are desirable because they certify the same privacy level with less noise and improved output quality. 
However, exact tightness is often difficult to obtain; thus, prior work often derives \textit{nearly tight} guarantees instead \cite{wang2023unified}.

\sloppy \noindent \textbf{Mismatched formalization of normalized DP-SGD.}
However, for the privacy analysis, as shown in Fig.~\ref{fig:privacy-game-and-mech}, prior studies~\cite{abadi2016deep,bu2020deep} formalize one-round EASGM-based and ASGM-based DP\text{-}SGD as SGM in the game, thereby omitting the normalization step in the training procedure, and accordingly use the SGM guarantee to bound their per-round privacy leakage.
Indeed, the SGM-based guarantee may be unsound for the actual mechanism, because privacy guarantees depend directly on the output distribution of the mechanism, whereas this mismatch changes the distribution. 
Previous work~\cite{lebeda2025avoiding} identifies this issue for EASGM-based mechanisms, but does not formally analyze the real privacy guarantees.

\noindent\textbf{Privacy auditing of DP-SGD. }
Theoretical bounds assume idealized, bug-free implementations, which may not hold in practice due to implementation errors~\cite{cebere2026privacy}. Thus, the privacy leakage of implementations could deviate from theoretical expectations~\cite{nasr2023tight}. Numerous studies~\cite{nasr2021adversary,nasr2023tight,xiang2025tight} have attempted to audit whether DP-SGD implementations violate theoretical privacy guarantees. 
Following the DP security game, they instantiate the adversary as a distinguisher~$\phi$, repeatedly execute the DP-SGD implementations with two neighboring datasets, and use $\phi$ to infer which dataset was used based on obtained model updates. 
Finally, they compare the empirical privacy leakage, measured by FNR at a fixed FPR, with the corresponding theoretical value.
If the empirical FNR is above the FNR at the same fixed FPR implied by the privacy guarantee, no privacy violation is considered to have been found in the implementation.
Many studies have found no violation in DP-SGD implementations in practice \cite{nasr2023tight,koskela2025auditing,nasrlast,steinke2023privacy,xiang2025tight}. However, prior auditing works, including~\cite{koskela2025auditing,steinke2023privacy,nasr2023tight,nasrlast,xiang2025tight}, have primarily focused on auditing per-round SGM-based DP-SGD, without considering the normalization variants used in practice. For instance, Steinke et al.~\cite{steinke2023privacy} audit the DP-SGD implementation in Algorithm~2 of \cite{steinke2023privacy_auditing_one_run_arxiv}, and Xiang et al.~\cite{xiang2025tight} follow the same setting. Nasr et al.~\cite{nasr2023tight} audit the implementation in Algorithm~2 of~\cite{nasr2023tight}, and Koskela et al.~\cite{koskela2025auditing} follow the setting in Algorithm~4 of \cite{koskela2024auditing_density_arxiv}.  Consequently, these works do not audit the privacy leakage of implementations of EASGM- and ASGM-based DP-SGD algorithms, leaving it unclear whether such implementations truly satisfy the theoretical SGM privacy guarantee.

\noindent \textbf{Mismatched formalization of other DP mechanisms.}
Many other state-of-the-art differentially private mechanisms are built on EASGM-based or ASGM-based DP-SGD, including DP\textendash FETA~\cite{11023490}, PrivImage~\cite{299531}, dp\textendash promise~\cite{wang2024dp}, DP\textendash SUR~\cite{fu2024dpsur}, DISK~\cite{zhangdisk}, and TrainDPGAN~\cite{bie2023private}. These mechanisms also rely on SGM-based privacy analysis to derive their privacy guarantees. Consequently, the privacy analysis of these mechanisms suffers from the same mismatch as that of normalized DP-SGD. 

The theoretical privacy analysis guides the configuration of model training in reality to achieve the required privacy guarantee. 
However, these mismatches may underestimate the actual privacy leakage, resulting in a training procedure failing to satisfy the claimed privacy guarantee.
Thus, it is necessary to re-examine the privacy guarantees used for EASGM-based and ASGM-based DP\text{-}SGD and their derivatives in the correct game setting.  

\noindent \textbf{Our contributions.} 
In this paper, we revisit the security analysis of per-round EASGM- and ASGM-based DP\text{-}SGD by explicitly incorporating their normalization steps into the privacy game, and derive corresponding privacy guarantees. Our analysis is conducted in the \(f\)-DP framework, a de facto standard framework for differential privacy analysis \cite{dong2019gaussian,nasr2023tight,xiang2025tight}, which characterizes privacy through a trade-off function between the false positive rate (FPR) and false negative rate (FNR). We compare the derived guarantees with the one-round SGM guarantee established in~\cite{bu2020deep}. To the best of our knowledge, only Bu et al.~\cite{bu2020deep} provide a tight one-round privacy guarantee in the \(f\)-DP framework and an explicit trade-off function, making their work a natural baseline for comparison.
 
Our analysis reveals that  
\textbf{the privacy guarantees of EASGM and ASGM become weaker as the output dimension increases, and they are generally weaker than that of SGM when the output dimension exceeds certain values.}
For instance, our analysis shows that, in some cases, under the same FNR of \(0.5\), the FPRs of EASGM and ASGM fall below \(0.2\), whereas the previously claimed value is above \(0.4\)~(Fig. \ref{fig:Theoretical}). 
Furthermore, when the sampling ratio is~1, our analysis is tight: the FNR--FPR trade-off is achievable by an adversary without any amplification or relaxation introduced by analytical approximation. 
This finding suggests that the privacy leakage of EASGM- and ASGM-based DP\text{-}SGD implementations can exceed the guarantee claimed under SGM. 

We also audit the privacy of EASGM- and ASGM-based implementations to assess whether their privacy leakage can exceed the SGM guarantee with both synthetic and real-world datasets. We perform tight auditing to detect maximum privacy leakage. However, because tight auditing is computationally prohibitive on large datasets when the output dimension matches real-world scenarios, our auditing focuses on small datasets. 
We observe that the audited privacy leakage increases with the output dimension and can exceed the SGM guarantee. 


We also conduct a code audit of \texttt{privacy\_engine.py} across Opacus versions v0.9.0 to v1.5.4 and identify an additional DP\text{-}SGD variant in recent v1.$\ast$ releases that normalizes the noisy gradient by $\lfloor Nq \rfloor$. We denote this variant by FEASGM (Definition~\ref{def:FEASGM}). Although a previous Opacus bug report suggests that, for large datasets, the privacy guarantee of FEASGM is tightly bounded by that of SGM (Issue~\#571~\cite{opacus_issue_571}), our analysis shows that FEASGM can still provide a significantly weaker privacy guarantee than SGM even in the large-dataset regime (Appendix~\ref{sec:visualizing_FEASGM}).

To the best of our knowledge, this is the first work to provide a systematic per-round privacy analysis of EASGM- and ASGM-based DP-SGD. We summarize our contributions as follows:

\noindent \textbf{Re-analyzing privacy guarantee of DP-SGD.}
We conduct a comprehensive analysis of the per-round privacy guarantees of EASGM- and ASGM-based DP\text{-}SGD using tight analysis tools and derive a bound that can be used to estimate the actual strength of these guarantees. Moreover, we visualize this bound using a Central Limit Theorem (CLT)-based approximation, providing an intuitive understanding of their privacy behavior. Our analysis shows that the privacy guarantees of EASGM and ASGM can become substantially weaker than the corresponding SGM-based guarantees, especially in high-dimensional settings. 

\noindent \textbf{Auditing the privacy leakage of DP\text{-}SGD implementations.}
We audit the per-round DP\text{-}SGD procedures used in four state-of-the-art differentially private machine learning methods—DP-SUR~\cite{fu2024dpsur}, DPSGD-HF~\cite{tramer2020differentially}, DP-FETA~\cite{11023490}, and PrivImage~\cite{299531}. DPSGD-HF, DP-FETA, and PrivImage rely on Meta’s Opacus library for their DP-SGD implementations. For these three methods, we audit the implementations using Opacus versions consistent with their algorithmic descriptions. Empirical evaluations on widely used public datasets—CIFAR-10, MNIST, and FMNIST—show that, under the audited settings, all four implementations exhibit privacy leakage exceeding the SGM guarantee established in~\cite{bu2020deep}.

\noindent \textbf{Thorough code auditing of the Opacus library.} We examine the evolution of the \texttt{privacy\_engine.py} file in the Opacus library from v0.9.0 to v1.5.4 and find that version v0.13.0 use ASGM, versions v0.14.0--v1.4.\(\ast\) use EASGM, and versions v1.5.\(\ast\) use FEASGM (Section~\ref{sec:Opacus}). We further analyze the privacy guarantee of multi-round FEASGM-based DP-SGD and show that FEASGM can yield a weaker privacy guarantee than SGM (Appendix~\ref{sec:visualizing_FEASGM}).

\theoremstyle{definition}

\section{Preliminary}

\begin{table}[t]
\caption{Notation used in the analysis.}
\centering
\scriptsize
\begin{tabularx}{\columnwidth}{@{}lX@{}}
\toprule
\textbf{Symbol} & \textbf{Description} \\
\midrule
$D, D'$ & Neighboring datasets, $D=\{d_1,\dots,d_N\}$, $D'=D\cup d_{N+1}$ \\
$D^{-}$ & Shared subset of neighboring datasets, $D^{-}=D\cap D'$ \\
$P$ & Standard normal distribution $\mathcal{N}(0,1)$ \\
$B$ & Sampled batch from $D$, $B\subseteq D$ \\
$|S|$ & Number of elements in set $S$ \\
$Q_B$ & $\mathcal{N}(0,(|B|/(|B|+1))^2)$ \\
$Q_B^{\mu}$ & $\mathcal{N}(\mu,(|B|/(|B|+1))^2)$ \\
$Q_x$ & For any positive integer $x$, 
$\mathcal{N}\!\left(0,\left(\frac{x}{x+1}\right)^2\right)$ \\
$Q_x^{\mu}$ & For any positive integer $x$ and any $\mu \in \mathbb{R}$, 
$\mathcal{N}\!\left(\mu,\left(\frac{x}{x+1}\right)^2\right)$ \\
$q$ & Sampling ratio \\
$n$ & Dimension of the output distribution \\
$\mathrm{ID}$ & Identical trade-off function (Definition~\ref{def:identical_trade_off}) \\
$G_{\mu}$ & Gaussian trade-off function (Definition~\ref{def:Gaussian_DP})\\
 $T(P_{w},Q_{w})$ & Trade-off function in Lemma~\ref{Lemma:Jensen} \\
$T(P_{I}\mid I,Q_{I}\mid I)$ & Lower bound (privacy guarantee) of $T(P_{w},Q_{w})$ in Lemma~\ref{Lemma:Jensen} \\
\bottomrule
\end{tabularx}
\label{tab:notation}
\end{table}
\subsection{Differential Privacy}\label{sec:differential privacy}

\begin{definition}
    \textbf{Differential Privacy~\cite{fu2024dpsur}.} 
\textit{
A randomized algorithm $\mathcal{M}$ satisfies $(\epsilon, \delta)$-differential privacy if, for any two adjacent datasets $D$ and $D'$, and for all possible output sets $O$,
$\Pr[\mathcal{M}(D) \in O] \leq e^{\epsilon} \Pr[\mathcal{M}(D') \in O] + \delta$, where $\epsilon > 0$ and $\delta \in [0,1]$ are privacy parameters. 
Here, $\epsilon$ and $\delta$ (collectively referred to as the privacy budget) quantify the level of privacy loss: 
smaller values of either parameter imply stronger privacy guarantees. 
Datasets $D$ and $D'$ are said to be neighboring if they differ by at most one individual record, i.e., $D = D' \cup \{d\}$ or $D' = D \cup \{d\}$.
}

\noindent  \textbf{A functional view of DP: $f$-DP.}  To interpret the privacy guarantee of a mechanism~\(\mathcal{M}\), 
prior work uses the security game illustrated in Fig~\ref{fig:privacy-game-and-mech} \cite{dong2019gaussian,annamalai2024shuffle}.

\end{definition}

The adversary's guess in the DP game can be formalized as the following hypothesis test,
$\mathbf{H_0} : O \sim \mathcal{M}(D)$, and $\quad \mathbf{H_1} : O \sim \mathcal{M}(D')$. 
Consider a rejection rule $\mathcal{R}: \mathcal{Y} \rightarrow \{0,1\}$ in above hypothesis testing setup. Two key types of errors are: \textbf{Type I error} (false positive rate): \(\alpha = \Pr[\mathcal{R}(y) = 1 \mid \mathbf{H_0}]\), the chance of incorrectly rejecting \(\mathbf{H_0}\) when it is actually true; \textbf{Type II error} (false negative rate): \(\beta = \Pr[\mathcal{R}(y) = 0 \mid \mathbf{H_1}]\), the chance of incorrectly accepting \(\mathbf{H_0}\) when \(\mathbf{H_1}\) is true.

The adversary aims to reduce one type of error while fixing the other, by determining the smallest achievable $\beta$ for a given $\alpha$, leading to the following definition. 


\begin{definition}\textbf{Trade-off function~\cite{dong2019gaussian}.}\label{def:trade-off}
\textit{Given two distributions \( P \) and \( Q \), the trade-off function is defined as:}
$T(P,Q)(\alpha) = \inf_{\mathcal{R}} \left\{ \beta_{\mathcal{R}} : \alpha_{\mathcal{R}} \leq \alpha \right\}$.
\textit{Let \(\mathcal{R}\) be a rejection rule for the hypothesis testing problem \((P, Q)\), which aims to determine whether a given sample \(\mathcal{O}\) is drawn from \(P\) or \(Q\). That is, $H_0: \mathcal{O} \sim P, \quad H_1: \mathcal{O} \sim Q$.
The infimum is taken over all such rejection rules \(\mathcal{R}\).}
\end{definition}


In the DP game, the trade\textendash off function \(f=T(P,Q)\) for neighboring outputs \(P=\mathcal{M}(D)\) and \(Q=\mathcal{M}(D')\) provides the tightest characterization of the mechanism’s privacy leakage for the pair \((D,D')\). 
It determines the adversary’s minimum Type~II error under a fixed Type~I error \(\alpha\) \cite{kaissis2024beyond}: 
\(f(\alpha)\). Hence, the trade\textendash off function directly characterizes the pairwise privacy leakage of \(\mathcal{M}\); 
In particular, if \(T(P, Q) \ge T(P', Q')\), then the pairwise privacy leakage
corresponding to \((P, Q)\) is smaller than that corresponding to \((P', Q')\).

 Based on the trade\textendash off function~\(f\), an \(f\)\nobreakdash-DP mechanism was defined in prior work (Definition~\ref{def:f_DP}), which guarantees that, for any pair of neighboring datasets \(D\) and \(D'\), the privacy leakage induced by the $f$-DP mechanism is at most \(f\).

\begin{definition}\label{def:f_DP}\textbf{$f$-differential privacy \cite{dong2019gaussian}.} \label{def:f_DP} Let $f$ be a trade-off function. A mechanism $\mathcal{M}$ is said to be
$f$-differentially private if
   $ T\big(\mathcal{M}(D),\,\mathcal{M}(D')\big) \ge f$
for all neighbouring datasets $D$ and $D'$.
\end{definition}


Following~\cite{dong2019gaussian}, $(\epsilon,\delta)$-DP can be expressed through the trade-off function 
\begin{equation}\label{eq:trade-off_epsilon_delta}
  f_{\epsilon,\delta}(\alpha)=\max\left\{0,\;1-\delta-e^{\epsilon}\alpha,\;e^{-\epsilon}(1-\delta-\alpha)\right\}
\end{equation}

A more commonly used form of the trade-off function 
$f$ is defined based on two normal distributions, as follows:

\begin{definition}\textbf{$\mu$-Gaussian DP ($\mu$-GDP)~\cite{dong2019gaussian}.}\label{def:Gaussian_DP}
The trade-off function for distinguishing between $\mathcal{N}(0,1)$ and $\mathcal{N}(\mu,1)$ is defined as $G_\mu(x) = T(\mathcal{N}(0,1), \mathcal{N}(\mu,1))(x) = \Phi\left(\Phi^{-1}(1 - x) - \mu\right)$,
where $\Phi$ denotes the cumulative distribution function (c.d.f.) of the standard normal distribution. A mechanism $\mathcal{M}$ is said to satisfy $\mu$-GDP if it is $G_\mu$-DP.
\end{definition}

In the \(f\)\nobreakdash-DP framework, there exists a special trade--off function 
\(\mathrm{ID}\), defined by two identical distributions, which corresponds to no 
privacy leakage. We overload $\otimes$ to denote both product distributions and tensoring of trade-off functions; in particular, if $f=T(P,Q)$, then $f^{\otimes n}=T(P^{\otimes n},Q^{\otimes n})$.

\begin{definition}\label{def:identical_trade_off}
\textbf{Identical trade-off function \cite{dong2019gaussian}.} 
For any distribution $P$, let $\mathrm{ID} = T(P, P)$, where $\mathrm{ID}(\alpha) = 1 - \alpha$. 
The function $\mathrm{ID}$ is referred to as the identical trade-off function, 
which represents the case of no privacy leakage. 
Moreover, for any trade-off function $f$, it holds that 
$\mathrm{ID} \otimes f = f$ and $\mathrm{ID}\ge f$.
\end{definition}

Here, we present several algebraic properties of the trade-off function, which serve as the foundation for our privacy analysis.

\newtheorem{property}{Property}
\begin{property}\cite{dong2019gaussian}\label{property:expression_composition}
    Let $P$, $P'$, $Q$, and $Q'$ be independent distributions. Then
    $T(P \otimes Q, P' \otimes Q') = T(P, P') \otimes T(Q, Q').$
    
\end{property}

\begin{property}\cite{dong2019gaussian}\label{property:order}
    Let $P$, $P'$, $Q$, $Q'$, $P_1$, and $P_1'$ be independent distributions. Suppose
    $T(P, P') = f_1, \quad T(Q, Q') = f_2, \quad T(P_1, P_1') = f_3$.
    If $f_2 \leq f_3$, then
    $f_1 \otimes f_2 \leq f_1 \otimes f_3$.
\end{property}

\begin{property}\label{property:useful}
Let \(P\), \(Q\), and \(Q'\) be probability distributions. Suppose
\(T(P,Q)=f\) and \(T(P,Q')=f'\). For \(n\ge 1\), let
\(f^{\otimes n}:=T(P^{\otimes n},Q^{\otimes n})\). Then
$f^{\otimes n} \ge f^{\otimes n}\otimes f'$.
Indeed, since \(\mathrm{ID}\ge f'\), Property~\ref{property:order} gives
$f^{\otimes n}
= f^{\otimes n}\otimes \mathrm{ID}
\ge f^{\otimes n}\otimes f'$.
\end{property}

Finally, we present an established result in \(f\)\nobreakdash-DP~\cite{dong2019gaussian}, which is used in our analysis.

\begin{lemma}[Proposition~D.2 in~\cite{dong2019gaussian}]
For any trade-off function \(f\) that is not \(\operatorname{ID}\) (Definition \ref{def:identical_trade_off}),
\[
\lim_{n \to \infty} f^{\otimes n}(\alpha) = 0,
\qquad \forall\, \alpha \in (0,1].
\]
In fact, the convergence is exponentially fast.
\label{Lemma:zero}
\end{lemma}

\subsection{ Differentially Private Mechanisms in Machine Learning}
The per-round update of DP-SGD can be categorized into three implementation forms in the literature: the Subsampled Gaussian Mechanism (SGM)~(Definition~\ref{def:SGM}) \cite{chua2024private}, the Expected-Averaged SGM (EASGM)\cite{abadi2016deep}~(Definition~\ref{def:EASGM}), and the Averaged SGM (ASGM)~(Definition~\ref{def:ASGM})\cite{bu2020deep}.

\begin{definition}\label{def:SGM}
\textbf{Subsampled Gaussian Mechanism (SGM) ~\cite{chua2024private}.} 
Given a dataset \( D = \{d_{1}, d_{2}, \ldots, d_{N}\} \) containing \( N \) data points, the SGM proceeds in three main steps: 1. Independently sample each data point with probability \( q \leq 1 \) to form a minibatch \( B \); 2. Compute an information vector \( g_i \) for each sampled data point \( d_i \in B \), and clip its norm to be at most \( C\), i.e., $\bar{g}_i = g_i/\max(1, \frac{\|g_i\|}{C})$;
3. Aggregate the clipped vectors and add isotropic Gaussian noise:
$\sum_{i \in B} \bar{g}_i \;+\; G$, $G \sim \mathcal{N}(0,\sigma^2 C^2 \mathbb{I})$.

\end{definition}

\begin{definition}\label{def:EASGM}
    \textbf{Expected-Averaged Subsampled Gaussian Mechanism (EASGM) \cite{abadi2016deep}.} Following the same setting as Definition~\ref{def:SGM}, the output of EASGM can be formalized as
$\frac{\sum_{i \in B} \bar{g}_i + \mathcal{N}(0, \sigma^2 C^2 \mathbb{I})}{N q}$,
where $N$ denotes the total number of data points in the dataset $D$.
 
\end{definition}

\begin{definition}\label{def:ASGM}
\textbf{Batch-Averaged Subsampled Gaussian Mechanism (ASGM) \cite{bu2020deep}\footnote{A full and explicit description of ASGM is provided as Algorithm~1 in the full paper \cite{bu2020deep_fullpaper}.}.} Following the same setting as Definition~\ref{def:SGM}, the output of ASGM can be formalized as
$\frac{\sum_{i \in B} \bar{g}_i + \mathcal{N}(0, \sigma^2 C^2 \mathbb{I})}{|B|}$,
where $|B|$ denotes the number of data points in the sampled batch $B$.
  
\end{definition}


\subsection{Privacy Analysis}\label{sec:analysis}
\noindent 
In the security analysis, it is desirable to derive a tight privacy guarantee, as a tighter bound enables achieving the same privacy level with less noise,  which in turn leads to higher output quality of the mechanism. To derive the privacy guarantee of \(\mathcal{M}\), one must compute the 
trade--off function for each pair of neighboring datasets, identify the 
tightest lower bound among these functions, and adopt this bound as the 
privacy guarantee. The trade-off function between two distributions follows from the Neyman-Pearson lemma, which states that the likelihood-ratio test is uniformly most powerful. By thresholding the likelihood ratio, the optimal Type~II error $\beta$ corresponding to each Type~I error $\alpha$ can be obtained~\cite{dong2019gaussian}.

However, when the output distribution of a mechanism is highly complex, directly deriving the trade-off function via likelihood-ratio thresholding becomes computationally infeasible. In such cases, performing a tight privacy analysis is no longer tractable. Instead, one can only establish a lower bound of the trade-off function and derive a corresponding nearly tight privacy guarantee from these lower bounds.

Prior work derives a {tight and computable lower bound} on the trade\textendash off function under a stronger adversary 
by granting the adversary additional information beyond the mechanism’s output~\cite{wang2023unified}. Let  
\(P_{w} = \sum_{i} w_{i} P_{i}\) and \(Q_{w} = \sum_{i} w_{i} Q_{i}\) 
be mixture distributions generated by first sampling an index \(I \sim w\)
(i.e., choosing index \(i\) with probability \(w_{i}\)) 
and then drawing a sample from \(P_{I}\) or \(Q_{I}\).
Here, \(w = (w_{1}, w_{2}, \ldots, w_{m})\) is a discrete probability distribution 
such that \(w_{i} \ge 0\) and \(\sum_{i=1}^{m} w_{i} = 1\).

\begin{lemma}\label{Lemma:Jensen}
    \textbf{Joint concavity of trade-off functions \cite{wang2023unified}}. Let $\{P_i\}_{i=1}^{m}$ and $\{Q_i\}_{i=1}^{m}$ be two sequences of probability distributions. Denote the probability density functions (pdfs) of $P_i$ and $Q_i$ as $p_i$ and $q_i$, respectively. Consider the mixture distributions $P_w$ and $Q_w$ with pdfs $p_w = \sum_{i=1}^{m} w_i p_i \quad \text{and} \quad q_w = \sum_{i=1}^{m} w_i q_i$,
where the weight $\mathbf{w} = (w_1, \cdots, w_m)$ is such that $w_i \geq 0$ and $\sum_{i=1}^{m} w_i = 1$. \textit{For two mixture distributions $P_w$ and $Q_w$, it holds}
    $T(P_w, Q_w)(\alpha(t, c)) \geq  T((P_{I}|I,I),(Q_{I}|I,I))(\alpha(t, c))= \sum_{i=1}^{m} w_i T(P_i, Q_i)(\alpha_i(t, c))$,
where
$\alpha_i(t, c) = \mathbb{P}_{X\sim P_i} \left[ \frac{q_i}{p_i}(X) > t \right] + c\mathbb{P}_{X\sim P_i} \left[ \frac{q_i}{p_i}(X) = t \right]$
is the type I error for testing $P_i$ v.s. $Q_i$ using the likelihood ratio test, and
$\alpha(t, c) = \sum_{i=1}^{m} w_i \alpha_i(t, c)$. 

\end{lemma}

Lemma~\ref{Lemma:Jensen} formalizes the construction of a lower bound on the trade\textendash off function \(T(P_{w},Q_{w})\), 
where revealing the auxiliary variable \(I\) to the adversary induces the conditional output  
distributions \((P_{I}\mid I,I)\) and \((Q_{I}\mid I,I)\), for which the trade\textendash off  
function can be expressed in terms of the component functions \(T(P_{i},Q_{i})\).  
By the post\textendash processing property of differential privacy~\cite{wang2023unified}, revealing \(I\) increases the adversary’s distinguishing power; hence,
$T(P_{w},Q_{w}) \;\ge\; T\bigl((P_{I}\mid I,I),\,(Q_{I}\mid I,I)\bigr)$,
as stated in Lemma~\ref{Lemma:Jensen}.  This result provides a computable and tight lower bound for \(T(P_{w},Q_{w})\) in terms of the component trade-off functions \(\{T(P_{i},Q_{i})\}\).

\subsection{Privacy Auditing}\label{sec:privacy_auditing}

Privacy auditing aims to experimentally assess whether a mechanism violates a claimed \(f\)\nobreakdash-DP guarantee~\cite{xiang2025privacy,annamalai2025hitchhiker}. 
It proceeds by constructing a distinguisher and a pair of neighboring datasets \((D,D')\), and then estimating the empirical Type~I and Type~II errors from the mechanism's outputs. 
The auditing process follows the DP security game illustrated in Fig.~\ref{fig:privacy-game-and-mech}. 
Specifically, the distinguisher is repeatedly applied to outputs generated from the neighboring datasets, yielding empirical estimates of the Type~I and Type~II errors, denoted by \(\overline{\alpha}\) and \(\overline{\beta}\), respectively.

Based on the empirical results, confidence intervals for the false positive rate 
\(\alpha \in (\alpha_l,\alpha_r)\) and the false negative rate 
\(\beta \in (\beta_l,\beta_r)\) at confidence level~\(\gamma\) are derived from the observed values \((\overline{\alpha},\overline{\beta})\) and the total number of auditing trials~\(t\). 
This is typically achieved using the Clopper--Pearson method~\cite{clopper1934use,nasr2023tight}, which models \(\alpha\) and \(\beta\) as the success probabilities of two binomial distributions.

Following prior work~\cite{nasr2023tight,xiang2025privacy}, we compare the empirical auditing result with the claimed privacy guarantee in the \((\varepsilon,\delta)\)\nobreakdash-DP view. 
For the claimed trade-off function~\(f\), we first derive the tightest \((\varepsilon,\delta)\)\nobreakdash-DP guarantee it implies by fixing \(\varepsilon\) and finding the smallest \(\delta\) such that \(f \ge f_{\varepsilon,\delta}\), where \(f_{\varepsilon,\delta}\) is given in Eq.~\ref{eq:trade-off_epsilon_delta}.

We then translate the auditing results into empirical lower bounds on the corresponding \((\varepsilon,\delta)\)\nobreakdash-DP guarantee. 
For a fixed~\(\delta\), the empirical lower bound on \(\varepsilon\) is
\begin{equation}
\label{eq:audited_epsilon}
\varepsilon_L = \max\!\left\{
\log\frac{1-\delta-\alpha_r}{\beta_r},\,
\log\frac{1-\delta-\beta_r}{\alpha_r},\,
0
\right\}.
\end{equation}
Alternatively, for a fixed~\(\varepsilon\), the empirical lower bound on \(\delta\) is
\begin{equation}
\label{eq:audited_delta}
\delta_L = \max\!\left\{
0,\,
1-\alpha_r-e^{\varepsilon}\beta_r,\,
1-\beta_r-e^{\varepsilon}\alpha_r
\right\}.
\end{equation}

To determine whether a violation occurs, these empirical lower bounds are compared with the claimed guarantees. 
If~\(\varepsilon_L > \varepsilon\) or~\(\delta_L > \delta\), then the neighboring pair~\((D,D')\) provides empirical evidence that the mechanism violates its claimed privacy guarantee.
\section{Privacy Analysis of EASGM and ASGM}

In this section, we analyze the privacy guarantees of EASGM and ASGM, 
focusing on neighboring dataset pairs where both datasets have sizes 
greater than~\(1\).


\noindent \textbf{Analytical setting.} For neighboring datasets \(D\) and \(D'\), we derive the privacy guarantee only for pairs \((D, D')\) with a fixed ordering, i.e., not for \((D', D)\). 
This is because the trade-off function for $(D', D)$ can be obtained from that for $(D, D')$~\cite{dong2019gaussian}: interchanging the output distributions 
\(P = \mathcal{M}(D)\) and \(Q = \mathcal{M}(D')\) yields the inverse function \(f^{-1} = T(Q, P)\), 
which is the reflection of \(f = T(P, Q)\) across the line \(y = x\). 

To avoid division-by-zero in the per-round normalization, we adopt the following convention: 
(i) in EASGM, if the nominal normalization factor \(Nq = 0\) (e.g., \(q = 0\) or \(N = 0\)), we set the factor to \(1\); 
(ii) in ASGM, if the mini-batch size \(|B| = 0\), we also set the normalization factor to \(1\). 
This convention ensures numerical stability without altering the mechanism’s behavior. 
When no data are sampled (\(q = 0\) or \(|B| = 0\)), the aggregated gradient becomes a zero vector, and no privacy-relevant computation occurs. 

We first fix the notation used throughout our analysis.  
Let the neighboring datasets be \(D = \{d_{1}, \ldots, d_{N}\}\) and \(D' = D \cup \{d_{N+1}\}\), and write \(N = |D|\).  
All symbols are summarized in Table~\ref{tab:notation}. The following theoretical tools are used in our analysis.

\noindent\textbf{Use of the theoretical results.}
These results are used to simplify and compare trade-off functions, and to derive a computationally efficient upper bound for visualizing the resulting guarantees.
Lemma~\ref{Lemma:min_equal} provides the foundation for the remaining results.
Building on Lemma~\ref{Lemma:min_equal}, Lemma~\ref{Lemma:Order_in_trade-off} establishes
$T(\overline{P}_I \mid I, \overline{Q}_I \mid I)$ as a lower bound for
$T(P_I \mid I, Q_I \mid I)$, while Lemma~\ref{Lemma:Upper_bound}
derives an analytic upper bound for
$T(P_I \mid I, Q_I \mid I)$.
This upper bound is computationally efficient because it only requires evaluating each component $T(P_i,Q_i)(\alpha)$,
rather than the likelihood-based expression
$T(P_i,Q_i)(\alpha(t,c))$.
Finally, Theorems~\ref{thm:bound} and~\ref{thm:bound_fix_variance}
are used to compare the trade-off functions induced by Gaussian distributions with different variances.

\begin{lemma}\label{Lemma:min_equal}
Given mixture distributions
$P_{w}=\sum_{i=1}^{m}w_{i}P_{i}$,
$Q_{w}=\sum_{i=1}^{m}w_{i}Q_{i}$,
let
$E_{\alpha}
=
\left\{
(\alpha_{1},\ldots,\alpha_{m})\in[0,1]^m
\;\middle|\;
\sum_{i=1}^{m}w_{i}\alpha_{i}=\alpha
\right\}$.
Then, for any \(\alpha\in[0,1]\),
$T((P_I\mid I,I), (Q_I\mid I,I))(\alpha)
=
\min_{(\alpha_{1},\ldots,\alpha_{m})\in E_{\alpha}}
\sum_{i=1}^{m}w_{i}\,T(P_{i},Q_{i})(\alpha_{i})$.
The proof is given in Appendix~\ref{Proof_Lemma_mim_equal}.
\end{lemma}

\begin{lemma}\label{Lemma:Order_in_trade-off}
For a fixed \(y\in[m]\), suppose that \(T(\overline{P}_{y},\overline{Q}_{y})(\alpha)\le T(P_{y},Q_{y})(\alpha)\) for all \(\alpha\in[0,1]\). Let \(\overline{Q}_{w}=\sum_{i\in[m]\setminus\{y\}}w_{i}Q_{i}+w_{y}\overline{Q}_{y}\) and \(\overline{P}_{w}=\sum_{i\in[m]\setminus\{y\}}w_{i}P_{i}+w_{y}\overline{P}_{y}\). Then, for any \(\alpha\in[0,1]\), \(T((P_{I}\mid I,I),(Q_{I}\mid I,I))(\alpha)\ge T((\overline{P}_{I}\mid I,I),(\overline{Q}_{I}\mid I,I))(\alpha)\). The proof is given in Appendix~\ref{Proof:Lemma_Order_in_trade-off}.
\end{lemma}

\begin{lemma}\label{Lemma:Upper_bound}
Let \(\mathbf{w}=(w_1,w_2,\ldots,w_m)\), where \(w_i\ge0\) and \(\sum_{i=1}^{m}w_i=1\), and let \(I\) be the random variable such that \(\mathbb{P}[I=i]=w_i\). Then, for any \(\alpha\in[0,1]\), \(T((P_I\mid I,I),(Q_I\mid I,I))(\alpha)\le\sum_{i=1}^{m}w_i\,T(P_i,Q_i)(\alpha)\). The proof is given in Appendix~\ref{Proof:Lemma_Upper_bound}.
\end{lemma}

\begin{theorem}\label{thm:bound}
Let \(0<\sigma<1\). For any \(\mu_{1}\ge \mu_{2}\ge 0\) and any
\(\alpha\in[0,1]\), it holds that
$T\!\bigl(\mathcal{N}(0,1), \mathcal{N}(\mu_{1},\sigma^{2})\bigr)(\alpha)
\le
T\!\bigl(\mathcal{N}(0,1), \mathcal{N}(\mu_{2},\sigma^{2})\bigr)(\alpha)$.
The proof is given in Appendix~\ref{App_proof:Theorem_bound}.
\end{theorem}

\begin{theorem}\label{thm:bound_fix_variance}
Let \(0<\sigma_{2}\le \sigma_{1}\le 1\). Then, for any
\(\alpha\in[0,1]\), it holds that
$T\!\bigl(\mathcal{N}(0,1), \mathcal{N}(0,\sigma_{2}^{2})\bigr)(\alpha)
\le
T\!\bigl(\mathcal{N}(0,1), \mathcal{N}(0,\sigma_{1}^{2})\bigr)(\alpha)$.
The proof is given in Appendix~\ref{App_proof:bound_fix_variance}.
\end{theorem}
\subsection{\textbf{Analysis Roadmap}}

We now provide a high-level overview of the analysis for EASGM and ASGM.
Our goal is to establish valid privacy guarantees for these mechanisms and to
understand how these guarantees compare with the closed-form privacy guarantee of SGM.

\noindent \textbf{Step 1: Characterize the privacy leakage.}
For each mechanism, we first characterize the output distributions induced by an
ordered neighboring pair $(D,D')$ and rewrite the corresponding trade-off function
in mixture form. For EASGM, this takes the form
$T(P_w, Q_w)$.
The ASGM case is handled analogously.

\noindent \textbf{Step 2: Derive valid privacy guarantees.}
Using this mixture representation, we construct lower bounds of the corresponding
trade-off functions to derive valid privacy guarantees. In particular, we first
obtain a valid per-pair privacy guarantee, denoted by
$T(P_I \mid I, Q_I \mid I)$, and then further derive a valid guarantee that applies
to all neighboring pairs with $|D|=N$ and $|D'|=N+1$, namely
$T(\overline{P}_I \mid I, \overline{Q}_I \mid I)$.
The privacy guarantee for ASGM is derived in the same manner.

\noindent \textbf{Step 3: Derive tractable upper bounds and compare with SGM.}
A direct comparison between the derived guarantees of EASGM and ASGM and the
closed-form privacy guarantee of SGM is challenging, because the guarantees of
EASGM and ASGM do not admit simple closed-form expressions.
We therefore derive analytic upper bounds $\overline{f}$ for the corresponding
per-pair guarantees such that
$\overline{f}\ge T(P_I \mid I, Q_I \mid I)$.
These tractable conservative bounds allow us to visualize the privacy guarantees and compare them with the SGM guarantee.

For both EASGM and ASGM, the resulting upper bound $\overline{f}$ contains a term of the form $T(P^{\otimes (n-1)},Q^{\otimes (n-1)})$, where $P$ and $Q$ are two different Gaussian distributions. As the output dimension $n$ increases, this term approaches zero (Lemma~\ref{Lemma:zero}), and hence $\overline{f}$ decreases, whereas the SGM guarantee does not depend on the output dimension. In other words, unlike SGM, the privacy guarantees of EASGM and ASGM degrade with the output dimension. This asymptotic behavior allows us to identify regimes in which the guarantees of EASGM and ASGM can become weaker than the SGM guarantee.
\subsection{Privacy Guarantee of EASGM } 
\label{sec:guarantee_EASGM}

Let \(\mathcal{M}\) denote the EASGM mechanism (Definition~\ref{def:EASGM}).
The first step in establishing the trade-off function is to formally characterize the output distributions of \(\mathcal{M}(D)\) and \(\mathcal{M}(D')\) \footnote{By a slight abuse of notation, we use random variables and their induced distributions interchangeably when the meaning is clear from context.}, as detailed in Section~\ref{sec:analysis}.
We can express the distribution of $\mathcal{M}(D)$ as
$\sum_{B} P(B) \left( \frac{\sum_{i \in B} \overline{g}_i + G}{N \cdot q} \right)$, where \(G \sim \mathcal{N}(0,\sigma^{2}C^{2} \mathbb{I})\) is \(n\)-dimensional isotropic Gaussian noise,
$P(B) = q^{|B|}(1-q)^{N - |B|}$ denotes the probability of sampling batch $B$,
and $\overline{g}_i$ denotes the clipped gradient information computed from the data point $d_i$.
Similarly, the distribution of $\mathcal{M}(D')$ can be written as
$ (1-q) \sum_{B} P(B) \left( \frac{\sum_{i \in B} \overline{g}_i + G}{(N+1)\cdot q} \right)
+ q \sum_{B} P(B) \left( \frac{\sum_{i \in B} \overline{g}_i + \overline{g}_{N+1} + G}{(N+1)\cdot q} \right)$.

\label{sec:Privacy_high_EASGM}

Directly computing the trade-off function is computationally infeasible, as discussed in Section~\ref{sec:analysis}. Therefore, leveraging Lemma~\ref{Lemma:Jensen}, we establish a lower bound and obtain the following result (Eq. \ref{eq:lower_bound_EASGM_high_dimentsion}), where
\(f_{B}^{1} = T\!\left(\tfrac{\sum_{i\in B}\overline{g}_{i} + G}{Nq}, \tfrac{\sum_{i\in B}\overline{g}_{i} + G}{(N+1)q}\right)\)
and
\(f_{B}^{2} = T\!\left(\tfrac{\sum_{i\in B}\overline{g}_{i} + G}{Nq}, \tfrac{\sum_{i\in B}\overline{g}_{i} + \overline{g}_{N+1} + G}{(N+1)q}\right)\).

\begin{equation}\label{eq:lower_bound_EASGM_high_dimentsion}
    \begin{aligned}
        &T(\mathcal{M}(D),\mathcal{M}(D'))(\alpha(t,c))
\;\ge\;
f(\alpha(t,c))\\&=
\sum_{B} P(B)((1-q)f_{B}^{1}(\alpha_{B}^{1}(t,c))
          +q\,f_{B}^{2}(\alpha_{B}^{2}(t,c))),
    \end{aligned}
\end{equation}

To simplify the privacy analysis, we apply distinct reversible (bijective) transformations to the distributions underlying \(f_{B}^{1}\) and \(f_{B}^{2}\); as these preserve the trade-off function \cite{dong2019gaussian}, we define them as follows: $\mathrm{Proc}^{B}_{j}: x \mapsto \mathrm{Rot}^{B}_{j}\!\left[\left(x - \frac{\sum_{i \in B} \overline{g}_i}{Nq}\right)\frac{Nq}{C\sigma}\right], \quad j \in \{1,2\}$.
Here, $\mathrm{Rot}^{B}_{1}$ is the rotation operation that aligns the direction vector of the last dimension with the vector  $\frac{\sum_{i\in B}\overline{g}_{i}}{(N+1) \cdot q}-\frac{\sum_{i\in B}\overline{g}_{i}}{N \cdot q}$, and $\mathrm{Rot}^{B}_{2}$ is the rotation operation that aligns the direction vector of the last dimension with the vector  $\frac{\overline{g}_{N+1} + \sum_{i\in B}\overline{g}_{i}}{(N+1) \cdot q} - \frac{\sum_{i\in B}\overline{g}_{i}}{N \cdot q}$. By applying \(\mathrm{Proc}^{B}_{1}\) to the distributions underlying \(f_{B}^{1}\), we obtain 
\begin{equation}\label{eq:Q_1_B}
    T(P^{\otimes n}, Q_{N}^{\otimes (n-1)} \otimes Q_{N}^{\mu^{1}_{B}}).
\end{equation}
Similarly, by applying $\mathrm{Proc}^{B}_{2}$ to the distributions that constitute \(f_{B}^{2}\), \(f_{B}^{2}\) can be simplified as 
\begin{equation}\label{eq:Q_2_B}
    T(P^{\otimes n}, Q_{N}^{\otimes (n-1)} \otimes Q_{N}^{\mu_{B}^{2}}),
\end{equation}
\noindent
where $P$, $Q_N$, $Q_N^{\mu_{B}^{1}}$ and $Q_N^{\mu_{B}^{2}}$ are defined in Table~\ref{tab:notation}, and 
$\mu_{B}^{1}=\|-\sum_{i\in B}\bar g_i\|_2/((N+1)C\sigma)$, 
$\mu_{B}^{2}=\|N\bar g_{N+1}-\sum_{i\in B}\bar g_i\|_2/((N+1)C\sigma)$.

By establishing the lower bounds of \(f_{B}^{1}\) and \(f_{B}^{2}\), 
we derive the privacy guarantee for neighboring datasets of sizes \(N\) and \(N+1\). 
The following theorem establishes the privacy guarantee.

\begin{theorem}\label{theorem:EASGM_HIGH}
For each neighboring pair \(D = \{d_{1}, \ldots, d_{N}\}\) and \(D' = D \cup \{d_{N+1}\}\), 
the EASGM mechanism admits the following privacy guarantee:

$\underline{f}(\alpha(t,c))=
  \sum_{B} P(B)\,[
      (1-q)\,\underline{f_{B}^{1}}(\alpha_{B}^{1}(t,c))+q\underline{f_{B}^{2}}(\alpha_{B}^{2}(t,c))
    ]$. 
    
\noindent Here, 
\(\underline{f_{B}^{1}} = T(P^{\otimes n},\, Q_{N}^{\otimes (n-1)} \otimes Q_{N}^{\overline{\mu}_{B}^{1}})\) and 
\(\underline{f_{B}^{2}} = T(P^{\otimes n},\, Q_{N}^{\otimes (n-1)} \otimes Q_{N}^{\overline{\mu}_{B}^{2}})\), 
where \(P\), \(Q_{N}\), and \(Q_{N}^{\mu}\) are defined in Table~\ref{tab:notation}, 
and 
$\overline{\mu}_{B}^{1}=\tfrac{|B|}{(N+1)\sigma}$, 
$\overline{\mu}_{B}^{2}=\tfrac{N+|B|}{(N+1)\sigma}$. The proof is in Appendix~\ref{Appendix:Proof_EASGM_HIGH}.
\end{theorem}
\subsection{Privacy Guarantee of ASGM}\label{sec:ASGM-SGM}
 
Following the notation in Section~\ref{sec:Privacy_high_EASGM}, let \(\mathcal{M}\) denote the ASGM mechanism (Definition~\ref{def:ASGM}).
We can write the distribution of \(\mathcal{M}(D)\) as 
$\sum_{B\neq \emptyset}P(B)(\frac{\sum_{i\in B}\overline{g}_{i}+G}{|B|})+P(B= \emptyset)G$.
Similarly, the distribution of \(\mathcal{M}(D')\) can be written as
\(\sum_{ B \neq \emptyset
}P(B)[(1-q)(\frac{\sum_{i\in B}\overline{g}_{i}+G}{|B|})
+(q)(\frac{\sum_{i\in B}\overline{g}_{i}+\overline{g}_{N+1}+G}{|B|+1})]+P(B=\emptyset)[G\cdot (1-q)+q(\overline{g}_{N+1}+G)]\). Using the same analytical approach as in Section~\ref{sec:guarantee_EASGM}, we establish a lower bound on the trade-off function as follows.

\begin{equation}
\begin{aligned}\label{eq:Low_ASGM}
    &T\bigl(\mathcal{M}(D), \mathcal{M}(D')\bigr)\!\bigl(\alpha(t,c)\bigr)
\;\ge\;
f(\alpha(t,c))\\&=
\sum_{B\neq \emptyset } P(B)((1-q)f_{B}^{1}(\alpha_{B}^{1}(t,c))
          +q\,f_{B}^{2}(\alpha_{B}^{2}(t,c)))\\&+P(B=\emptyset)((1-q)f_{\emptyset}^{1}(\alpha_{\emptyset}^{1}(t,c))+q f_{\emptyset}^{2}(\alpha_{\emptyset}^{2}(t,c))).
\end{aligned}
\end{equation}
\sloppy Here, $f_{B}^{1}$ $=T(\frac{\sum_{i\in B} \overline{g}_{i}+G}{|B|},$ $ \frac{\sum_{i\in B} \overline{g}_{i}+G}{|B|})$ and 
$f_{B}^{2}=T(\frac{\sum_{i\in B} \overline{g}_{i}+G}{|B|},\frac{\sum_{i\in B} \overline{g}_{i}+\overline{g}_{N+1}+G}{|B|+1})$, $f_{\emptyset}^{1}=T(G,G)$,$f_{\emptyset}^{2}=T(G,G+\overline{g}_{N+1})$. For $B\neq \emptyset$ the expressions of \(f_B^1\) and \(f_B^2\) can be simplified 
using an invertible transformation defined as follows:
\begin{equation}\label{eq:invertible_projection_high_dimension_ASGM}
   \mathrm{Proc}^{B}: x \mapsto \mathrm{Rot}^{B}\!\left[\left(x - \frac{\sum_{i \in B} \overline{g}_i}{|B|}\right)\frac{|B|}{C\sigma}\right].
\end{equation}

Here, $\mathrm{Rot}^{B}$ is the rotation that aligns the direction vector of the last dimension with that of 
$\frac{\overline{g}_{N+1}+\sum_{i\in B}\overline{g}_{i}}{|B|+1}-\frac{\sum_{i\in B}\overline{g}_{i}}{|B|}$. 
Then, $f_{B}^{1}$ can be simplified as 
\begin{equation}\label{eq:ASGM_F_1}
    T(P^{\otimes n}, P^{\otimes n})
\end{equation}, and 
$f_{B}^{2}$ can be simplified as 
\begin{equation}\label{eq:ASGM_F_2}
    T(P^{\otimes n}, Q_{B}^{\otimes (n-1)} \otimes Q_{B}^{\mu_{B}})
\end{equation}, 
where $P$, $Q_B$ and $Q_B^{\mu_{B}}$ are defined in Table~\ref{tab:notation}, and 
$\mu_{B}=\|\,|B|\overline g_{N+1}-\sum_{i\in B}\overline g_i\|_2/((|B|+1)C\sigma)$.

By establishing lower bounds for \(f_B^1\), \(f_B^2\), \(f_\emptyset^1\), and \(f_\emptyset^2\), we derive the privacy guarantee for neighboring datasets of sizes \(N\) and \(N+1\).
The following theorem presents the established guarantee.

\begin{theorem}
\label{thm:ASGM_HIGH}
For every neighboring pair of datasets $D = \{d_{1}, d_{2}, \ldots, d_{N}\}$ and $D' = D \cup \{d_{N+1}\}$, 
the ASGM mechanism admits the following privacy guarantee:
\begin{equation}
\begin{aligned}
   & \underline{f}(\alpha(t,c)) 
    = \sum_{B\neq \emptyset} P(B) \bigl[(1-q)\underline{f_{B}^{1}}(\alpha_{B}^{1}(t,c)) 
    + q\,\underline{f_{B}^{2}}(\alpha^{2}_{B}(t,c))\bigr]\\&+(1-q)^{N}[(1-q)\underline{f}^{1}_{\emptyset}(\alpha^{1}_{\emptyset}(t,c))+q\underline{f}^{2}_{\emptyset}(\alpha^{2}_{\emptyset}(t,c))].
\end{aligned}
\end{equation}
Here, $\underline{f_{B}^{1}} = \mathrm{ID}$ and 
$\underline{f_{B}^{2}} = T(P^{\otimes n}, Q_{B}^{\otimes (n-1)} \otimes Q^{\overline{\mu}_{B}}_{B})$, 
$\underline{f}_{\emptyset}^{1}=\mathrm{ID}$, $\underline{f}_{\emptyset}^{2}=G_{1/\sigma}$,
where \(\mathrm{ID}\), \(G_{1/\sigma}\), \(P\), \(Q_B\), and \(Q_B^{\overline{\mu}_B}\) are defined in Table~\ref{tab:notation}, with
$\overline{\mu}_{B}=\frac{2|B|}{(|B|+1)\sigma}$. The proof is in Appendix~\ref{Appendix:Thm_ASGM_HIGH}.
\end{theorem}

\subsection{Comparison with SGM}\label{sec:comparison}

 
 In this section, we compare the privacy guarantees of EASGM and ASGM with the SGM guarantee derived in~\cite{bu2020deep}, as it provides a closed-form \(f\)\nobreakdash-DP baseline for comparison. Our comparison follows the accounting convention: EASGM- and ASGM-based implementations are typically analyzed using the SGM guarantee~\cite{abadi2016deep,bu2020deep}.

Before making comparisons, we first clarify the distinction between privacy leakage and privacy guarantees, so as to distinguish potential violations from actual ones.

The privacy leakage for a neighboring pair \(D\) and \(D'\) under a mechanism \(\mathcal{M}\) is represented by the trade-off function \(T(\mathcal{M}(D), \mathcal{M}(D'))\) (see Section~\ref{sec:differential privacy}). In contrast, a privacy guarantee \(f\) for the pair is a lower bound on \(T(\mathcal{M}(D), \mathcal{M}(D'))\), that is,
\(T(\mathcal{M}(D), \mathcal{M}(D')) \ge f\).
Due to approximations in the analysis, equality does not generally hold.

A violation means that the privacy leakage of EASGM or ASGM exceeds the claimed SGM guarantee \(g\), namely, there exists some \(\alpha\) such that
\(T(\mathcal{M}(D), \mathcal{M}(D'))(\alpha) < g(\alpha)\),
where \(g\) denotes the privacy guarantee claimed for SGM.

Comparing \(f\) with \(g\) can only indicate a potential violation.
If \(f(\alpha) < g(\alpha)\) for some \(\alpha\), then it is possible that
\(T(\mathcal{M}(D), \mathcal{M}(D'))(\alpha) < g(\alpha)\).
A potential violation becomes an actual violation only when the guarantee \(f\) is tight, that is, when
\(T(\mathcal{M}(D), \mathcal{M}(D')) = f\).

\subsubsection{EASGM vs.~SGM}\label{sec:EASGM-SGM}
We now compare the privacy guarantees achieved by EASGM and SGM under identical parameter settings. The per-pair privacy guarantee $f$ of EASGM, defined in Eq.~\ref{eq:lower_bound_EASGM_high_dimentsion}, does not admit a simple closed-form expression. To enable comparison, we therefore derive an upper bound $\overline{f}$ on $f$ that has a tractable analytic form. For a fixed $\alpha$, if the upper bound $\overline{f}(\alpha)$ lies below the SGM guarantee $g(\alpha)$, then the derived guarantee for EASGM can be weaker than the claimed SGM guarantee, indicating a regime where a potential privacy violation may occur.

\begin{theorem}\label{thm:upperbound_EASGM} 
Let \(f\) be the trade-off function defined in Eq.~\ref{eq:lower_bound_EASGM_high_dimentsion}, and let \(P\) and \(Q_N\) be the distributions defined in Table~\ref{tab:notation}. Then,
\begin{equation}
    \overline{f}=T(P^{\otimes (n-1)},Q_{N}^{\otimes (n-1)})\ge f .
\end{equation}
The proof is in Appendix~\ref{Appendix:upperbound_EASGM}.
\end{theorem}

We now show that, as the output dimension $n$ increases, the EASGM privacy guarantee collapses for all neighboring pairs.

\begin{theorem}\label{theorem:EASGM_flaws}
\textbf{Privacy flaw of EASGM.}
Let $\overline{f}$ be the upper bound function defined in Theorem~\ref{thm:upperbound_EASGM}. Then
\begin{equation}
    \lim_{n\to\infty} \overline{f}(\alpha)=0, \quad \forall \alpha\in(0,1].
\end{equation}
The proof is in Appendix~\ref{Appendix:EASGM_flaws}.
\end{theorem}

The privacy guarantee of SGM is given by \( q\, G_{1/\sigma} + (1 - q)\,\mathrm{ID} \)~\cite{bu2020deep},
where \(G_{1/\sigma}\) and \(\mathrm{ID}\) are the trade-off functions defined in
Definitions~\ref{def:Gaussian_DP} and~\ref{def:identical_trade_off}, respectively. We observe that, for every \(\alpha \in (0,1)\), the per-pair upper bound of EASGM's privacy guarantee satisfies
\(\overline{f}(\alpha) < q G_{1/\sigma}(\alpha) + (1 - q)\mathrm{ID}(\alpha)\)
when \(n\) is sufficiently large.
Consequently, the per-pair privacy guarantee of EASGM is weaker than the SGM guarantee
\(q G_{1/\sigma} + (1 - q)\,\mathrm{ID}\).

The per-pair privacy guarantee $f$ in Eq.~\ref{eq:lower_bound_EASGM_high_dimentsion} becomes tight when the sampling ratio \(q = 1\), indicating that for each neighboring pair \((D, D')\), \(T(\mathcal{M}(D), \mathcal{M}(D')) = f\). 
This is because the privacy guarantee \(f\) in Eq.~\ref{eq:lower_bound_EASGM_high_dimentsion} is derived from the approximation introduced in Lemma~\ref{Lemma:Jensen}. 
When \(q = 1\), there is only one distribution in the mixtures \(P_w\) and \(Q_w\); thus, 
\(T(P_w, Q_w) = T((P_I\,|\,I), (Q_I\,|\,I))\). 
In this case, the approximation becomes exact, and there is no overestimation of the per-pair privacy leakage.

As shown above, when \(n\) is large, the upper bound already falls below the claimed SGM guarantee, indicating that a potential privacy violation can occur in EASGM. Moreover, when \(q=1\), the privacy guarantee becomes tight, so this potential violation turns into an actual privacy violation.

\subsubsection{ASGM vs.~SGM}
\label{sec:ASGM_SGM}

Following the same comparison procedure as in Section~\ref{sec:EASGM-SGM},
we first derive an upper bound $\overline{f}$ on the per-pair ASGM privacy
guarantee $f$, defined in Eq.~\ref{eq:Low_ASGM}.

\begin{theorem}\label{thm:upperbound_ASGM} 
Let $f$ be the trade-off function defined in Eq.~\ref{eq:Low_ASGM}, and let the distributions $P$ and $Q_N$, as well as the trade-off function $\mathrm{ID}$, be as defined in Table~\ref{tab:notation}. Then
\begin{equation}
\begin{aligned}\label{eq:ASGM_upper_bound}
   &\overline{f} = (q-q(1-q)^{N})T(P^{\otimes (n-1)}, Q_{N}^{\otimes (n-1)}) \\
   &+ (1 - q+q(1-q)^{N})\mathrm{ID} \ge f.
\end{aligned}
\end{equation}
The proof is in Appendix~\ref{Appendix:upperbound_ASGM}.
\end{theorem}

We now show that, as the output dimension $n$ increases, the upper bound $\overline{f}$ for the ASGM per-pair privacy guarantee converges to \((1-q+q(1-q)^N)\mathrm{ID}\).

\begin{theorem}\label{theorem:ASGM_flaws}
\textbf{Privacy flaw of ASGM.}
Let $\overline{f}$ be the upper bound defined in Theorem~\ref{thm:upperbound_ASGM}. Then we have 
\begin{equation}
    \lim_{n\to\infty} \overline{f}(\alpha)=(1-q+q(1-q)^{N})\mathrm{ID}(\alpha), \quad \forall \alpha\in(0,1].
\end{equation}
The proof is in Appendix~\ref{Appendix:ASGM_flaws}.
\end{theorem}

The privacy guarantee of SGM is given by \( q G_{1/\sigma} + (1 - q)\mathrm{ID} \)~\cite{bu2020deep}, 
where \( G_{1/\sigma} \) and \(\mathrm{ID}\) are the trade-off functions defined in Table~\ref{tab:notation}.
We observe that, for any fixed \(\alpha \in (0,1)\), when \(N\) is sufficiently large and \(q\) is sufficiently close to one,
$\lim_{n\to \infty}\overline{f}(\alpha) - \bigl[q\,G_{1/\sigma}(\alpha) + (1 - q)\mathrm{ID}(\alpha)\bigr]
= q(1-q)^{N}\mathrm{ID}(\alpha) - q\,G_{1/\sigma}(\alpha)< 0$.
Consequently, ASGM's per-pair privacy guarantee is weaker than the SGM guarantee, indicating a potential privacy violation.

When the sampling ratio is \(1\), the two mechanisms, EASGM and ASGM, become equivalent. 
According to Definitions~\ref{def:EASGM} and~\ref{def:ASGM}, their output distributions on any neighboring pair are identical when the sampling ratio equals \(1\). 
Consequently, for any neighboring pair, the privacy guarantee provided by ASGM is identical to that of EASGM.

As shown above, when the sampling ratio $q$ equals~1, the privacy guarantee of ASGM is the same as that of EASGM. 
According to Section~\ref{sec:EASGM-SGM}, the actual privacy leakage can exceed the SGM guarantee, revealing the existence of a privacy violation in ASGM.

\subsection{Privacy Upper Bound Visualization}
\label{sec:upperbound_visulization}Visualizing the privacy guarantees of EASGM and ASGM provides intuition about their privacy behavior. However, since directly visualizing the exact guarantees is computationally expensive, we instead illustrate their upper bounds~$\overline{f}$, defined in Theorems~\ref{thm:upperbound_EASGM} and~\ref{thm:upperbound_ASGM}, respectively.

\sloppy
The challenge is that the upper bound $\overline{f}$ contains a component
$f_{\ast}^{\otimes(n-1)} = T(P^{\otimes (n-1)}, Q_{N}^{\otimes (n-1)})$.
To visualize the upper bound $\overline{f}$, we therefore need to characterize
$f_{\ast}^{\otimes(n-1)}$.
Our strategy is to derive an analytic upper bound for this term. Specifically, we first derive an upper bound for $f_{\ast}$.
Let $\overline{f_{\ast}} = \max\{f_{\ast}, f_{\ast}^{-1}\}$.
According to Proposition~2.2 in~\cite{dong2019gaussian},
$\overline{f_{\ast}}$ is itself a trade-off function and satisfies
$\overline{f_{\ast}} \ge f_{\ast}$.
By Property~\ref{property:order}, this further implies that
$\overline{f_{\ast}}^{\otimes (n-1)} \ge f_{\ast}^{\otimes (n-1)}$.
Using Theorem~\ref{thm:estimation}, we then derive an analytic upper bound for
$\overline{f_{\ast}}^{\otimes (n-1)}$.
Substituting this bound for $f_{\ast}^{\otimes (n-1)}$ into the expressions of
Theorems~\ref{thm:upperbound_EASGM} and~\ref{thm:upperbound_ASGM} then yields
an analytic upper bound for $\overline{f}$ (Theorem \ref{thm:upperbound_estimate}).  We further establish upper bounds for the multi-round EASGM and ASGM training processes in Appendix \ref{sec:Multi_EASGM_ASGM_GUARANTEE}Theorem \ref{thm:EASGM_ASGM_Upperbound_T_round}, and visualize the multi-round upper bound for EASGM in Appendix \ref{sec:Multi_EASGM_ASGM_GUARANTEE} Corollary~\ref{cor:EASGM_upperbound_estimate_T_round}. 

\begin{theorem} \textbf{(Theorem 5 in \cite{dong2019gaussian})}\label{thm:estimation}
Let $f_1,\ldots,f_n$ be symmetric trade-off functions such that
$\kappa_3(f_i)<\infty$ for all $1\le i\le n$. Denote
$\mu := \frac{2\|\mathbf{kl}\|_{1}}
  {\sqrt{\|\boldsymbol{\kappa}_2\|_{1}-\|\mathbf{kl}\|_{2}^{2}}}$,
  $\gamma := \frac{0.56\,\|\overline{\boldsymbol{\kappa}}_{3}\|_{1}}
  {\bigl(\|\boldsymbol{\kappa}_2\|_{1}-\|\mathbf{kl}\|_{2}^{2}\bigr)^{3/2}}$.
Assume $\gamma<\tfrac{1}{2}$. Then, for all $\alpha\in[\gamma,\,1-\gamma]$, we have
$
  G_\mu(\alpha+\gamma)-\gamma
  \;\le\;
  (f_1\otimes f_2\otimes\cdots\otimes f_n)(\alpha)
  \;\le\;
  G_\mu(\alpha-\gamma)+\gamma.$
\end{theorem}

\begin{theorem}\label{thm:upperbound_estimate}
Let $f_{\ast} = T(P, Q_{N})$, where $P$ and $Q_{N}$ are defined in Table~\ref{tab:notation}, and let $\overline{f_{\ast}} = \max\{f_{\ast}, f_{\ast}^{-1}\}$.
For all $\alpha \in [\gamma,\, 1 - \gamma]$, the following bounds hold:
1. For $\overline{f}$ in Theorem~\ref{thm:upperbound_EASGM}, we have
$\overline{f}(\alpha) \le G_{\mu}(\alpha - \gamma)+\gamma$.
2. For $\overline{f}$ in Theorem~\ref{thm:upperbound_ASGM}, we have
$\overline{f}(\alpha) \le (q-q(1-q)^{N})\,( G_{\mu}(\alpha - \gamma)+\gamma) + (1 - q+q(1-q)^{N})\, ID(\alpha)$, where $\mu$ and $\gamma$ are derived from the composition of $(n-1)$ independent copies of $\overline{f_{\ast}}$. 
The proof is in Appendix~\ref{App:proof:estimationg}.
\end{theorem}

\subsection{Takeaways}
\label{sec:Takeaways}

Using the upper bounds derived in Theorems~\ref{thm:upperbound_EASGM} and~\ref{thm:upperbound_ASGM}, we now summarize several key observations about the privacy behavior of EASGM and ASGM under different conditions.

For both EASGM and ASGM, our analysis shows that the analytical upper bound \(\overline{f}\) decreases as the output dimension increases (Lemma~\ref{fact:dimension}). For EASGM, \(\overline{f}\) further decreases as the dataset size becomes smaller (Lemma~\ref{fact:number_of_data}). For ASGM, \(\overline{f}\) further decreases as the sampling ratio \(q\) increases (Lemma~\ref{fact:sample_ratio}).

\begin{lemma}\label{fact:number_of_data}
Let \(\overline{f}\) be the function defined in Theorem~\ref{thm:upperbound_EASGM}. 
For any \(\alpha \in [0,1]\) and \(n\in \mathbb{Z}_{>0}\), \(\overline{f}(\alpha)\) weakly decreases as \(N\) decreases.
The proof is given in Appendix~\ref{Proof:Lemma_number_of_data}.
\end{lemma}

\begin{lemma}\label{fact:dimension}
Let \(\overline{f}\) be the function defined in Theorems~\ref{thm:upperbound_ASGM} and~\ref{thm:upperbound_EASGM}. 
For any \(\alpha \in [0,1]\) and \(N\in \mathbb{Z}_{>0}\), \(\overline{f}(\alpha)\) weakly decreases as \(n\) increases.
The proof is given in Appendix~\ref{Proof:Lemma:dimension}.
\end{lemma}

\begin{lemma}\label{fact:sample_ratio}
Let \(\overline{f}\) be the function defined in Theorem~\ref{thm:upperbound_ASGM}. 
Given fixed \(N\in\mathbb{Z}_{>0}\) and \(n\in\mathbb{Z}_{>0}\), for any \(\alpha\in[0,1]\), \(\overline{f}(\alpha)\) weakly decreases as \(q\) increases.
The proof is given in Appendix~\ref{Appendix:proof_asgm_sample_ratio}.
\end{lemma}

\section{EASGM and ASGM Auditing}
\label{sec:Auditing_method}

In this section, we propose an auditing method for EASGM- and ASGM-based DP-SGD implementations to test whether the privacy leakage in their per-round training procedures violates the SGM privacy guarantee. Since our goal is to detect privacy violations, we focus on the parameter regime in which the privacy guarantee is weakest.

The main challenge lies in performing tight auditing in a high-dimensional output space: directly applying the likelihood-ratio test given by the Neyman--Pearson lemma is computationally inefficient. To address this issue, we exploit the fact that, under each hypothesis, the output can be normalized by a suitable invertible linear transformation so that the transformed output has i.i.d.\ standard Gaussian coordinates. This allows us to replace direct likelihood-ratio testing with a normality-based auditing procedure.

\noindent\textbf{Distinguisher design.}
Lemma~\ref{Lemma:Jensen} provides a design principle for distinguishers. When the adversary observes the auxiliary information \(I\), the optimal distinguisher is obtained by conditioning on \(I\) and applying the corresponding per-condition optimal test. When \(I\) is unobservable but takes only a small number of values, the adversary can instead apply all per-condition distinguishers and aggregate their outputs.

By Lemma~\ref{Lemma:Jensen}, the privacy analysis reduces to distinguishing the pair of distributions defining each trade-off function \(f_B^i\). Following the notation in Sections~\ref{sec:ASGM-SGM} and~\ref{sec:guarantee_EASGM}, we construct a distinguisher \(A_B^i\) for the associated distributions \(P_B^i\) and \(Q_B^i\). Both are \(n\)-dimensional distributions with independent components, but generally have different means and standard deviations, denoted by \((\mu_{P,B}^i,\sigma_{P,B}^i)\) and \((\mu_{Q,B}^i,\sigma_{Q,B}^i)\), respectively. We define
\(\mathrm{Proc}_{P_B^i}(x)=\frac{x-\mu_{P,B}^i}{\sigma_{P,B}^i}\) and
\(\mathrm{Proc}_{Q_B^i}(x)=\frac{x-\mu_{Q,B}^i}{\sigma_{Q,B}^i}\),
which map \(P_B^i\) and \(Q_B^i\) to \(n\)-dimensional vectors with i.i.d.\ standard normal components.

Let \(\mathrm{Test}(\cdot)\) denote a statistical test for whether an \(n\)-dimensional vector has i.i.d.\ standard normal components. Given an observed output \(x\), \(A_B^i\) applies \(\mathrm{Test}\) to both \(\mathrm{Proc}_{P_B^i}(x)\) and \(\mathrm{Proc}_{Q_B^i}(x)\): it outputs \(1\) if only \(\mathrm{Proc}_{P_B^i}(x)\) passes, \(0\) if only \(\mathrm{Proc}_{Q_B^i}(x)\) passes, and a uniformly random bit in \(\{0,1\}\) otherwise. The outputs of all \(A_B^i\) are then aggregated by the final distinguisher \(A_\ast\) to decide whether the underlying dataset is \(D\) or \(D'\).

\noindent\textbf{Auditing algorithm.}
Algorithm~\ref{alg:score-function} in Appendix~\ref{Appendix:auditing_algorithm} summarizes the auditing procedure.
\section{Experiment}
\subsection{Experiment Setup}\label{sec:Experiment_set_up}
All experiments are performed on a dual-socket Intel Xeon Gold~5320 server equipped with a single NVIDIA A40 GPU featuring 40~GB of device memory.

\noindent \textbf{Baseline.}
Since our experiments aim to compare the audited privacy leakage of EASGM- and ASGM-based DP\text{-}SGD implementations in the one-round setting with the SGM-based privacy guarantee they claim to satisfy, we adopt the state-of-the-art DP\text{-}SGD analysis of~\cite{bu2020deep} as the baseline. This guarantee has been adopted in privacy accounting, including by the DPSUR implementation (see \cite{dpsur_gdp_accountant_line25}) and by the dp-promise algorithm (see Lemma~3 of~\cite{wang2024dp}). Although Opacus implements a GDP accountant~\cite{opacus_gdp_accountant_v154_line23} based on~\cite{gopi2021numerical}, it targets multi-round DP-SGD and relies on numerical approximation. In our one-round setting, the guarantee of~\cite{bu2020deep} is tight and analytic, making it the appropriate baseline.

\noindent \textbf{Audited algorithms.} 
We evaluate four state-of-the-art differentially private machine learning methods based on DP-SGD: DPSGD-HF~\cite{tramer2020differentially}, DPSUR~\cite{fu2024dpsur}, DP-FETA~\cite{11023490}, and PrivImage~\cite{li2024privimage}. Their implementations are publicly available, and we audit their one-round DP-SGD updates.

Three of these methods—DPSGD-HF, DP-FETA, and PrivImage—are built on Meta’s open-source library \texttt{Opacus}~\cite{yousefpour2021opacus}, but rely on different library versions. In particular, PrivImage \cite{privimage_opacus_version} and DP-FETA \cite{dpfeat_opacus_version} rely on \texttt{Opacus} v1.0, whose update rule corresponds to FEASGM, while DPSGD-HF \cite{dpsgdhf_opacus_version} relies on \texttt{Opacus} v0.13.0, whose update rule corresponds to ASGM (Table~\ref{tab:audited_libraries} in Appendix presents more details). However, according to their paper descriptions, these methods are based on the EASGM update rule. 

To control for this confounding factor, we reimplement these three methods under the \texttt{Opacus} version corresponding to the update rule stated in their papers. Specifically, since these methods are described based on EASGM, we audit them using \texttt{Opacus} version v0.15.0; see Table~\ref{tab:audited_libraries}. In contrast, DPSUR manually implements the ASGM variant of DP-SGD and therefore is not affected by version-dependent differences in \texttt{Opacus} \cite{dpsur_code_line146}. We therefore directly audit its implementation. Collectively, these four methods cover both EASGM- and ASGM-based DP-SGD implementations, making them natural targets for our auditing.

\noindent \textbf{Privacy metric.} Following previous work~\cite{chua2024private}, 
we use the trade-off function and the privacy profile $\delta(\epsilon)$ as our privacy metrics. 
The value $\delta(\epsilon)\in[0,1]$ is derived from the audited results under a fixed $\epsilon$ 
using the method in Eq.~\ref{eq:audited_delta}. 
A larger $\delta(\epsilon)$ indicates higher privacy leakage.

\noindent\textbf{Parameter settings.}  
Tight auditing~\cite{annamalai2024you} aims to audit the maximum privacy leakage of a mechanism. It typically requires performing auditing on two small neighboring datasets that are easy to distinguish, i.e., two datasets $D$ and $D’$ with a small overlap $D^{-} = D \cap D'$. For example, in \cite{nasr2021adversary} and \cite{annamalai2024you}, $D^{-} = \emptyset$ and $|D^{-}| = 2$, respectively. Annamalai et al.~\cite{annamalai2024you} introduce the principle of tight auditing and suggest conducting audits on small neighboring datasets, as this facilitates the detection of maximum privacy leakage while substantially reducing computational cost.

To audit EASGM and ASGM, we follow the tight auditing setup and design an auditing procedure tailored to small-sized neighboring datasets\footnote{The cost of tight auditing grows rapidly with the dataset size, making large-$N$ auditing computationally prohibitive. In particular, to achieve tight auditing for high-dimensional outputs, we trade off efficiency against scalability to larger datasets. This limitation should not be interpreted as evidence that the SGM-based guarantee remains sound for normalized DP-SGD mechanisms at larger $N$; rather, it reflects the current scalability limit of the auditing algorithm.}. Specifically, we set the size of the shared subset of neighboring datasets to $|D^{-}| = N$, where $D^{-} = D \cap D'$, with \(N \in \{2,3,4,5\}\). Another reason for targeting small datasets is that our theoretical analysis (Lemma~\ref{fact:number_of_data}) shows that, for EASGM, the upper bound of the privacy guarantee, $\overline{f}$, decreases as the dataset size decreases. This indicates a weaker privacy guarantee for each pair of neighboring datasets of smaller size. 
Consequently, privacy auditing is more likely to detect the maximum leakage when performed on smaller datasets.

We set the clipping bound to $C=0.1$ and the noise scale to $\sigma=10$. 
In theory, a large noise scale makes the SGM's derived privacy guarantee appear strong, while the privacy guarantees of EASGM and ASGM remain weak, thereby making the discrepancy between their actual guarantees and the SGM guarantee more apparent. For EASGM, we vary the sampling ratio $q\in\{0.005, 0.1, 0.5, 0.6, 0.9, 1\}$, 
ranging from small to large values. 
According to the privacy analysis, the sampling ratio does not affect the upper bound of the privacy guarantee of EASGM.
(Theorem~\ref{thm:upperbound_EASGM}). 
For ASGM, we vary $q\in\{0.5, 0.6, 0.7, 0.8, 0.9\}$, 
since larger sampling ratios make the privacy guarantee weaker 
(Lemma~\ref{fact:sample_ratio}).

\noindent \textbf{Auditing time and confidence interval.} For each audited result, we perform \(10{,}000\) auditing trials and 
adopt a \(95\%\) confidence interval, following prior work~\cite{annamalai2024you}, 
to ensure the statistical reliability of the audited results.

\begin{figure}[!t]
    \centering
\includegraphics[width=0.4\textwidth]{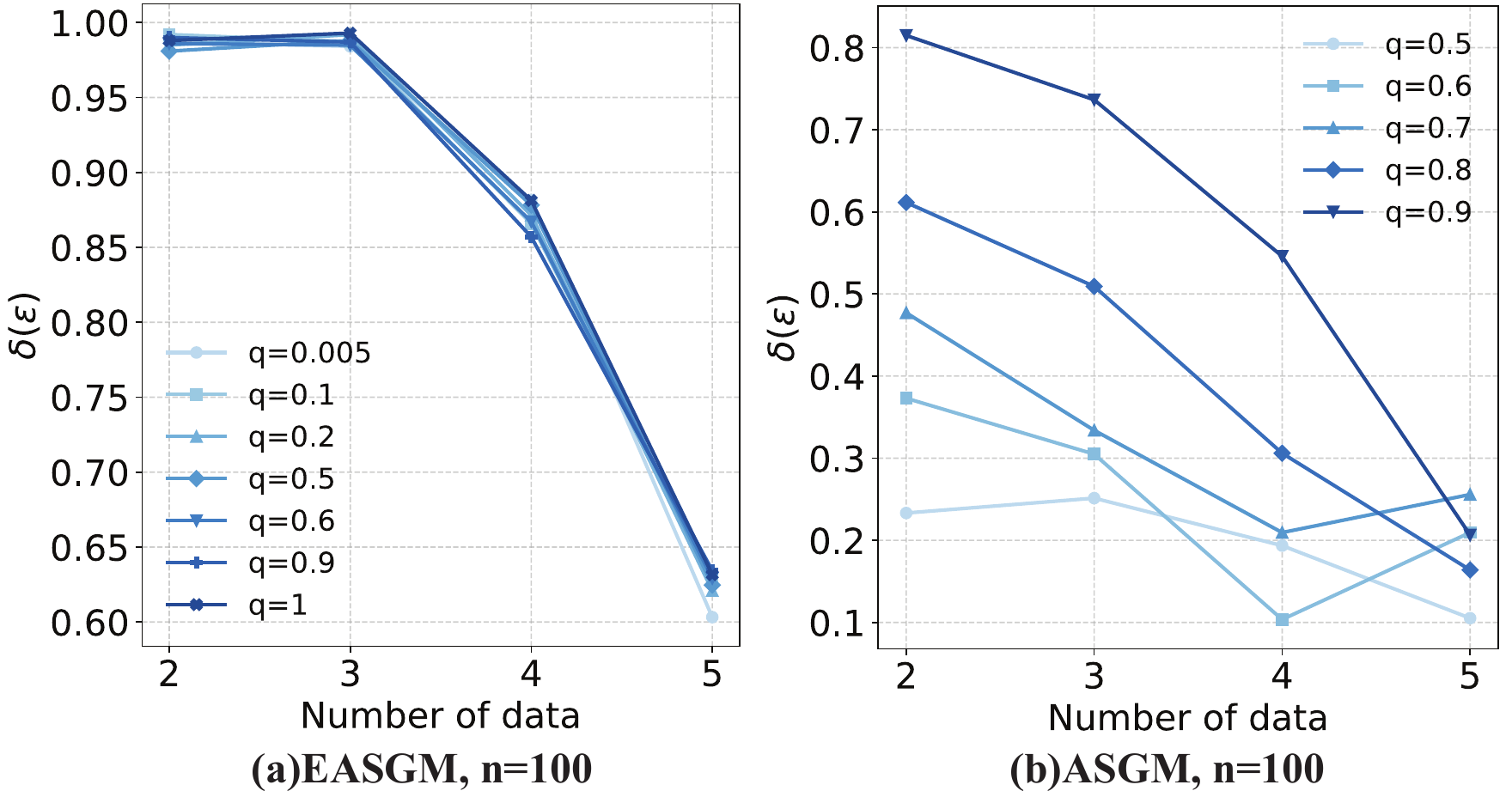}
\vspace{-0.2em}
    \caption{Audited privacy leakage $\delta(\epsilon)$ evaluated at $\epsilon = 0$, with the output dimension fixed at $n = 100$. 
The x-axis indicates the number of data points $N$, and the y-axis reports the audited leakage $\delta(\epsilon)$.}
\label{fig:DELTA_1}
\end{figure}

\subsection{Auditing with Synthetic Datasets}
\label{sec:RQ_1}
\begin{figure}
    \centering
\includegraphics[width=0.4\textwidth]{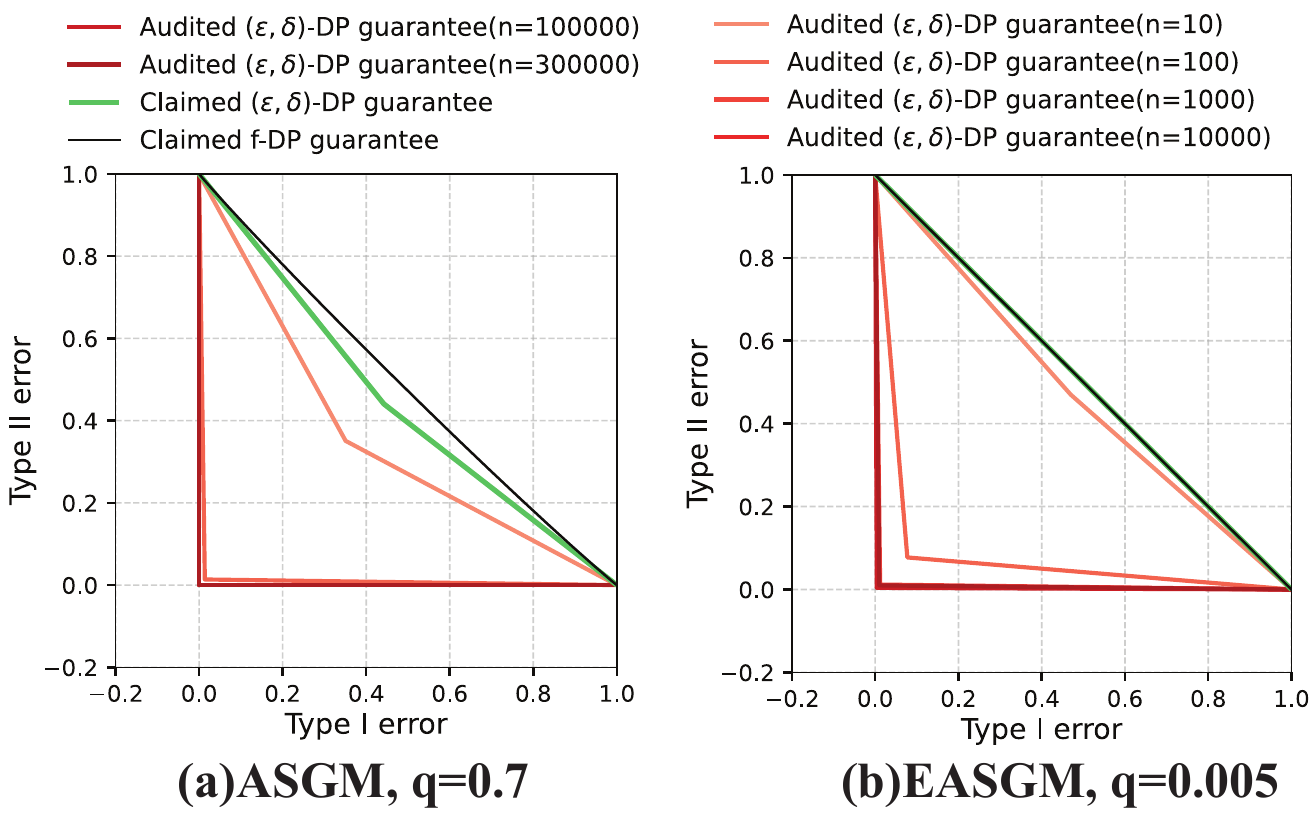}
  \caption{Audited leakage using the $(\epsilon,\delta)$-trade-off function compared against the privacy guarantee claimed under SGM, with $\delta = 10^{-5}$ and $N=4$ fixed.}

    \label{fig:DIMENSION}
\end{figure} 
Synthetic data can be generated in arbitrary dimensions, which allows us to examine how the audited leakage scales with different output dimension $n$. 
So we first use a synthetic dataset to systematically audit whether the privacy leakage of EASGM and ASGM follows the trends predicted by our theoretical analysis.

The experimental results are shown in Figs.~\ref{fig:DELTA_1} and~\ref{fig:DIMENSION}. 
Fig.~\ref{fig:DELTA_1} shows how the privacy leakage changes with different data sizes and sampling rates, where we fix the dimension $n=100$. 
From the figure, we can see that, regardless of what $q$ is, the audited privacy leakage show a generally decrease trend as the number of data points increases for both EASGM and ASGM, which follows the trends predicted by our theoretical analysis in Section~\ref{sec:Takeaways}. 
In particular, for ASGM (Fig.~\ref{fig:DELTA_1}(b)), we observe two trends in the audited privacy leakage. 
First, the leakage does not always decrease monotonically with \(N\), as shown by the cases \(q=0.6\) and \(q=0.7\). 
This non-monotonicity arises because the privacy guarantee depends on two terms that vary in opposite directions as \(N\) increases, as discussed in Section~\ref{sec:RQ3}. 
Second, the audited leakage generally increases with the sampling ratio \(q\), consistent with the trend predicted by Lemma~\ref{fact:sample_ratio}.

Fig.~\ref{fig:DIMENSION} shows the $(\epsilon,\delta)$-trade-off function for 6 different $n$ values with $q=0.7$ for ASGM and $q=0.05$ for EASGM. 
If the mechanism has more privacy leakage, the adversary can take advantage of more information to win the game with fewer type-I or type-II errors. Thus, for the $(\epsilon,\delta)$-trade-off curve, getting closer to the origin point indicates more privacy leakage.
From Fig.~\ref{fig:DIMENSION}, we can first observe that for both EASGM and ASGM, the red curves get lower as the dimension $n$ increases, which means the audited privacy leakage increases as the output dimension $n$ becomes larger, following our theoretical analysis. 
Fig.~\ref{fig:DIMENSION} also shows that all the red curves are under the green curves, indicating the audited privacy leakage of both EASGM and ASGM exceeds the privacy guarantee predicted by the SGM analysis, which also follows our theoretical analysis.
Notably, when $n\geq 100$, the distance from the green curve gets remarkably larger, showing that the privacy guarantees of EASGM and ASGM can be substantially weaker than the SGM guarantee\footnote{$q$ should also be large for ASGM.}.





\tcbset{
  answerbox/.style={
    colback=gray!10,       
    colframe=black,        
    boxrule=0.6pt,         
    arc=1mm,               
    left=2mm, right=2mm, top=1mm, bottom=1mm,
  }
}

\subsection{Auditing with Real-world Datasets} \label{sec:auditing_with_real_world_datasets}
\begin{figure*}[tp]
    \centering
\includegraphics[width=0.9\textwidth]{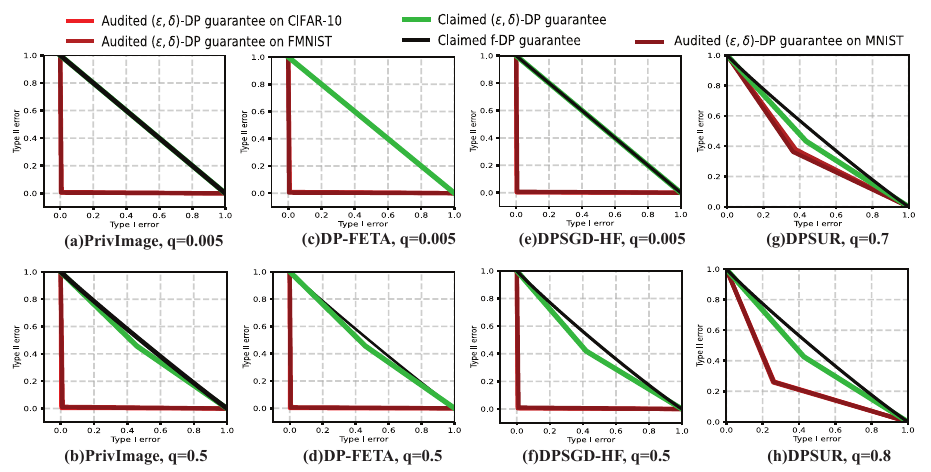}
 \caption{Audited privacy leakage of real-world DP\mbox{-}SGD implementations, computed from the \((\epsilon,\delta)\)-trade-off function with \(\delta = 10^{-5}\) under the fixed setting \(N = 4\), and compared with their claimed privacy guarantees.}

    \label{fig:IMPLEMENTATION}
\end{figure*}
We also audit the one-round model updates of EASGM- and ASGM-based DP-SGD implementations on real-world datasets. Specifically, we use CIFAR-10~\cite{fu2024dpsur}, MNIST~\cite{thudi2024gradients}, and FMNIST~\cite{wei2022dpis} to evaluate privacy leakage in practical DP-SGD implementations. These widely used benchmark datasets are compatible with many differentially private learning algorithms~\cite{fu2024dpsur,10.1145/3658644.3690194,li2024privimage,tramer2020differentially}, making them suitable for auditing such implementations. We fix the neighboring dataset sizes to \( |D| = 4 \) and \( |D'| = 5 \). For each dataset, we randomly select \(5\) data points to form neighboring pairs. To examine whether the privacy leakage exceeds the claimed SGM guarantee under the weakest regime, we follow the setup in Section~\ref{sec:Experiment_set_up} and use the SGM guarantee established in~\cite{bu2020deep} as the claimed baseline.

The experimental results are shown in Fig.~\ref{fig:IMPLEMENTATION}. From the figure, we can observe that all audited implementations exhibit privacy violations. This is because all four algorithms use models with large numbers of trainable parameters: DPSUR and DPSGD-HF each contain more than \(20{,}000\) parameters, while PrivImage and DP-FETA each contain more than one million. This leads to a very large output dimension \(n\), which, combined with the small dataset size (\(N = 4\)), weakens the privacy guarantee, as discussed in Section~\ref{sec:Takeaways}. The discrepancy between the audited leakage and the SGM guarantee is substantial for EASGM and relatively smaller for ASGM, which is consistent with the theoretical predictions in Section~\ref{sec:RQ3}. Moreover, for ASGM, the discrepancy becomes larger as the sampling ratio increases to \(0.8\), because our analysis predicts a weaker privacy guarantee as \(q\) increases (Lemma~\ref{fact:sample_ratio}).

\subsection{Privacy Guarantees for Large Datasets}
\label{sec:RQ3}
\begin{figure*}[tp]
    \centering
\includegraphics[width=1\textwidth]{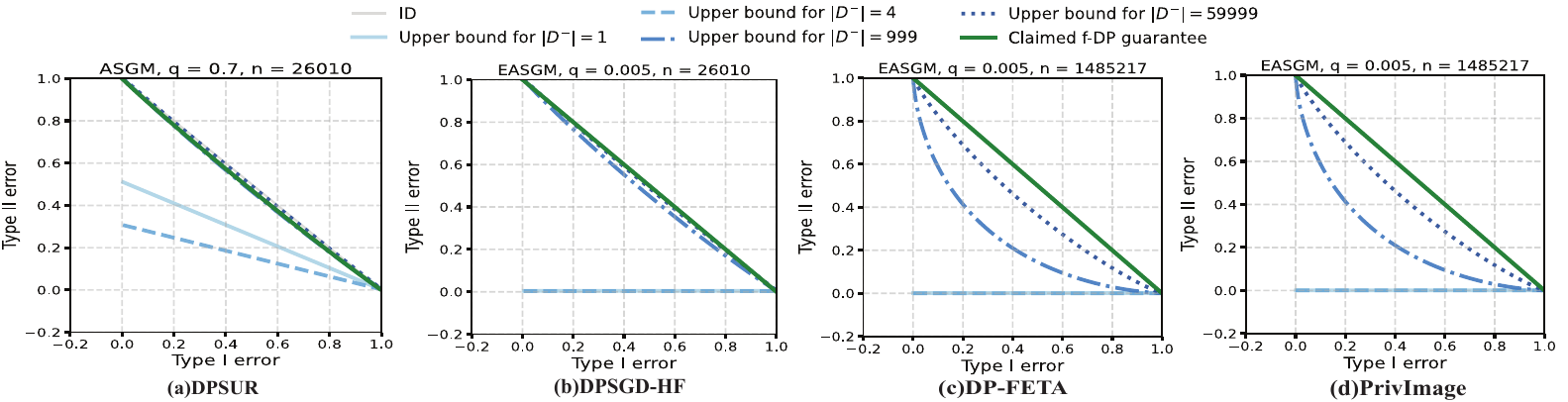}
    \caption{Theoretical upper bounds on the per-round privacy guarantees of real-world implementations, with each privacy guarantee lying below its corresponding upper bound.
}
    \label{fig:Theoretical}
\end{figure*}
Due to the heavy computational overhead, our tight auditing method is not effective for auditing large datasets. 
To show how the privacy guarantees vary across large dataset sizes, we use the method described in Theorem~\ref{thm:estimation} to visualize upper bounds on the privacy guarantees of the EASGM- and ASGM-based DP-SGD implementations audited in Section~\ref{sec:auditing_with_real_world_datasets}. 

For the EASGM-based implementations, we set the sampling ratio to \(q=0.005\), whereas for the ASGM-based implementation, we set \(q=0.7\). We use the models adopted by these implementations for MNIST and FMNIST. Specifically, DP-SUR and DPSGD-HF use models with \(26{,}010\) trainable parameters, whereas PrivImage and DP-FETA use models with \(1{,}485{,}217\) trainable parameters. To assess how the privacy guarantees vary with the neighboring dataset size, we visualize the upper bounds for four values of \(|D^-| \in \{1,4,999,59{,}999\}\). Here, \(|D^-|=59{,}999\) corresponds to the full MNIST/FMNIST training set, which contains \(60{,}000\) samples.

The results are shown in Figure~\ref{fig:Theoretical}. For small datasets, namely \(|D^-| \in \{1,4,999\}\), all methods exhibit guarantees weaker than their claimed SGM guarantees. For larger datasets such as the full MNIST and FMNIST training sets, all EASGM-based implementations, including PrivImage, DP-FETA, and DPSGD-HF, still show noticeably weaker guarantees than claimed, whereas the upper bound for the ASGM-based implementation DP-SUR remains close to the claimed SGM guarantee.

These trends are consistent with our theoretical analysis. For the ASGM-based implementation, the upper bound does not vary monotonically with the dataset size. This is because the upper bound in Eq.~\eqref{eq:ASGM_upper_bound} contains two competing effects: As \(N\) increases, the term
$(q-q(1-q)^N)T(P^{\otimes(n-1)},Q_N^{\otimes(n-1)})$
increases, while the term
$(1-q+q(1-q)^N)\mathrm{ID}$
decreases. In contrast, for EASGM-based implementations, the upper bound increases with the dataset size and decreases with the output dimension, matching the theoretical trends established in Lemma~\ref{fact:number_of_data} and Lemma~\ref{fact:dimension}.

\subsection{Opacus Code Audit}\label{sec:Opacus}

We conduct a code audit of Opacus across versions to identify the normalization schemes used in its DP-SGD implementations. 
Opacus has supported Poisson sampling since v0.13.0 and introduced explicit handling of Poisson-sampled batches in v0.15.0. Under \texttt{poisson\_sampling=True}, the \texttt{loss\_reduction} argument determines whether aggregated gradients are normalized before the update; in particular, \texttt{loss\_reduction=mean}, the default setting, averages the batch gradient. We therefore focus on this default case. The audited normalization schemes are summarized in Table~\ref{tab:audited_libraries} in the Appendix.

We observe that Opacus involves a floor-based normalization, where the summed noisy gradients are normalized by \(\lfloor N\cdot q\rfloor\). We call it \textit{FEASGM }and define it in Definition~\ref{def:FEASGM}. Although this normalization behavior has been documented in prior auditing work \cite{cebere2026privacy}, to the best of our knowledge, it has not been studied as a mechanism design in the DP literature. We therefore analyze its privacy behavior in detail.
Cebere et al.~\cite{cebere2026privacy} reported a privacy issue in Opacus's FEASGM implementation on small datasets \cite{opacus_issue_571}. Their audit found a gap between the built-in accountant's guarantee, approximately \(1.2\), and the audited leakage, approximately \(2.5\), while observing tight accounting for large datasets. 
We have two new findings: 1) for small datasets, when the output dimension is high, the audited leakage can be significantly higher than the SGM guarantee; see Section~\ref{sec:Opacus_auditing}; and 2) for large datasets under practical training settings, the privacy guarantee can still be significantly weaker than that of SGM; see Section~\ref{sec:visualizing_FEASGM}.


\section{Related Work}

\noindent\textbf{Poisson-Sampled DP-SGD.}
DP-SGD was first introduced by Abadi et al.~\cite{abadi2016deep}. 
Its per-round update normalizes the summed gradients by the expected batch size, 
an averaging operation that we formalize as EASGM. 
However, in their original privacy analysis, each per-round update was modeled as a 
subsampled Gaussian mechanism (SGM), and the privacy guarantee was derived under 
a security game in which the adversary is assumed to know the sampling ratio \(q\) 
and noise scale \(\sigma\), while the expected batch size remains unknown. 
A variety of differentially private learning methods, including 
DP-FETA~\cite{11023490}, PrivImage~\cite{li2024privimage}, 
DP-SUR~\cite{fu2024dpsur}, and DPSGD-HF~\cite{tramer2020differentially}, 
build on the DP-SGD training framework and adopt the same SGM-based privacy analysis.

Dong et al.~\cite{dong2019gaussian} proposed the \(f\)-DP framework, which characterizes privacy guarantees through a trade-off function between Type~I and Type~II errors. This operational view has been widely adopted in privacy auditing~\cite{xiang2025privacy,nasr2023tight,koskela2025auditing}. Building on this framework, the same research team later derived the state-of-the-art one-round privacy analysis for DP\mbox{-}SGD and proposed an ASGM variant~\cite{bu2020deep}. By modeling each per-round update as an SGM, they obtained an explicit trade-off function and established a tight one-round guarantee.

Recent investigations have highlighted inconsistencies between certain 
DP\mbox{-}SGD implementations and their claimed privacy guarantees.
Chua et al.~\cite{chua2024private} reported that some implementations use shuffled 
sampling while reporting SGM-based guarantees, and reanalyzed the privacy of 
shuffle-based DP-SGD under this mismatch. 
Lebeda et al.~\cite{lebeda2025avoiding} pointed out that the expected-averaging 
operation used in EASGM-based DP\mbox{-}SGD may cause the per-round update 
to violate the privacy guarantee of SGM, although no formal privacy analysis 
or auditing was provided.

Motivated by these observations, we reexamine the per-round \(f\)-DP guarantees of EASGM- and ASGM-based DP\mbox{-}SGD~\cite{bu2020deep,dong2019gaussian} under the same adversary assumption as~\cite{abadi2016deep}.

\noindent\textbf{Privacy Auditing of DP-SGD.}
Nasr et al.~\cite{nasr2021adversary} systematically introduced privacy auditing for DP\mbox{-}SGD by evaluating a hierarchy of adversaries with increasing capabilities. They showed that stronger adversaries reveal greater leakage, with the strongest matching the theoretical upper bound. Their later tight auditing work~\cite{nasr2023tight} proposed a canary-based tight auditing framework for DP\mbox{-}SGD. 
In white-box settings, they set the canary sampling rate to $q_c=1$, and their audits found no violations in SGM-based DP\mbox{-}SGD implementations.

Annamalai et al.~\cite{annamalai2024you} later formalized the tight auditing 
principle, arguing that auditing should be performed on small datasets. 
Steinke et al.~\cite{steinke2023privacy} proposed a one-run auditing method 
to reduce computational cost. Xiang et al.~\cite{xiang2025privacy} further 
extended this approach by formalizing the auditing process as a bit-transmission 
problem, characterizing the feasible and infeasible conditions for applying the 
technique, and proposing methods to tighten the resulting privacy bounds. 
Koskela et al.~\cite{10992373} adopted the same SGM-based formalization, 
auditing each per-round DP\mbox{-}SGD update using a density-estimation–based method. 

Despite extensive auditing of DP-SGD implementations in prior work, the per-round training procedures of EASGM- and ASGM-based DP-SGD have not been specifically examined, and the privacy behavior of these implementations has therefore remained unexplored.

Cebere et al.~\cite{cebere2026privacy} audited Opacus's FEASGM-based DP-SGD implementation on small synthetic datasets with low-dimensional outputs and found tight accounting for large datasets. In contrast, we show that real-world FEASGM implementations can exhibit substantially larger privacy leakage in high-dimensional output spaces; see Appendix~\ref{sec:Opacus_auditing}. We further derive a multi-round upper bound for FEASGM-based DP-SGD (Theorem~\ref{Thm:FEASGM_Upperbound_T_round}) and show that, for certain neighboring dataset pairs, its privacy guarantee can be significantly weaker than the SGM guarantee; see Appendix~\ref{sec:visualizing_FEASGM}.

\section{Conclusion and Limitations}
In this paper, we revisit and re-audit the per-round privacy guarantees of EASGM- and ASGM-based DP-SGD under \(f\)-DP. 
We show that these guarantees can weaken as the output dimension increases and can be weaker than the corresponding SGM guarantee. 
Since Opacus implementations use FEASGM, we further extend our analysis to FEASGM and visualize the multi-round privacy guarantees for both EASGM (Corollary \ref{cor:EASGM_upperbound_estimate_T_round}) and FEASGM (Theorem \ref{thm:upperbound_estimate_T_round}). 
This analysis is challenging because the relevant trade-off functions involve Gaussian distributions with unequal variances, whereas prior theoretical results mainly focus on the equal-variance case. 
We therefore derive new theoretical tools tailored to trade-off functions between Gaussian distributions with unequal variances.

Our analysis and auditing have the following limitations:

\noindent \textbf{No privacy-guarantee visualization for multi-round ASGM.}
For ASGM, we do not derive a closed-form characterization or visualization of the multi-round guarantee because its per-round upper bound contains an analytically intractable trade-off term. We leave it to future work. 

\noindent \textbf{Limited auditing scope.}
Our tight auditing method is computationally expensive in high-dimensional output spaces, and its cost grows rapidly with dataset size. We therefore focus on small neighboring datasets and parameter regimes where privacy leakage is expected to be large. Moreover, the current audit covers only single-step training and does not extend to multi-step settings. Scaling the audit to more practical settings remains future work.

\bibliographystyle{acm}
\bibliography{references}
\appendix
\section{Privacy Analysis of FEASGM}

To understand the privacy behavior of FEASGM, we analyze its privacy guarantee under the same setting as that used for the privacy analysis of EASGM and ASGM in Section~\ref{sec:analysis}. The derivation directly follows the analytical framework and techniques introduced there.

\begin{definition}\label{def:FEASGM}
    \textbf{Floor-based Expected Average Subsampled Gaussian Mechanism (FEASGM)\cite{cebere2026privacy}.} Following the same setting as Definition~\ref{def:SGM}, the output of FEASGM can be formalized as
$\frac{\sum_{i \in B} \bar{g}_i + \mathcal{N}(0, \sigma^2 C^2 \mathbb{I})}{\lfloor N\cdot q\rfloor}$,
where $N$ denotes the total number of data points in the dataset $D$.
 
\end{definition}

 \noindent \textbf{Analysis roadmap.} 
We begin by identifying the conditions under which the privacy behavior of FEASGM reduces to that of SGM and those under which it follows the EASGM-style regime. In particular, this distinction is determined by whether the normalization factors $\lfloor N\cdot q \rfloor$ and $\lfloor (N+1)\cdot q \rfloor$ are equal. When they are equal, the normalization step is identical on both neighboring datasets and can therefore be treated as post-processing, so the privacy guarantee reduces to that of SGM. Otherwise, the two normalization factors differ by exactly \(1\), and the privacy behavior follows the EASGM-style regime.

We therefore focus on the EASGM-style scenario. Building on the EASGM analysis in Sections~\ref{sec:guarantee_EASGM} and~\ref{sec:EASGM-SGM}, we derive the corresponding FEASGM results by replacing \(N\) in those theorems with \(K=\lfloor Nq \rfloor\). Finally, using the CLT-based result in Theorem~\ref{thm:estimation}, we establish and visualize an upper bound on the privacy guarantee of \(T\)-round FEASGM training in Theorem~\ref{thm:upperbound_estimate_T_round}.

\subsection{Formal Results.} 
\noindent \textbf{Overview.} We present the formal results as follows. We first establish the per-pair guarantee in Theorem~\ref{thm:per_pair_FEASGM}, and then derive the worst-case guarantee in Theorem~\ref{thm:guarantee_of_FEASGM}. Next, we construct an upper bound on the per-pair guarantee to facilitate visualization of the FEASGM per-pair guarantee. Finally, we extend this upper bound to the \(T\)-round setting and derive a corresponding upper bound for visualizing the privacy guarantee of \(T\)-round training in Theorem~\ref{thm:upperbound_estimate_T_round}.

\begin{theorem}[Per-pair Privacy Guarantee of FEASGM] \label{thm:per_pair_FEASGM}
Let \(D=\{d_1,d_2,\dots,d_N\}\) and \(D' = D \cup \{d_{N+1}\}\) be two neighboring datasets, and let \(q\) be a fixed sampling ratio.

If $\lfloor N\cdot q \rfloor = K$
and
$\lfloor (N+1)\cdot q \rfloor = K+1,$
then the privacy guarantee is given by
\begin{equation}\label{eq:FEASGM_High_per_pair}
\begin{aligned}
  f(\alpha(t,c))
  \;=\;&
  \sum_{B} P(B)\,[
      (1-q)\,f_{B}^{1}(\alpha_{B}^{1}(t,c))\\
      &\qquad\qquad
      + q\,f_{B}^{2}(\alpha_{B}^{2}(t,c))
    ].
\end{aligned}
\end{equation}

Here,
$f_{B}^{1} = T(P^{\otimes n},\, Q_{K}^{\otimes (n-1)} \otimes Q_{K}^{\mu_{B}^{1}})$
and
$f_{B}^{2} = T(P^{\otimes n},\, Q_{K}^{\otimes (n-1)} \otimes Q_{K}^{\mu_{B}^{2}})$,
where \(P\), \(Q_{K}\), and \(Q_{K}^{\mu}\) are defined in Table~\ref{tab:notation}, and
$\mu_{B}^{1}=\|-\sum_{i\in B}\bar g_i\|_2/((K+1)C\sigma)$,
$\mu_{B}^{2}=\|K\bar g_{N+1}-\sum_{i\in B}\bar g_i\|_2/((K+1)C\sigma)$.
\end{theorem}
\begin{proof}
   The theorem follows directly by substituting $K$ for $N$ in the $f$ defined in Eq.~\ref{eq:lower_bound_EASGM_high_dimentsion}.
\end{proof}

Following the same proof procedure as in Section~\ref{sec:analysis}, we can establish the privacy guarantee for any neighboring pair of datasets with $|D|=N$ and $|D'|=N+1$, as follows.

\begin{theorem}[Privacy guarantee of FEASGM] \label{thm:guarantee_of_FEASGM}
Let \(D=\{d_1,d_2,\dots,d_N\}\) and \(D' = D \cup \{d_{N+1}\}\) be two neighboring datasets under the add/remove notion of adjacency, and let \(q\) be a fixed sampling ratio. If
$\lfloor N\cdot q \rfloor = \lfloor (N+1)\cdot q \rfloor$,
then the privacy guarantee of FEASGM reduces to  SGM guarantee in \cite{bu2020deep} which is  $q\,G_{1/\sigma} + (1-q)\mathrm{ID}$,
where \(\mathrm{ID}\) and \(G_{1/\sigma}\) denote the trade-off functions defined in Table~\ref{tab:notation}. Otherwise, if $\lfloor N\cdot q \rfloor = K$
$and$
$\lfloor (N+1)\cdot q \rfloor = K+1,$
then the privacy guarantee is
\begin{equation}\label{eq:FEASGM_High}
\begin{aligned}
  \underline{f}(\alpha(t,c))
  \;=\;&
  \sum_{B} P(B)\,[
      (1-q)\,\underline{f_{B}^{1}}(\alpha_{B}^{1}(t,c))\\
      &\qquad\qquad
      + q\,\underline{f_{B}^{2}}(\alpha_{B}^{2}(t,c))
    ].
\end{aligned}
\end{equation}

\noindent Here,
$\underline{f_{B}^{1}} = T(P^{\otimes n},\, Q_{K}^{\otimes (n-1)} \otimes Q_{K}^{\overline{\mu}_{B}^{1}})$
and
$\underline{f_{B}^{2}} = T(P^{\otimes n},\, Q_{K}^{\otimes (n-1)} \otimes Q_{K}^{\overline{\mu}_{B}^{2}})$,
where \(P\), \(Q_{K}\), and \(Q_{K}^{\mu}\) are defined in Table~\ref{tab:notation}, and
$\overline{\mu}_{B}^{1}=\tfrac{|B|}{(K+1)\sigma},
\qquad
\overline{\mu}_{B}^{2}=\tfrac{K+|B|}{(K+1)\sigma}$.
\end{theorem}

\begin{proof}
If $\lfloor N\cdot q \rfloor = \lfloor (N+1)\cdot q \rfloor$,
then both neighboring datasets are normalized by the same factor. Therefore, the normalization step is identical across the neighboring pair and can be treated as post-processing, so the privacy guarantee reduces directly to that of SGM in~\cite{bu2020deep}. If instead
$\lfloor N\cdot q \rfloor \neq \lfloor (N+1)\cdot q \rfloor = K+1$,
then the neighboring datasets are normalized by two different factors, \(K\) and \(K+1\). In this case, the guarantee follows by substituting \(K\) into the EASGM guarantee derived in Theorem~\ref{theorem:EASGM_HIGH}.
\end{proof}

Similarly to estimate the EASGM guarantee using upperbound, we can derive a analytic upperbound $\overline{f}$ as follow. 

\begin{theorem}\label{thm:upperbound_FEASGM}
Let \(f\) be the privacy trade-off function defined in Eq.~\ref{eq:FEASGM_High_per_pair}, and let \(P\) and \(Q_K\) be the distributions defined in Table~\ref{tab:notation}, where \(K=\lfloor Nq \rfloor\). Then the following inequality holds:
\begin{equation}\label{eq:FEASGM_upper}
    \overline{f}
    =
    T(P^{\otimes (n-1)}, Q_K^{\otimes (n-1)})
    \ge f.
\end{equation}
\end{theorem}

\begin{proof}
The proof follows directly from that of Theorem~\ref{thm:upperbound_EASGM} by replacing \(N\) with \(K\).
\end{proof}
According to Theorem~\ref{theorem:EASGM_flaws}, when $n$ is large, $\overline{f}$ can approach $0$. This implies that the privacy guarantee of FEASGM can be significantly weaker than that of SGM for certain pairs of neighboring datasets satisfying $\lfloor N\cdot q \rfloor \neq \lfloor (N+1)\cdot q \rfloor$. The theorem \ref{Thm:FEASGM_Upperbound_T_round} extends Theorem~\ref{thm:upperbound_FEASGM} to the $T$-round setting and gives the upper bound $\overline{f}^{\otimes T}$.

\begin{theorem}[Comparison composition theorem]\label{thm:comparison_theorem}
Let $\mathcal M_1 : X \times I_1 \to Y_1$ be the first mechanism, and let $\mathcal M_2 : X \times Y_1 \times I_2 \to Y_2$ be the second mechanism, where $I_1$ and $I_2$ are independent random variables following distribution $w$. Define the joint mechanism $\mathcal M : X \to Y_1 \times Y_2$ by $\mathcal M(D) = (Y_1, Y_2)$, where $Y_1 = \mathcal M_1(D, I_1)$ and $Y_2 = \mathcal M_2(D, Y_1, I_2)$. Define the transparent mechanism $\mathcal M_I : X \to Y_1 \times I_1 \times Y_2 \times I_2$ by $\mathcal M_I(D) = (Y_1, I_1, Y_2, I_2)$. Also define $\mathcal M_{1,I_1}(D) = (Y_1, I_1)$ and $\mathcal M_{2,I_2}(D, y_1) = (\mathcal M_2(D, y_1, I_2), I_2)$.

Then, for any neighboring datasets $D$ and $D'$, we have $T(\mathcal M(D), \mathcal M(D')) \ge T(\mathcal M_I(D), \mathcal M_I(D'))$.

Moreover, let $g_1$ and $g_2$ be two trade-off functions. If $g_1 \ge T(\mathcal M_{1,I_1}(D), \mathcal M_{1,I_1}(D'))$ and, for every fixed $y_1$, $g_2 \ge T(\mathcal M_{2,I_2}(D, y_1), \mathcal M_{2,I_2}(D', y_1))$, then $g_1 \otimes g_2 \ge T(\mathcal M_I(D), \mathcal M_I(D'))$.
\end{theorem}
\begin{proof}
The first inequality follows directly from the post-processing property. Indeed, $\mathcal M$ is obtained from $\mathcal M_I$ by hiding the released indices $I_1$ and $I_2$, and hence $T(\mathcal M(D), \mathcal M(D')) \ge T(\mathcal M_I(D), \mathcal M_I(D'))$.

We now prove the second claim using Lemma 11 and Lemma 12 in~\cite{dong2019gaussian}. Let $A = \mathcal M_{1,I_1}(D)$ and $A' = \mathcal M_{1,I_1}(D')$. By the definition of trade-off functions, choose distributions $P,Q,P',Q'$ such that $g_1 = T(P,Q)$ and $g_2 = T(P',Q')$.

Since $g_1 \ge T(A, A')$, Lemma 12 gives $g_1 \otimes g_2 = T(P \times P', Q \times Q') \ge T(A \times P', A' \times Q')$. Therefore, it remains to show that $T(A \times P', A' \times Q') \ge T(\mathcal M_I(D), \mathcal M_I(D'))$.

To this end, let the input space be the common measurable space of $A$ and $A'$. For each input $a = (y_1,i_1)$, define four randomized algorithms as follows: $K_1(a) = P'$, $K_1'(a) = Q'$, $K_2(a) = \mathcal M_{2,I_2}(D, y_1)$, and $K_2'(a) = \mathcal M_{2,I_2}(D', y_1)$. Note that $K_1$ and $K_1'$ ignore the input and simply output random variables distributed as $P'$ and $Q'$, respectively. Hence, for every input $a = (y_1,i_1)$, we have $T(K_1(a), K_1'(a)) = T(P',Q') = g_2$. By assumption, for every $y_1$, we also have $g_2 \ge T(\mathcal M_{2,I_2}(D, y_1), \mathcal M_{2,I_2}(D', y_1)) = T(K_2(a), K_2'(a))$. Therefore, $T(K_1(a), K_1'(a)) \ge T(K_2(a), K_2'(a))$ holds for every input $a$.

Applying Lemma 11 with input distributions $A$ and $A'$, we obtain $T(A \times P', A' \times Q') \ge T\bigl((A, K_2(A)), (A', K_2'(A'))\bigr)$. By construction, $(A, K_2(A))$ is exactly $\mathcal M_I(D)$, and $(A', K_2'(A'))$ is exactly $\mathcal M_I(D')$. Hence, $T(A \times P', A' \times Q') \ge T(\mathcal M_I(D), \mathcal M_I(D'))$.

Combining this inequality with the previous bound from Lemma 12 yields $g_1 \otimes g_2 \ge T(\mathcal M_I(D), \mathcal M_I(D'))$. This completes the proof.
\end{proof}

\begin{theorem}\label{Thm:FEASGM_Upperbound_T_round}
Let $\overline{f}$ be as defined in Theorem~\ref{thm:upperbound_FEASGM}. Let $\mathcal M_I$ denote the transparent per-round mechanism whose trade-off function is upper bounded by $\overline{f}$, and let $\mathcal M$ denote the corresponding FEASGM mechanism obtained by hiding the additional released information $I$. For each $T \in \mathbb Z_{>0}$, let $\mathcal M_I^{\otimes T}$ and $\mathcal M^{\otimes T}$ denote the $T$-round compositions of $\mathcal M_I$ and $\mathcal M$, respectively. Then, for any neighboring datasets $D$ and $D'$, we have
\[
T(\mathcal M^{\otimes T}(D), \mathcal M^{\otimes T}(D'))
\ge
T(\mathcal M_I^{\otimes T}(D), \mathcal M_I^{\otimes T}(D'))
\]
and
\[
\overline{f}^{\otimes T}
\ge
T(\mathcal M_I^{\otimes T}(D), \mathcal M_I^{\otimes T}(D')).
\]
We refer to \(T(\mathcal M_I^{\otimes T}(D), \mathcal M_I^{\otimes T}(D'))\) as the transparent per-pair \(T\)-round privacy guarantee of FEASGM. Consequently, \(\overline{f}^{\otimes T}\) is an upper bound on this transparent per-pair guarantee.
\end{theorem}

\begin{proof}
The first inequality again follows from the post-processing property. Indeed, $\mathcal M^{\otimes T}$ is obtained from $\mathcal M_I^{\otimes T}$ by hiding the additional information $I$ released in each round. Therefore, $T(\mathcal M^{\otimes T}(D), \mathcal M^{\otimes T}(D')) \ge T(\mathcal M_I^{\otimes T}(D), \mathcal M_I^{\otimes T}(D'))$.

We now prove $\overline{f}^{\otimes T} \ge T(\mathcal M_I^{\otimes T}(D), \mathcal M_I^{\otimes T}(D'))$ by induction on $T$. For the base case $T=1$, the claim follows directly from the assumption that $\overline{f} \ge T(\mathcal M_I(D), \mathcal M_I(D'))$. Assume now that the claim holds for some $T \ge 1$, namely, $\overline{f}^{\otimes T} \ge T(\mathcal M_I^{\otimes T}(D), \mathcal M_I^{\otimes T}(D'))$. Consider the $(T+1)$-round composition. View $\mathcal M_I^{\otimes (T+1)}$ as a two-stage mechanism: the first stage consists of the first $T$ transparent rounds, and the second stage consists of the $(T+1)$-st transparent round. The second stage may depend on the previous transparent transcript, but by Theorem~\ref{thm:upperbound_FEASGM}, for every fixed previous transcript, its trade-off function is still upper bounded by $\overline{f}$. Therefore, the first stage is upper bounded by $\overline{f}^{\otimes T}$ by the induction hypothesis, and the second stage is upper bounded by $\overline{f}$ uniformly over every fixed previous transcript.

Applying Theorem~\ref{thm:comparison_theorem} with $g_1 = \overline{f}^{\otimes T}$ and $g_2 = \overline{f}$, we obtain $\overline{f}^{\otimes (T+1)} = (\overline{f}^{\otimes T}) \otimes \overline{f} \ge T(\mathcal M_I^{\otimes (T+1)}(D), \mathcal M_I^{\otimes (T+1)}(D'))$.

Thus, by induction, $\overline{f}^{\otimes T} \ge T(\mathcal M_I^{\otimes T}(D), \mathcal M_I^{\otimes T}(D'))$ holds for every $T \in \mathbb Z_{>0}$. This completes the proof.
\end{proof}
\subsection{Visualizing the Privacy Guarantee of FEASGM}
Using the same notation as in Theorem~\ref{Thm:FEASGM_Upperbound_T_round}, let \(T(\mathcal M^{\otimes T}(D), \mathcal M^{\otimes T}(D'))\) denote the per-pair privacy leakage of \(T\)-round training, let \(T(\mathcal M_I^{\otimes T}(D), \mathcal M_I^{\otimes T}(D'))\) denote the corresponding \(T\)-round per-pair privacy guarantee, and let \(\overline{f}^{\otimes T}\) denote the upper bound on the \(T\)-round guarantee. Following the same approach used to visualize the privacy guarantee of EASGM in Section~\ref{sec:upperbound_visulization}, we estimate \(\overline{f}^{\otimes T}\). This estimate upper bounds \(T(\mathcal M_I^{\otimes T}(D), \mathcal M_I^{\otimes T}(D'))\), thereby enabling us to characterize the privacy guarantee of \(T\)-round FEASGM-based DP-SGD training. The following theorem formalizes this result.

\begin{theorem}\label{thm:upperbound_estimate_T_round}
Let \(f_{\ast} = T(P, Q_{K})\), where \(P\) and \(Q_{K}\) are defined in Table~\ref{tab:notation}, and \(K=\lfloor Nq \rfloor\). Let \(\overline{f_{\ast}} = \max\{f_{\ast}, f_{\ast}^{-1}\}\). Then, for \(\overline{f}\) defined in Theorem~\ref{Thm:FEASGM_Upperbound_T_round}, the following bound holds for all \(\alpha \in [\gamma,\, 1-\gamma]\):
\[
\overline{f}^{\otimes T}(\alpha) \le G_{\mu}(\alpha - \gamma)+\gamma.
\]
Here, \(\mu\) and \(\gamma\) are derived from the composition of \(T\cdot (n-1)\) independent copies of \(\overline{f_{\ast}}\).
\end{theorem}

\begin{proof}
The proof follows the same argument as that of Theorem~\ref{thm:upperbound_estimate} in Appendix~\ref{App:proof:estimationg}, with \(N\) replaced by \(K=\lfloor Nq \rfloor\).
\end{proof}

\subsection{Auditing Opacus's FEASGM Implementation}\label{sec:Opacus_auditing}

We now audit the FEASGM implementation. Our goal is to examine the weak regime and show that, when the output dimension is high, the privacy leakage of FEASGM can significantly exceed the SGM guarantee. We follow the same experimental setting used for auditing the EASGM and ASGM implementations in Section~\ref{sec:Experiment_set_up}.

\noindent \textbf{Audited algorithm.} 
We use the implementation of DP-FETA, as it is based on Opacus v1.0 (see \cite{dpfeat_opacus_version}).

\noindent \textbf{Parameter setting.} 
We use the same setting as that for EASGM in Section~\ref{sec:Experiment_set_up}. We focus on the weak regime and set $q=0.5$, $|D|=3$, $|D'|=4$, $\sigma=10$, and $C=0.1$. Under this setting, the normalization factors differ across the neighboring datasets.

\noindent \textbf{Dataset.} 
We use the MNIST dataset for auditing. Where we randomly select $4$ data point to constitute the neighboring datasets.

\noindent \textbf{Baseline.} 
Following Section~\ref{sec:Experiment_set_up}, we compare our results with the privacy guarantee in~\cite{bu2020deep}.

\noindent \textbf{Results.} 
Fig.~\ref{fig:FEASGM_IMPLEMENTATION} shows that the audited privacy leakage of the FEASGM implementation exceeds the level implied by the claimed SGM privacy guarantee.
\subsection{Visualizing the Privacy Guarantee of Opacus's FEASGM Implementation}\label{sec:visualizing_FEASGM}
We now visualize the privacy guarantee of FEASGM in the real-world setting with a low sample rate and large datasets. The case $\lfloor N \cdot q\rfloor=\lfloor (N+1) \cdot q\rfloor$ is not our focus, as it follows the SGM guarantee. Instead, we focus on the case $\lfloor N \cdot q\rfloor\neq \lfloor (N+1) \cdot q\rfloor$.
\noindent \textbf{Setting.} We follow the setting in Section~\ref{sec:RQ3}. Specifically, we set the number of training rounds to $T=200$ and consider neighboring datasets with $|D^-|=|D|=140199$ and $|D'|=140200$. We evaluate two model sizes: a 26,010-parameter model used by DP-SGD-HF and DP-SUR on MNIST, and a 1,485,217-parameter model used by DP-FETA and PrivImage on MNIST. We set the sampling ratio to $q=0.005$, since the Opacus FEASGM implementation uses $q=1/T$. We set $\sigma=10$ and use Theorem~\ref{thm:upperbound_estimate_T_round} to estimate the upper bound of the privacy guarantee.

\noindent \textbf{Baseline.} To compare the multi-round privacy guarantee of FEASGM under our estimated guarantee with that under the SGM privacy guarantee, we adopt Opacus's RDP privacy accountant \cite{dpfeat_rdp_accountant}, which is also used by DP-FETA in \cite{11023490}. 

\noindent \textbf{Results.} The results are shown in Fig. \ref{fig:FEASGM_Theoretical}. It shows that FEASGM's guarantee in multi-round training can be weaker than the claimed guarantee.

\begin{figure}[htbp]
    \centering
\includegraphics[width=0.4\textwidth]{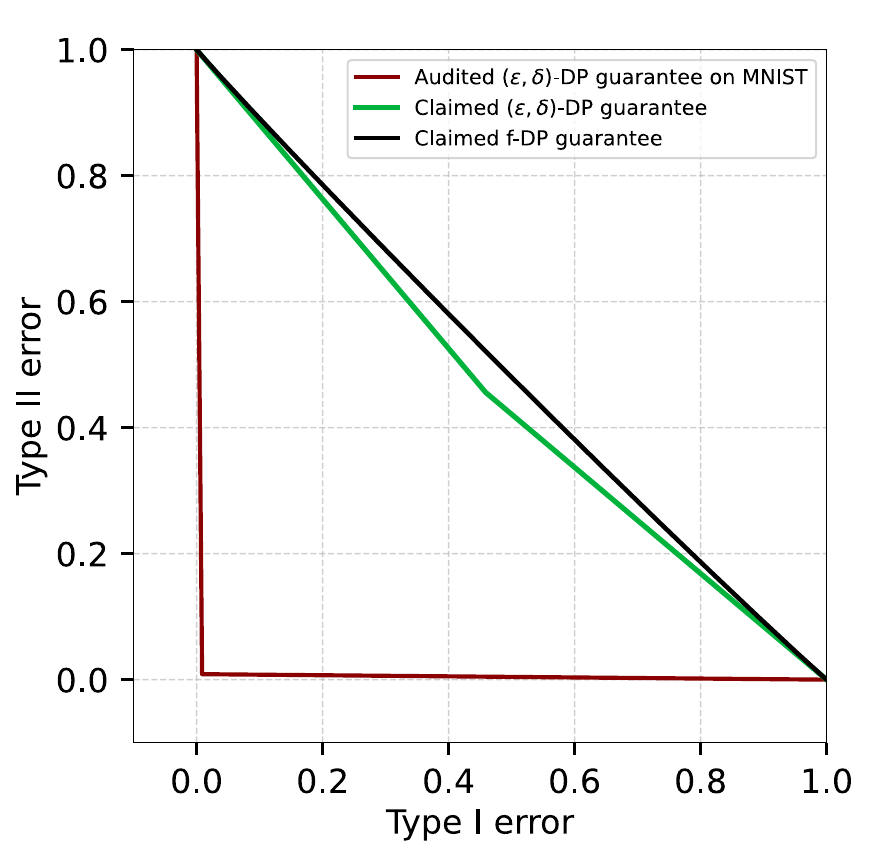}
 \caption{Audited privacy leakage of the real-world Opacus FEASGM implementation in DP-FETA, computed using the $(\epsilon,\delta)$-trade-off function with $\delta = 10^{-5}$, and compared with the claimed privacy guarantees.}
    \label{fig:FEASGM_IMPLEMENTATION}
\end{figure}

\begin{figure*}[htbp]
    \centering
\includegraphics[width=0.8\textwidth]{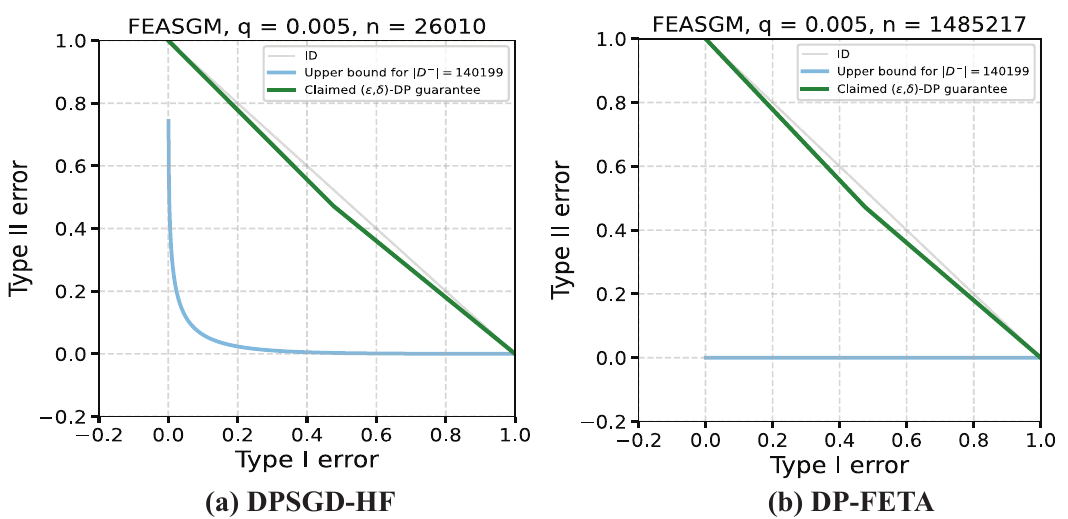}
    \caption{Theoretical upper bounds on the 200-round privacy guarantees of the real-world Opacus FEASGM implementations, showing that each privacy guarantee lies below its corresponding upper bound.
}
    \label{fig:FEASGM_Theoretical}
\end{figure*}

\begin{table*}[t]
\centering
\small
\setlength{\tabcolsep}{6pt}
\renewcommand{\arraystretch}{1.15}
\caption{Overview of the audited \texttt{Opacus} versions, focusing on \texttt{privacy\_engine.py}, together with the associated commit hashes for reproducibility and re-auditing. The six columns report: (1) the audited \texttt{Opacus} version(s); (2) the GitHub link to the corresponding \texttt{privacy\_engine.py} file; (3) the associated commit hash; (4) the normalization operation used under Poisson sampling, where \textbf{A} denotes averaging over the realized batch size $|B|$, \textbf{EA} denotes averaging over the expected batch size $N\!\cdot\!q$, and \textbf{FEA} denotes averaging over $\lfloor N\!\cdot\!q \rfloor$; (5) whether the version provides built-in support for Poisson sampling (e.g., through a dedicated interface such as \texttt{poisson\_sampling=True}, as in v1.5.4); and (6) the overall model-update mechanism used when Poisson sampling is enabled and the \texttt{loss\_reduction} is set to \texttt{mean}. Here, \textbf{Not supported} means that Poisson-based data selection cannot be directly applied because the implementation imposes an upper bound on the batch size, whereas Poisson sampling produces variable-size batches that may exceed this bound. \textbf{ASGM}, \textbf{EASGM}, and \textbf{FEASGM} are defined in Definitions~\ref{def:ASGM},~\ref{def:EASGM}, and~\ref{def:FEASGM}, respectively. Opacus also supports standard SGM. In v1.5.4, when \texttt{loss\_reduction="sum"}, the gradient sum is not normalized before the parameter update, making this code path consistent with the standard SGM update rule~\cite{opacus_sgm_implementation_v154_line488}.}
\label{tab:audited_libraries}
\begin{tabular}{p{2cm}p{2cm}p{2cm}p{2cm}p{4cm}p{2cm}}
\toprule
\textbf{Version} & \textbf{Link} & \textbf{Commit Hash} & \textbf{Normalization} & \textbf{Built-in Poisson update logic} & \textbf{Mechanism} \\
\midrule
v0.9.0-v0.9.1                    & \href{https://github.com/meta-pytorch/opacus/blame/v0.9.1/opacus/privacy_engine.py#L241}{GitHub} & 33a7ca4 & A &  $\times$ & Not supported\\

   v0.10.0-v0.10.1             & \href{https://github.com/meta-pytorch/opacus/blame/v0.10.0/opacus/privacy_engine.py#L290}{GitHub} & 9d76ea5 & A & $\times$ & Not supported\\
   v0.11.0             & \href{https://github.com/meta-pytorch/opacus/blame/v0.11.0/opacus/privacy_engine.py#L301}{GitHub} & 9d76ea5 & A & $\times$ & Not supported\\
   v0.12.0             & \href{https://github.com/meta-pytorch/opacus/blame/v0.12.0/opacus/privacy_engine.py#L365}{GitHub} & 9d76ea5 & A & $\times$ & ASGM\\
    v0.13.0             & \href{https://github.com/meta-pytorch/opacus/blame/v0.13.0/opacus/privacy_engine.py#L365}{GitHub} & 9d76ea5 & A & $\times$ & ASGM\\
    v0.14.0             & \href{https://github.com/meta-pytorch/opacus/blame/v0.14.0/opacus/privacy_engine.py#L485}{GitHub} & e9983ec & EA & $\checkmark$ & EASGM\\
    v0.15.0             & \href{https://github.com/meta-pytorch/opacus/blame/v0.15.0/opacus/privacy_engine.py#L523}{GitHub} & e9983ec & EA &$\checkmark$ & EASGM\\
 v1.0.0-v1.0.2             & \href{https://github.com/meta-pytorch/opacus/blame/v1.0.2/opacus/privacy_engine.py#L342C10-L342C10}{GitHub} & 1241e62 & FEA & $\checkmark$&  FEASGM\\

 v1.1.0-v1.1.3             & \href{https://github.com/meta-pytorch/opacus/blame/v1.1.3/opacus/privacy_engine.py#L390}{GitHub} & 1241e62 & FEA & $\checkmark$& FEASGM\\

 v1.2.0             & \href{https://github.com/meta-pytorch/opacus/blame/v1.2.0/opacus/privacy_engine.py#L411}{GitHub} & 1241e62 & FEA &$\checkmark$& FEASGM\\

 v1.3             & \href{https://github.com/meta-pytorch/opacus/blame/v1.3/opacus/privacy_engine.py#L412C6-L412C6}{GitHub} & 1241e62 & FEA& $\checkmark$ & FEASGM\\

 v1.4-v1.4.1           & \href{https://github.com/meta-pytorch/opacus/blame/v1.4.1/opacus/privacy_engine.py#L365C26-L365C26}{GitHub} & 1241e62 & FEA &$\checkmark$&  FEASGM\\

     v1.5-v1.5.4         & \href{https://github.com/meta-pytorch/opacus/blame/v1.5.4/opacus/privacy_engine.py#L407C17-L407C17}{GitHub} & 1241e62 & FEA & $\checkmark$ & FEASGM\\

\bottomrule
\end{tabular}
\end{table*}

\section{Multiround Guarantees for EASGM and ASGM}\label{sec:Multi_EASGM_ASGM_GUARANTEE}

For both EASGM and ASGM, we can combine the per-round upper bounds in
Theorems~\ref{thm:upperbound_EASGM} and~\ref{thm:upperbound_ASGM}
with the comparison theorem (Theorem~\ref{thm:comparison_theorem})
to obtain valid upper bounds for their $T$-round compositions.

\begin{theorem}[Multiround upper bound for EASGM and ASGM]
\label{thm:EASGM_ASGM_Upperbound_T_round}
Consider either EASGM or ASGM, and let $\overline{f}$ denote the corresponding
per-round upper bound given in Theorem~\ref{thm:upperbound_EASGM} or
Theorem~\ref{thm:upperbound_ASGM}. Let $\mathcal{M}_{I}$ denote the transparent
per-round mechanism, and let $\mathcal{M}$ denote the corresponding hidden
mechanism obtained by suppressing the additional released information~$I$.
Assume that the randomness is independent across rounds. For each
$T \in \mathbb{Z}_{>0}$, let $\mathcal{M}_{I}^{\otimes T}$ and
$\mathcal{M}^{\otimes T}$ denote the $T$-round compositions of
$\mathcal{M}_{I}$ and $\mathcal{M}$, respectively. Then, for any neighboring
datasets $D$ and $D'$, we have
\[
T\!\left(\mathcal{M}^{\otimes T}(D),\, \mathcal{M}^{\otimes T}(D')\right)
\ge
T\!\left(\mathcal{M}_{I}^{\otimes T}(D),\, \mathcal{M}_{I}^{\otimes T}(D')\right),
\]
and
\[
\overline{f}^{\otimes T}
\ge
T\!\left(\mathcal{M}_{I}^{\otimes T}(D),\, \mathcal{M}_{I}^{\otimes T}(D')\right).
\]
Consequently, $\overline{f}^{\otimes T}$ is a valid upper bound on the
$T$-round privacy guarantee of the corresponding EASGM or ASGM mechanism.
\end{theorem}

\begin{proof}
By Theorem~\ref{thm:comparison_theorem}, in each round the hidden mechanism
$\mathcal{M}$ is at least as private as its transparent counterpart
$\mathcal{M}_{I}$, and the trade-off function of $\mathcal{M}_{I}$ is upper
bounded by $\overline{f}$. Composing these per-round relations over $T$
independent rounds yields
\[
T\!\left(\mathcal{M}^{\otimes T}(D),\, \mathcal{M}^{\otimes T}(D')\right)
\ge
T\!\left(\mathcal{M}_{I}^{\otimes T}(D),\, \mathcal{M}_{I}^{\otimes T}(D')\right),
\]
and
\[
\overline{f}^{\otimes T}
\ge
T\!\left(\mathcal{M}_{I}^{\otimes T}(D),\, \mathcal{M}_{I}^{\otimes T}(D')\right).
\]
This proves the claim.
\end{proof}

For EASGM, this abstract multiround upper bound can be further approximated
in closed form by using the same argument as in
Theorem~\ref{Thm:FEASGM_Upperbound_T_round}.

\begin{corollary}
\label{cor:EASGM_upperbound_estimate_T_round}
Let $\overline{f}$ denote the per-round upper bound for EASGM given in
Theorem~\ref{thm:upperbound_EASGM}. Define
\[
f_{\ast} = T(P, Q_{N}),
\qquad
\overline{f_{\ast}} = \max\{f_{\ast}, f_{\ast}^{-1}\},
\]
where $P$ and $Q_{N}$ are defined in Table~\ref{tab:notation}. Then, using
the same Gaussian approximation as in
Theorem~\ref{Thm:FEASGM_Upperbound_T_round}, for all
$\alpha \in [\gamma,\, 1-\gamma]$,
\[
\overline{f}^{\otimes T}(\alpha)
\le
G_{\mu}(\alpha-\gamma)+\gamma,
\]
where $\mu$ and $\gamma$ are obtained from the composition of
$T(n-1)$ copies of $\overline{f_{\ast}}$.
\end{corollary}

\begin{proof}
The proof follows the same argument as that of
Theorem~\ref{Thm:FEASGM_Upperbound_T_round}, after replacing the FEASGM
per-round upper bound with the EASGM per-round upper bound in
Theorem~\ref{thm:upperbound_EASGM}. We therefore omit the repeated details.
\end{proof}

For ASGM, Theorem~\ref{thm:EASGM_ASGM_Upperbound_T_round} still provides a
valid $T$-round upper bound $\overline{f}^{\otimes T}$. However, the per-round
upper bound in Theorem~\ref{thm:upperbound_ASGM} contains the term
\[
T\!\left(P^{\otimes (n-1)},\, Q_{N}^{\otimes (n-1)}\right),
\]
for which we do not derive a closed-form analytic expression in this paper.
Therefore, unlike EASGM, we do not provide a closed-form visualization of the
multiround ASGM guarantee. Obtaining such a visualization would require a
numerical approximation of this composition term, for example using the method
in~\cite{gopi2021numerical}, which we leave to future work.

\subsection{Impact of One-Round Unsoundness on Multi-Round Guarantees}
In differential privacy, a multiround guarantee obtained by composition is
sound only if the per-round guarantee being composed is itself sound for the
underlying mechanism~\cite{dong2019gaussian}. Therefore, if the one-round SGM guarantee does not correctly characterize the actual EASGM or ASGM mechanism, then the corresponding multi-round SGM-based estimate is not soundly justified for that mechanism. This observation does not, by itself, establish a numerical violation of the multi-round estimate; rather, it shows that the multi-round SGM guarantee lacks a valid per-round foundation.
\section{Auditing Algorithm}\label{Appendix:auditing_algorithm}
\begin{algorithm}[H]\small
\caption{Tight auditing of EASGM/ASGM}
\label{alg:score-function}
\begin{algorithmic}[1]
\Require Neighboring datasets $D=\{d_{1},...,d_{N}\},\; D'=D\cup \{d_{N+1}\}$; statistical test $\mathrm{Test}(\cdot)$; intermediate distinguisher set $\mathcal{A}=\{A_{B}^{i}=(\mathrm{Test},\mathrm{Proc}_{P_{B}^{i}},\mathrm{Proc}_{Q_{B}^{i}})\}$; final distinguisher $A_{\ast}$; auditing time $t$; audited mechanism $\mathcal{M}(\cdot)$.
\Ensure Empirical type I / type II errors $(\alpha_{r},\beta_{r})$.
\State $Type_{I_c}\gets 0$, $Type_{I_{ic}}\gets 0$, $Type_{II_c}\gets 0$, $Type_{II_{ic}}\gets 0$
\For{$r=1$ \textbf{to} $t$}
    \State $D_{\text{in}}\gets \varnothing$
    \State Randomly sample $b \in \{0,1\}$
    \If{$b == 0$}
        \State $D_{\text{in}} \gets D$
    \Else
        \State $D_{\text{in}} \gets D'$
    \EndIf
    \State $O \gets \mathcal{M}(D_{\text{in}})$;\quad $Score \gets \varnothing$
    \ForAll{$A_{B}^{i} \in \mathcal{A}$}
        \State $Score \gets (Score, A_{B}^{i}(O))$
    \EndFor
    \State $b_{guess} \gets A_{\ast}(Score)$
    \If{$b==0$ and $b==b_{guess}$}
        \State $Type_{I_c}++$
    \ElsIf{$b==0$ and $b\ne b_{guess}$}
        \State $Type_{I_{ic}}++$
    \ElsIf{$b==1$ and $b==b_{guess}$}
        \State $Type_{II_c}++$
    \ElsIf{$b==1$ and $b\ne b_{guess}$}
        \State $Type_{II_{ic}}++$
    \EndIf
\EndFor
\State Compute $(\alpha_{r},\beta_{r})$ using the Clopper--Pearson method \\

\Return $(\alpha_{r},\,\beta_{r})$
\end{algorithmic}
\end{algorithm}

\section{Proof}

\subsection{Proof of Lemma~\ref{Lemma:min_equal}}
\label{Proof_Lemma_mim_equal}

\begin{proof}
Let
\[
E_{\alpha}
=
\left\{
(\alpha_1,\ldots,\alpha_m)\in[0,1]^m:
\sum_{i=1}^m w_i\alpha_i=\alpha
\right\}.
\]
We prove that
\[
T((P_I\mid I,I),(Q_I\mid I,I))(\alpha)
=
\min_{(\alpha_1,\ldots,\alpha_m)\in E_{\alpha}}
\sum_{i=1}^m w_i\,T(P_i,Q_i)(\alpha_i).
\]

Let
\[
(\alpha_1^\ast,\ldots,\alpha_m^\ast)
\in
\arg\min_{(\alpha_1,\ldots,\alpha_m)\in E_{\alpha}}
\sum_{i=1}^m w_i\,T(P_i,Q_i)(\alpha_i).
\]

We first prove the ``\(\ge\)'' direction. Let \(\phi\) be an optimal decision rule for
\[
H_0:(P_I\mid I,I)
\qquad\text{vs.}\qquad
H_1:(Q_I\mid I,I)
\]
with Type~I error \(\alpha\). For each \(i\), let \(\phi^i\) denote the restriction of \(\phi\) to the slice \(I=i\), and let \(\alpha_i\) and \(\beta_i\) be the Type~I and Type~II errors of \(\phi^i\), respectively. Then
\[
\alpha=\sum_{i=1}^m w_i\alpha_i,
\qquad
T((P_I\mid I,I),(Q_I\mid I,I))(\alpha)
=
\sum_{i=1}^m w_i\beta_i.
\]
Moreover, for each \(i\), \(\phi^i\) must be optimal for testing
\[
H_0:P_i
\qquad\text{vs.}\qquad
H_1:Q_i
\]
among all tests with Type~I error \(\alpha_i\). Otherwise, replacing \(\phi^i\) by another test with the same Type~I error \(\alpha_i\) and a strictly smaller Type~II error would yield a global test with the same total Type~I error \(\alpha\) but a smaller total Type~II error, contradicting the optimality of \(\phi\). Hence,
\[
\beta_i=T(P_i,Q_i)(\alpha_i),
\qquad \forall i\in[m].
\]
Therefore,
\[
T((P_I\mid I,I),(Q_I\mid I,I))(\alpha)
=
\sum_{i=1}^m w_iT(P_i,Q_i)(\alpha_i).
\]
Since \((\alpha_1,\ldots,\alpha_m)\in E_{\alpha}\), by the definition of
\((\alpha_1^\ast,\ldots,\alpha_m^\ast)\), we have
\[
\sum_{i=1}^m w_iT(P_i,Q_i)(\alpha_i)
\ge
\sum_{i=1}^m w_iT(P_i,Q_i)(\alpha_i^\ast).
\]
Thus,
\[
T((P_I\mid I,I),(Q_I\mid I,I))(\alpha)
\ge
\sum_{i=1}^m w_iT(P_i,Q_i)(\alpha_i^\ast).
\]

Next, we prove the ``\(\le\)'' direction. For each \(i\), let \(\bar{\phi}^i\) be an optimal decision rule for
\[
H_0:P_i
\qquad\text{vs.}\qquad
H_1:Q_i
\]
with Type~I error \(\alpha_i^\ast\), so that its Type~II error is
\(T(P_i,Q_i)(\alpha_i^\ast)\). Define a global decision rule \(\bar{\phi}\) by applying \(\bar{\phi}^i\) on the slice \(I=i\). Then the Type~I error of \(\bar{\phi}\) is
\[
\sum_{i=1}^m w_i\alpha_i^\ast=\alpha,
\]
and its Type~II error is
\[
\sum_{i=1}^m w_iT(P_i,Q_i)(\alpha_i^\ast).
\]
Since \(T((P_I\mid I,I),(Q_I\mid I,I))(\alpha)\) is the minimum achievable Type~II error among all tests with Type~I error \(\alpha\), we obtain
\[
T((P_I\mid I,I),(Q_I\mid I,I))(\alpha)
\le
\sum_{i=1}^m w_iT(P_i,Q_i)(\alpha_i^\ast).
\]

Combining the two inequalities yields
\[
T((P_I\mid I,I),(Q_I\mid I,I))(\alpha)
=
\sum_{i=1}^m w_iT(P_i,Q_i)(\alpha_i^\ast),
\]
which proves the claim.
\end{proof}
\subsection{Proof of Lemma~\ref{Lemma:Order_in_trade-off}}
\label{Proof:Lemma_Order_in_trade-off}

\begin{proof}
Let
\[
E_{\alpha}
=
\left\{
(\alpha_{1},\ldots,\alpha_{m})\in[0,1]^m:
\sum_{i=1}^{m}w_i\alpha_i=\alpha
\right\}.
\]
By Lemma~\ref{Lemma:min_equal}, we have
\[
T((P_{I}\mid I,I),(Q_{I}\mid I,I))(\alpha)
=
\min_{(\alpha_{1},\dots,\alpha_{m})\in E_{\alpha}}
\sum_{i=1}^{m}w_{i}T(P_{i},Q_{i})(\alpha_{i}).
\]
Since \(E_{\alpha}\) is compact and the component trade-off functions are continuous, the minimum is attained. Let \((\alpha_{1}^{\ast},\dots,\alpha_{m}^{\ast})\in E_{\alpha}\) be a minimizer.

For each \(i\in[m]\), by the Neyman--Pearson lemma, allowing randomized rejection rules if necessary, there exists an optimal rejection rule \(\phi^{i}\colon(x,i)\mapsto[0,1]\) for
\[
H_{0}\!:\!(P_{i},i)
\qquad\text{vs.}\qquad
H_{1}\!:\!(Q_{i},i)
\]
with Type~I error \(\alpha_i^\ast\). Define the global rejection rule \(\phi\) by setting \(\phi(x,i)=\phi^{i}(x,i)\). Then \(\phi\) has Type~I error \(\sum_{i=1}^{m}w_i\alpha_i^\ast=\alpha\), and its Type~II error is
\[
\sum_{i=1}^{m}w_iT(P_i,Q_i)(\alpha_i^\ast)
=
T((P_{I}\mid I,I),(Q_{I}\mid I,I))(\alpha).
\]
Thus, \(\phi\) is an optimal rejection rule for
\[
H_{0}\!:\!(P_{I}\mid I,I)
\qquad\text{vs.}\qquad
H_{1}\!:\!(Q_{I}\mid I,I).
\]

For notational convenience, set \(\overline{P}_i=P_i\) and \(\overline{Q}_i=Q_i\) for all \(i\neq y\). We now construct a rejection rule \(\phi_{\mathrm{sub}}\colon(x,i)\mapsto[0,1]\) with Type~I error \(\alpha\) for
\[
H_{0}\!:\!(\overline{P}_{I}\mid I,I)
\qquad\text{vs.}\qquad
H_{1}\!:\!(\overline{Q}_{I}\mid I,I).
\]
Let \(\phi_{\mathrm{sub}}^{y}\) be an optimal rejection rule for
\[
H_{0}\!:\!(\overline{P}_{y},y)
\qquad\text{vs.}\qquad
H_{1}\!:\!(\overline{Q}_{y},y)
\]
with Type~I error \(\alpha_{y}^{\ast}\). Then
\begin{equation}\label{eq:construct_sub_fact}
\begin{aligned}
    &\mathbb{E}_{(x,y)\sim(\overline{Q}_{y},y)}[1-\phi_{\mathrm{sub}}^{y}(x,y)]
    = T(\overline{P}_{y},\overline{Q}_{y})(\alpha_{y}^{\ast})\\
    &\le T(P_{y},Q_{y})(\alpha_{y}^{\ast})
    = \mathbb{E}_{(x,y)\sim(Q_{y},y)}[1-\phi^{y}(x,y)],
\end{aligned}
\end{equation}
where the inequality follows from the pointwise assumption
\(T(\overline{P}_{y},\overline{Q}_{y})(a)\le T(P_{y},Q_{y})(a)\) for all \(a\in[0,1]\), applied at \(a=\alpha_y^\ast\).

Define \(\phi_{\mathrm{sub}}\) by replacing the slice \(\phi^{y}\) in \(\phi\) with \(\phi_{\mathrm{sub}}^{y}\), while keeping all other slices unchanged. Since \(\overline{P}_i=P_i\) for \(i\neq y\), and since \(\phi_{\mathrm{sub}}^{y}\) has Type~I error \(\alpha_y^\ast\) under \(\overline{P}_y\), the rule \(\phi_{\mathrm{sub}}\) has Type~I error
\[
\sum_{i\neq y}w_i\alpha_i^\ast+w_y\alpha_y^\ast=\alpha
\]
under \((\overline{P}_{I}\mid I,I)\). We now compare the Type~II errors:
\begin{equation}
\begin{aligned}
        &T((P_{I}\mid I,I),(Q_{I}\mid I,I))(\alpha)\\
        \overset{(1)}{=}\,&
        \mathbb{E}_{i\sim w}\mathbb{E}_{(x,i)\sim(Q_{i},i)}[1-\phi(x,i)]\\
        \overset{(2)}{=}\,&
        \sum_{i\in[m]\setminus\{y\}}w_{i}
        \mathbb{E}_{(x,i)\sim(Q_{i},i)}[1-\phi^{i}(x,i)]\\
        &\quad +\,w_{y}\mathbb{E}_{(x,y)\sim(Q_{y},y)}[1-\phi^{y}(x,y)]\\
        \overset{(3)}{\ge}\,&
        \sum_{i\in[m]\setminus\{y\}}w_{i}
        \mathbb{E}_{(x,i)\sim(\overline{Q}_{i},i)}[1-\phi^{i}(x,i)]\\
        &\quad +\,w_{y}\mathbb{E}_{(x,y)\sim(\overline{Q}_{y},y)}[1-\phi_{\mathrm{sub}}^{y}(x,y)]\\
        \overset{(4)}{=}\,&
        \mathbb{E}_{i\sim w}\mathbb{E}_{(x,i)\sim(\overline{Q}_{i},i)}[1-\phi_{\mathrm{sub}}(x,i)]\\
        \overset{(5)}{\ge}\,&
        T((\overline{P}_{I}\mid I,I),(\overline{Q}_{I}\mid I,I))(\alpha).
\end{aligned}
\end{equation}

Here, \(\overset{(1)}{=}\) follows from the optimality of \(\phi\) and the definition of the trade-off function. Equality \(\overset{(2)}{=}\) decomposes the Type~II error over the slices \(I=i\). Inequality \(\overset{(3)}{\ge}\) follows from Eq.~\eqref{eq:construct_sub_fact} and the identities \(\overline{Q}_i=Q_i\) for all \(i\neq y\). Equality \(\overset{(4)}{=}\) follows from the definition of \(\phi_{\mathrm{sub}}\). Finally, \(\overset{(5)}{\ge}\) holds because \(\phi_{\mathrm{sub}}\) is a valid rejection rule for
\[
H_{0}\!:\!(\overline{P}_{I}\mid I,I)
\qquad\text{vs.}\qquad
H_{1}\!:\!(\overline{Q}_{I}\mid I,I)
\]
with Type~I error \(\alpha\), and its Type~II error is therefore no smaller than the optimal Type~II error. Hence,
\[
T((P_{I}\mid I,I),(Q_{I}\mid I,I))(\alpha)
\ge
T((\overline{P}_{I}\mid I,I),(\overline{Q}_{I}\mid I,I))(\alpha).
\]
This completes the proof.
\end{proof}
\subsection{Proof of Lemma \ref{Lemma:Upper_bound}}
\label{Proof:Lemma_Upper_bound}

Let \(\overline{\phi}\) denote the optimal rejection rule with Type~I error \(\alpha\) for the hypothesis test
\[
H_{0}:(P_{I}\mid I,I)
\qquad \text{vs.} \qquad
H_{1}:(Q_{I}\mid I,I).
\]
Then, by definition,
\[
1-\mathbb{E}_{(x,i)\sim(Q_{I}\mid I,I)}[\overline{\phi}(x,i)]
=
T((P_{I}\mid I,I),(Q_{I}\mid I,I))(\alpha).
\]

To prove the claim, we construct a rejection rule \(\phi\) for the same testing problem.
For each \(i\in[m]\), let \(\phi_i\) be the optimal rejection rule for testing
\[
H_{0}:P_i
\qquad \text{vs.} \qquad
H_{1}:Q_i
\]
with Type~I error \(\alpha\). Thus,
\[
\mathbb{E}_{x\sim P_i}[\phi_i(x)]=\alpha,
\qquad
1-\mathbb{E}_{x\sim Q_i}[\phi_i(x)]
=
T(P_i,Q_i)(\alpha).
\]

Now define the rejection rule \(\phi\) by
\[
\phi(x,i)=\phi_i(x).
\]
Then its Type~I error is
\[
\mathbb{E}_{(x,i)\sim(P_I\mid I,I)}[\phi(x,i)]
=
\sum_{i\in[m]}w_i\,\mathbb{E}_{x\sim P_i}[\phi_i(x)]
=
\sum_{i\in[m]}w_i\,\alpha
=
\alpha.
\]
Similarly, its Type~II error is
\[\begin{aligned}
    &1-\mathbb{E}_{(x,i)\sim(Q_I\mid I,I)}[\phi(x,i)]
=
\sum_{i\in[m]}w_i\Bigl(1-\mathbb{E}_{x\sim Q_i}[\phi_i(x)]\Bigr) \\
&=
\sum_{i\in[m]}w_i\,T(P_i,Q_i)(\alpha).
\end{aligned}
\]

Since \(\overline{\phi}\) is the optimal rejection rule at Type~I error \(\alpha\), its Type~II error is no larger than that of any feasible rejection rule. Therefore,
\[
T((P_I\mid I,I),(Q_I\mid I,I))(\alpha)
\le
\sum_{i\in[m]}w_i\,T(P_i,Q_i)(\alpha).
\]
This completes the proof.
\subsection{Proof of Theorem \ref{thm:bound}}
\label{App_proof:Theorem_bound}
Let \(P=\mathcal{N}(0,1)\) and \(Q_y=\mathcal{N}(y,\sigma^2)\), where \(0<\sigma<1\) and \(y>0\).
We fix a Type~I error level \(\alpha\in(0,1)\).

First, the likelihood ratio between \(Q_y\) and \(P\) is
\begin{equation}
\begin{aligned}
\Lambda_y(x)
&= \frac{f_{Q_y}(x)}{f_P(x)}
 = \frac{1}{\sigma}
 \exp\!\left(
 -\frac{1}{2}\left[\frac{(x-y)^2}{\sigma^2}-x^2\right]
 \right) \\
&= \frac{1}{\sigma}\exp\!\bigl(\lambda(x,y)\bigr),
\end{aligned}
\end{equation}
where
\[
\lambda(x,y)=\frac{-(1-\sigma^2)x^2+2yx-y^2}{2\sigma^2}.
\]

Since \(0<\sigma<1\), \(\lambda(x,y)\) is a concave quadratic function of \(x\), and hence is symmetric about
\[
z=\frac{y}{1-\sigma^2}.
\]
Therefore, by the Neyman--Pearson lemma, for fixed Type~I error \(\alpha\), the rejection region minimizing the Type~II error \(\beta\) is of the form
\[
[z-t,z+t]
\]
for some \(t>0\), determined by
\[
\int_{z-t}^{z+t} f_P(x)\,dx=\alpha
\quad\Longleftrightarrow\quad
\Phi(z+t)-\Phi(z-t)=\alpha.
\]

Define
\[
F(y,t)=\Phi(z(y)+t)-\Phi(z(y)-t)-\alpha,
\qquad z(y)=\frac{y}{1-\sigma^2}.
\]
Since
\[
\frac{\partial F}{\partial t}
=
\phi(z(y)+t)+\phi(z(y)-t)>0,
\]
the implicit function theorem implies that the solution \(t=t(y)\) is locally continuously differentiable in \(y\).

Thus,
\begin{equation}\label{eq:one_minus_beta}
\begin{aligned}
1-\beta
&= \int_{z-t}^{z+t} f_{Q_y}(x)\,dx \\
&= \Phi\!\left(\frac{\sigma}{1-\sigma^2}y+\frac{t}{\sigma}\right)
 - \Phi\!\left(\frac{\sigma}{1-\sigma^2}y-\frac{t}{\sigma}\right).
\end{aligned}
\end{equation}

We now differentiate \(t\) with respect to \(y\). From
\[
\Phi(z+t)-\Phi(z-t)=\alpha,
\]
we obtain
\begin{equation}
\begin{aligned}
\phi(z+t)\left(\frac{\partial z}{\partial y}+\frac{\partial t}{\partial y}\right)
-\phi(z-t)\left(\frac{\partial z}{\partial y}-\frac{\partial t}{\partial y}\right)
=0.
\end{aligned}
\end{equation}
Since \(\frac{\partial z}{\partial y}=\frac{1}{1-\sigma^2}\), it follows that
\begin{equation}\label{eq:dt_dy}
\begin{aligned}
\frac{\partial t}{\partial y}
=
\frac{\phi(z-t)-\phi(z+t)}{\phi(z-t)+\phi(z+t)}
\cdot \frac{1}{1-\sigma^2}.
\end{aligned}
\end{equation}

Next, differentiating \eqref{eq:one_minus_beta} with respect to \(y\), we get
\begin{equation}
\begin{aligned}
\frac{\partial(1-\beta)}{\partial y}
&=
\left(
\frac{\sigma}{1-\sigma^2}+\frac{1}{\sigma}\frac{\partial t}{\partial y}
\right)
\phi\!\left(\frac{\sigma}{1-\sigma^2}y+\frac{t}{\sigma}\right) \\
&\quad -
\left(
\frac{\sigma}{1-\sigma^2}-\frac{1}{\sigma}\frac{\partial t}{\partial y}
\right)
\phi\!\left(\frac{\sigma}{1-\sigma^2}y-\frac{t}{\sigma}\right) \\
&=
\frac{\sigma}{1-\sigma^2}
\left[
\phi\!\left(\frac{\sigma}{1-\sigma^2}y+\frac{t}{\sigma}\right)
-
\phi\!\left(\frac{\sigma}{1-\sigma^2}y-\frac{t}{\sigma}\right)
\right] \\
&\quad +
\frac{1}{\sigma}\frac{\partial t}{\partial y}
\left[
\phi\!\left(\frac{\sigma}{1-\sigma^2}y+\frac{t}{\sigma}\right)
+
\phi\!\left(\frac{\sigma}{1-\sigma^2}y-\frac{t}{\sigma}\right)
\right].
\end{aligned}
\end{equation}

Define
\[
\Gamma(y)
=
\frac{\partial(1-\beta)}{\partial y}\cdot
\sigma(1-\sigma^2)\bigl(\phi(z-t)+\phi(z+t)\bigr).
\]
Since \(\sigma>0\), \(1-\sigma^2>0\), and \(\phi(z-t)+\phi(z+t)>0\), \(\Gamma(y)\) and \(\frac{\partial(1-\beta)}{\partial y}\) have the same sign. Using \eqref{eq:dt_dy}, we obtain
\begin{equation}
\begin{aligned}
\Gamma(y)
&=
\sigma^2\bigl(\phi(z+t)+\phi(z-t)\bigr)
\left[
\phi\!\left(\frac{\sigma}{1-\sigma^2}y+\frac{t}{\sigma}\right)
-
\phi\!\left(\frac{\sigma}{1-\sigma^2}y-\frac{t}{\sigma}\right)
\right] \\
&\quad +
\bigl(\phi(z-t)-\phi(z+t)\bigr)
\left[
\phi\!\left(\frac{\sigma}{1-\sigma^2}y+\frac{t}{\sigma}\right)
+
\phi\!\left(\frac{\sigma}{1-\sigma^2}y-\frac{t}{\sigma}\right)
\right].
\end{aligned}
\end{equation}

Now define
\[
\Gamma^*(y)
=
2\pi\,\Gamma(y)\,
e^{\frac{(z+t)^2}{2}}
e^{\frac{(z-t)^2}{2}}
e^{\frac{\left(\frac{\sigma}{1-\sigma^2}y+\frac{t}{\sigma}\right)^2}{2}}
e^{\frac{\left(\frac{\sigma}{1-\sigma^2}y-\frac{t}{\sigma}\right)^2}{2}}.
\]
Since the multiplying factor is strictly positive, \(\Gamma^*(y)\) and \(\Gamma(y)\) have the same sign. A direct expansion gives
\begin{equation}
\begin{aligned}
\Gamma^*(y)
&=
(1+\sigma^2)
\left[
e^{\frac{(z+t)^2+\left(\frac{\sigma}{1-\sigma^2}y-\frac{t}{\sigma}\right)^2}{2}}
-
e^{\frac{(z-t)^2+\left(\frac{\sigma}{1-\sigma^2}y+\frac{t}{\sigma}\right)^2}{2}}
\right] \\
&\quad +
(1-\sigma^2)
\left[
e^{\frac{(z+t)^2+\left(\frac{\sigma}{1-\sigma^2}y+\frac{t}{\sigma}\right)^2}{2}}
-
e^{\frac{(z-t)^2+\left(\frac{\sigma}{1-\sigma^2}y-\frac{t}{\sigma}\right)^2}{2}}
\right].
\end{aligned}
\end{equation}

Substituting \(z=\frac{y}{1-\sigma^2}\), define
\[
X=
\frac{\left(\frac{y}{1-\sigma^2}+t\right)^2
+\left(\frac{\sigma}{1-\sigma^2}y-\frac{t}{\sigma}\right)^2}{2},
Y=
\frac{\left(\frac{y}{1-\sigma^2}-t\right)^2
+\left(\frac{\sigma}{1-\sigma^2}y+\frac{t}{\sigma}\right)^2}{2},
\]
and
\[
Z=
\frac{\left(\frac{y}{1-\sigma^2}+t\right)^2
+\left(\frac{\sigma}{1-\sigma^2}y+\frac{t}{\sigma}\right)^2}{2},
W=
\frac{\left(\frac{y}{1-\sigma^2}-t\right)^2
+\left(\frac{\sigma}{1-\sigma^2}y-\frac{t}{\sigma}\right)^2}{2}.
\]
Then
\[
\Gamma^*(y)
=
(1+\sigma^2)(e^X-e^Y)
+
(1-\sigma^2)(e^Z-e^W).
\]

We now compare these exponents. First,
\begin{equation}
\begin{aligned}
X-Y
=
2\left(\frac{y}{1-\sigma^2}\right)t
-
2\left(\frac{\sigma y}{1-\sigma^2}\right)\left(\frac{t}{\sigma}\right)
=0.
\end{aligned}
\end{equation}
Hence \(e^X-e^Y=0\). Second,
\begin{equation}
\begin{aligned}
Z-W
=
\frac{4ty}{1-\sigma^2}>0,
\end{aligned}
\end{equation}
because \(t>0\), \(y>0\), and \(0<\sigma<1\). Therefore \(e^Z-e^W>0\). Since \(1-\sigma^2>0\), we conclude that
\[
\Gamma^*(y)>0.
\]
Hence \(\Gamma(y)>0\), and therefore
\[
\frac{\partial(1-\beta)}{\partial y}>0.
\]
Equivalently,
\[
\frac{\partial\beta}{\partial y}<0.
\]

Thus, for fixed \(\alpha\in(0,1)\), the Type~II error is strictly decreasing in \(y\). Therefore, if \(y_2>y_1> 0\), then
\[
T(\mathcal{N}(0,1),\mathcal{N}(y_2,\sigma^2))(\alpha)
<
T(\mathcal{N}(0,1),\mathcal{N}(y_1,\sigma^2))(\alpha).
\]

For \(\alpha\in(0,1)\), the above argument shows that, for any
\(\mu_1>\mu_2>0\),
\[
T(\mathcal{N}(0,1),\mathcal{N}(\mu_1,\sigma^2))(\alpha)
<
T(\mathcal{N}(0,1),\mathcal{N}(\mu_2,\sigma^2))(\alpha).
\]
The boundary case \(\mu_2=0\) follows by continuity of the trade-off function
with respect to the mean parameter. Indeed, for fixed \(\alpha\in(0,1)\),
the threshold \(t=t(\mu)\) is determined by
\[
\Phi\!\left(\frac{\mu}{1-\sigma^2}+t\right)
-
\Phi\!\left(\frac{\mu}{1-\sigma^2}-t\right)
=
\alpha .
\]
Since the derivative of the left-hand side with respect to \(t\) is
\[
\phi\!\left(\frac{\mu}{1-\sigma^2}+t\right)
+
\phi\!\left(\frac{\mu}{1-\sigma^2}-t\right)>0,
\]
the implicit function theorem implies that \(t(\mu)\) is continuous in
\(\mu\). Hence
\[
T\!\left(\mathcal N(0,1),\mathcal N(\mu,\sigma^2)\right)(\alpha)
\]
is continuous in \(\mu\). Therefore, for any \(0<\varepsilon<\mu_1\),
\[
T\!\left(\mathcal N(0,1),\mathcal N(\mu_1,\sigma^2)\right)(\alpha)
<
T\!\left(\mathcal N(0,1),\mathcal N(\varepsilon,\sigma^2)\right)(\alpha).
\]
Letting \(\varepsilon\downarrow 0\) gives
\[
T\!\left(\mathcal N(0,1),\mathcal N(\mu_1,\sigma^2)\right)(\alpha)
\le
T\!\left(\mathcal N(0,1),\mathcal N(0,\sigma^2)\right)(\alpha),
\alpha\in(0,1).
\]
At the endpoints \(\alpha=0\) and \(\alpha=1\), the trade-off functions
coincide, since
\[
T(P,Q)(0)=1,
\qquad
T(P,Q)(1)=0.
\]
Therefore, for any \(\mu_1>\mu_2\ge 0\),
\[
T\!\left(\mathcal N(0,1),\mathcal N(\mu_1,\sigma^2)\right)(\alpha)
\le
T\!\left(\mathcal N(0,1),\mathcal N(\mu_2,\sigma^2)\right)(\alpha),
\forall \alpha\in[0,1].
\]

\subsection{Proof of Theorem \ref{thm:bound_fix_variance}}\label{App_proof:bound_fix_variance}
\begin{proof}
Let \(P=\mathcal N(0,1)\) and \(Q_\sigma=\mathcal N(0,\sigma^2)\) with \(0<\sigma< 1\). 
Fix \(\alpha\in(0,1)\).

For \(0<\sigma<1\), the likelihood ratio is
\begin{equation}
\begin{aligned}
\Lambda(x)
&=\frac{f_{Q_\sigma}(x)}{f_P(x)}
= \frac{1}{\sigma}\exp\!\left(-\frac{1}{2}\left[\frac{x^2}{\sigma^2}-x^2\right]\right) \\
&= \frac{1}{\sigma}\exp\!\left(-\frac{1-\sigma^2}{2\sigma^2}x^2\right).
\end{aligned}
\end{equation}
Since \(\Lambda(x)\) is symmetric about \(0\) and strictly decreasing in \(|x|\), the Neyman--Pearson most-powerful test with Type~I error \(\alpha\) has rejection region \([-t,t]\) for some \(t>0\), where
\[
\alpha=\int_{-t}^{t} f_P(x)\,dx
=\Phi(t)-\Phi(-t)=2\Phi(t)-1.
\]
Hence \(t\) is determined by \(\alpha\) alone and does not depend on \(\sigma\).

Under \(Q_\sigma\), we have
\[
1-\beta
=\int_{-t}^{t} f_{Q_\sigma}(x)\,dx
=\Phi\!\left(\frac{t}{\sigma}\right)-\Phi\!\left(-\frac{t}{\sigma}\right)
=2\Phi\!\left(\frac{t}{\sigma}\right)-1.
\]
Differentiating with respect to \(\sigma\), and using the fact that \(t\) is independent of \(\sigma\), yields
\[
\frac{\partial(1-\beta)}{\partial \sigma}
=
-\frac{2t}{\sigma^2}\phi\!\left(\frac{t}{\sigma}\right)
<0.
\]
Therefore,
\[
\frac{\partial \beta}{\partial \sigma}>0.
\]
Thus, for any fixed \(\alpha\in(0,1)\), the trade-off function
\(T(\mathcal N(0,1),\mathcal N(0,\sigma^2))(\alpha)\) is strictly increasing
in \(\sigma\) on \(0<\sigma\le 1\).

When \(\sigma=1\), we have \(Q_\sigma=P\), and hence
\[
T(P,Q_\sigma)(\alpha)=\mathrm{ID}(\alpha)=1-\alpha.
\]
Moreover, as \(\sigma\uparrow 1\), we have
\[
2\Phi\left(\frac{t}{\sigma}\right)-1 \to 2\Phi(t)-1=\alpha,
\]
and hence
\[
T(P,Q_\sigma)(\alpha)=\beta \to 1-\alpha = T(P,P)(\alpha).
\]
Thus the monotonicity extends to \(0<\sigma\le 1\).

For the endpoint cases in \(\alpha\), we have
\[
T(P,Q_\sigma)(0)=1,
\qquad
T(P,Q_\sigma)(1)=0.
\]
Hence the monotonicity also holds for \(\alpha\in[0,1]\).

Therefore, if \(1\ge \sigma_1\ge \sigma_2>0\), then
\[
T(\mathcal N(0,1),\mathcal N(0,\sigma_1^2))(\alpha)
\ge
T(\mathcal N(0,1),\mathcal N(0,\sigma_2^2))(\alpha),
\qquad
\forall \alpha\in[0,1].
\]
\end{proof}

\subsection{Proof of Theorem \ref{theorem:EASGM_HIGH}}\label{Appendix:Proof_EASGM_HIGH}
First, we note that the trade-off function $T(P, Q^{\mu_{B}^{1}}_{N}) \ge T(P,\, Q_{N}^{\overline{\mu}_{B}^{1}})$ holds, where $\overline{\mu}_{B}^{1}$ is defined in Eq.~\ref{eq:Q_1_B}. 
This result follows from the inequality
$\frac{\|\sum_{i \in B} \overline{g}_{i}\|_{2}}{(N+1)C\sigma} 
\le \frac{|B|}{(N+1)\sigma}$,
and from Theorem~\ref{thm:bound}, which guarantees that 
\(T(P, Q_{N}^{\mu_{B}^{1}}) \ge T(P, Q_{N}^{\overline{\mu}_{B}^{1}})\).
Similarly, we can show that $T(P, Q^{\mu_{B}^{2}}_{N}) \ge T(P,\, Q_{N}^{\overline{\mu}_{B}^{2}})$ holds, where $\overline{\mu}_{B}^{2}$ is defined in Eq.~\ref{eq:Q_2_B}. 
Hence, we have
\begin{equation}
\begin{aligned}
f_{B}^{i} 
&= T(P^{\otimes (n-1)}, Q_{N}^{\otimes (n-1)}) \otimes T(P, Q^{\mu_{B}^{i}}_{N}) \\
&\overset{(2)}{\ge} T(P^{\otimes (n-1)}, Q^{\otimes (n-1)}_{N}) \otimes T(P, Q^{\overline{\mu}_{B}^{i}}_{N}) \\
&\overset{(3)}{=} T(P^{\otimes n}, Q_{N}^{\otimes (n-1)} \otimes Q^{\overline{\mu}_{B}^{i}}_{N}) = \underline{f_{B}^{i}}, \quad i \in \{1,2\}.
\end{aligned}
\end{equation}

\noindent Here, $\overset{(3)}{=}$ holds due to Property~\ref{property:expression_composition}, 
while $\overset{(2)}{\ge}$ is justified by Property~\ref{property:order}. 
Finally, according to Lemma~\ref{Lemma:Order_in_trade-off}, we have 
$T(\mathcal{M}(D), \mathcal{M}(D')) \ge f \ge \underline{f}$, 
where $f$ is defined in Eq.~\ref{eq:lower_bound_EASGM_high_dimentsion}.

\subsection{Proof of Theorem \ref{thm:ASGM_HIGH}}
\label{Appendix:Thm_ASGM_HIGH}
\begin{proof}
Following the same proof strategy as in Theorem~\ref{theorem:EASGM_HIGH}, 
by the definitions of the corresponding component trade-off functions, we have
\[
f_{\emptyset}^{1}=\underline f_{\emptyset}^{1},
\qquad
f_B^{1}=\underline f_B^{1}.
\]
Moreover, since the Gaussian trade-off function \(G_\mu\) is decreasing in \(\mu\), and since \(\|\overline g_{N+1}\|_2\le C\), we have
\[
f_{\emptyset}^{2}\ge G_{1/\sigma}.
\]
Thus, it remains to show that
\[
f_B^{2}\ge \underline f_B^{2}.
\]

For \(Q_B^{\mu_B}\) defined in Eq.~\ref{eq:ASGM_F_2}, the clipping condition gives
\[
\left\||B|\overline g_{N+1}-\sum_{i\in B}\overline g_i\right\|_2
\le
|B|\|\overline g_{N+1}\|_2+\sum_{i\in B}\|\overline g_i\|_2
\le
2|B|C.
\]
Therefore,
\[
\mu_B
=
\frac{
\left\||B|\overline g_{N+1}-\sum_{i\in B}\overline g_i\right\|_2
}{
(|B|+1)C\sigma
}
\le
\frac{2|B|}{(|B|+1)\sigma}
=
\overline\mu_B.
\]
By Theorem~\ref{thm:bound}, this implies
\[
T(P,Q_B^{\mu_B})\ge T(P,Q_B^{\overline\mu_B}).
\]
Thus, given \(f_B^{2}\) defined in Eq.~\ref{eq:ASGM_F_2}, we have
\[
\begin{aligned}
f_B^{2}
&=
T(P^{\otimes (n-1)},Q_B^{\otimes (n-1)})
\otimes T(P,Q_B^{\mu_B})\\
&\ge
T(P^{\otimes (n-1)},Q_B^{\otimes (n-1)})
\otimes T(P,Q_B^{\overline\mu_B})\\
&=
T(P^{\otimes n},Q_B^{\otimes (n-1)}\otimes Q_B^{\overline\mu_B})
=
\underline f_B^{2}.
\end{aligned}
\]
Here, the inequality follows from Property~\ref{property:order}, and the equalities follow from Property~\ref{property:expression_composition}.

Since these component-wise inequalities hold for all relevant \(B\), Lemma~\ref{Lemma:Order_in_trade-off} implies
\[
f\ge \underline f.
\]
Combining this with the per-pair guarantee
\[
T(\mathcal M(D),\mathcal M(D'))\ge f
\]
gives
\[
T(\mathcal M(D),\mathcal M(D'))\ge f\ge \underline f.
\]
\end{proof}

\subsection{Proof of Theorem \ref{thm:upperbound_EASGM}}
\label{Appendix:upperbound_EASGM}
First, according to Property~\ref{property:useful}, we have 
$\overline{f} \ge f_{B}^{1}$ and $\overline{f} \ge f_{B}^{2}$, 
where $f_{B}^{1}$ and $f_{B}^{2}$ are defined in Eqs.~\ref{eq:Q_1_B} and~\ref{eq:Q_2_B}, respectively. 
By replacing $f_{B}^{i}$ with $\overline{f}$ for $i \in \{1, 2\}$ in the definition of $f$, we obtain $f'$. 
Then, by Lemma~\ref{Lemma:Order_in_trade-off}, we have $f \le f'$, and by Lemma~\ref{Lemma:Upper_bound}, we have $f' \le \overline{f}$. 
The proof is complete.

\subsection{Proof of Theorem \ref{theorem:EASGM_flaws}}
\label{Appendix:EASGM_flaws}
First, let $f^{\ast} = T(P, Q_{N})$, where $P$ and $Q_{N}$ are defined in table \ref{tab:notation}. 
Note that $f^{\ast} \neq \mathrm{ID}$. 
Then, according to Lemma~\ref{Lemma:zero}, we have 
$\lim_{n \to \infty} (f^{\ast})^{\otimes (n-1)}(\alpha) = 0$ for all $\alpha \in (0, 1]$. 
Then, by Property~\ref{property:expression_composition}, 
we have $(f^{\ast})^{\otimes (n-1)} = \overline{f}$. 
Therefore,  $\lim_{n \to \infty} \overline{f}(\alpha) = 0 , 
 \forall \alpha \in (0, 1]$.

\subsection{Proof of Theorem \ref{thm:upperbound_ASGM}} \label{Appendix:upperbound_ASGM}

\noindent
Let $f_{B}^{1}$ and $f_{B}^{2}$ be defined in Eqs.~\ref{eq:ASGM_F_1} and~\ref{eq:ASGM_F_2}, respectively, and let $f_{\emptyset}^{1}=\mathrm{ID}$ and $f_{\emptyset}^{2}=T(G,G+\overline{g}_{N+1})$. 
We first show that
$T\!\big(P^{\otimes (n-1)},\, Q_{N}^{\otimes (n-1)}\big) \ge f_{B}^{2}$.

Define \(F(x)=x^{2}/(x+1)^{2}\) for \(x\ge 0\). 
Since \(F'(x)=2x/(x+1)^{3}\ge 0\), the function \(F\) is monotonically increasing, and hence \(F(N)\ge F(|B|)\). 
Let $P$, $Q_{N}$, and $Q_{B}$ be defined in Table~\ref{tab:notation}, and let $Q^{\mu_{B}}_{B}$ be defined in Eq.~\ref{eq:ASGM_F_2}. 
Then we have
\begin{equation}
\begin{aligned}
    &T\!\big(P^{\otimes (n-1)},\, Q_{N}^{\otimes (n-1)}\big)
    = T\!\big(P^{\otimes (n-1)},\, Q_{N}^{\otimes (n-1)}\big)\otimes \mathrm{ID} \\
    &\overset{(1)}{\ge}\; T\!\big(P^{\otimes (n-1)},\, Q_{B}^{\otimes (n-1)}\big)\otimes T(P, Q^{\mu_{B}}_{B})
    \overset{(2)}{=}\; f_{B}^{2}.
\end{aligned}
\end{equation}
Here, \(\overset{(1)}{\ge}\) follows from Theorem~\ref{thm:bound_fix_variance}, which establishes that \(T(P,Q_{N})\ge T(P,Q_{B})\), together with Properties~\ref{property:expression_composition} and~\ref{property:order}. Equality \(\overset{(2)}{=}\) follows from Property~\ref{property:expression_composition}.

Moreover, since \(\mathrm{ID}=f_{B}^{1}\) and \(\mathrm{ID}\ge f_{\emptyset}^{i}\) for \(i\in\{1,2\}\), substituting \(f_{B}^{1}\), \(f_{B}^{2}\), and \(f_{\emptyset}^{i}\) \((i\in\{1,2\})\) in the definition of \(f\) with \(\mathrm{ID}\), \(T\!\big(P^{\otimes (n-1)},\, Q_{N}^{\otimes (n-1)}\big)\), and \(\mathrm{ID}\), respectively, yields a new trade-off function \(f^{\ast}\).

By Lemma~\ref{Lemma:Order_in_trade-off}, we have \(f^{\ast}\ge f\). Combined with Lemma~\ref{Lemma:Upper_bound}, which gives \(\overline{f}\ge f^{\ast}\), we conclude that \(\overline{f}\ge f\).
\subsection{Proof of Theorem \ref{theorem:ASGM_flaws}}\label{Appendix:ASGM_flaws}
We only need to prove that
\[
\lim_{n\to\infty}\overline{f}(\alpha)
=
\bigl(1-q+q(1-q)^N\bigr)\mathrm{ID}(\alpha),
\qquad
\forall \alpha\in(0,1].
\]

Let
\[
f^\ast := T(P,Q_N),
\]
where $P$ and $Q_N$ are defined in Table~\ref{tab:notation}. Since $f^\ast\neq \mathrm{ID}$, Lemma~\ref{Lemma:zero} implies that
\[
\lim_{n\to\infty}(f^\ast)^{\otimes (n-1)}(\alpha)=0,
\qquad
\forall \alpha\in(0,1].
\]
Moreover, by Property~\ref{property:expression_composition},
\[
(f^\ast)^{\otimes (n-1)}
=
T(P^{\otimes n-1},Q_N^{\otimes n-1}).
\]

According to the definition of $\overline{f}$ in Theorem~\ref{thm:upperbound_ASGM},
\[
\overline{f}(\alpha)
=
q\bigl(1-(1-q)^N\bigr)\,
T(P^{\otimes n-1},Q_N^{\otimes n-1})(\alpha)
+
\bigl(1-q+q(1-q)^N\bigr)\mathrm{ID}(\alpha).
\]
Taking the limit as $n\to\infty$, we obtain
\[
\begin{aligned}
\lim_{n\to\infty}\overline{f}(\alpha)
&=
q\bigl(1-(1-q)^N\bigr)\cdot 0
+
\bigl(1-q+q(1-q)^N\bigr)\mathrm{ID}(\alpha) \\
&=
\bigl(1-q+q(1-q)^N\bigr)\mathrm{ID}(\alpha),
\end{aligned}
\]
for all $\alpha\in(0,1]$.

Finally, since Theorem~\ref{thm:upperbound_ASGM} shows that
\[
\overline{f}(\alpha)\ge f(\alpha),
\qquad
\forall \alpha\in(0,1],
\]
taking the limit yields
\[
\bigl(1-q+q(1-q)^N\bigr)\mathrm{ID}(\alpha)\ge f(\alpha),
\qquad
\forall \alpha\in(0,1].
\]
This completes the proof.
\subsection{Proof of Lemma \ref{fact:number_of_data}}\label{Proof:Lemma_number_of_data}
\begin{proof}
For $\overline{f}$ in Theorem~\ref{thm:upperbound_EASGM}, we follow the same notation as in that theorem. 
According to the definition of $\overline{f}$, it suffices to prove that if $N_{1}\ge N_{2}$, then
\[
T(P^{\otimes (n-1)},Q_{N_{1}}^{\otimes (n-1)})\ge T(P^{\otimes (n-1)},Q_{N_{2}}^{\otimes (n-1)}).
\]
Note that $P=\mathcal{N}(0,1)$ and $Q_{N}=\mathcal{N}(0,F(N))$, where $F(x)=\frac{x^{2}}{(x+1)^{2}}$. 
As established in the proof of Theorem~\ref{thm:upperbound_ASGM}, $F(x)$ increases as $x$ increases. Thus, according to Theorem~\ref{thm:bound_fix_variance}, we have
\[
T(P,Q_{N_{1}})\ge T(P,Q_{N_{2}}).
\]
Further, by Properties~\ref{property:expression_composition} and~\ref{property:order},
\[
\begin{aligned}
    &T(P^{\otimes (n-1)},Q_{N_{1}}^{\otimes (n-1)})
=
T(P,Q_{N_{1}})^{\otimes (n-1)}
\ge
T(P,Q_{N_{2}})^{\otimes (n-1)} \\
& = 
T(P^{\otimes (n-1)},Q_{N_{2}}^{\otimes (n-1)}).
\end{aligned}
\]
This completes the proof.
\end{proof}
\subsection{Proof of Lemma \ref{fact:dimension}}
\label{Proof:Lemma:dimension}

\begin{proof}
We first consider \(\overline f\) in Theorem~\ref{thm:upperbound_ASGM}. 
By the definition of \(\overline f\), it suffices to show that, if \(n_1\ge n_2\), then
\[
T(P^{\otimes(n_1-1)},Q_N^{\otimes(n_1-1)})
\le 
T(P^{\otimes(n_2-1)},Q_N^{\otimes(n_2-1)}).
\]
By Properties~\ref{property:expression_composition} and~\ref{property:order}, we have
\[
\begin{aligned}
&T(P^{\otimes(n_2-1)},Q_N^{\otimes(n_2-1)})\\
&=
T(P^{\otimes(n_2-1)},Q_N^{\otimes(n_2-1)})
\otimes \mathrm{ID}^{\otimes(n_1-n_2)}\\
&\ge
T(P^{\otimes(n_2-1)},Q_N^{\otimes(n_2-1)})
\otimes
T(P^{\otimes(n_1-n_2)},Q_N^{\otimes(n_1-n_2)})\\
&=
T(P^{\otimes(n_1-1)},Q_N^{\otimes(n_1-1)}).
\end{aligned}
\]
Thus, \(\overline f\) decreases pointwise as \(n\) increases.

The EASGM case follows by the same argument, since its upper bound depends on \(n\) through the same term
\[
T(P^{\otimes(n-1)},Q_N^{\otimes(n-1)}).
\]
This completes the proof.
\end{proof}
\subsection{Proof of Lemma \ref{fact:sample_ratio}}
\label{Appendix:proof_asgm_sample_ratio}

\sloppy
Let
$a(q):=q-q(1-q)^N$,
$b(q):=1-q+q(1-q)^N,$
so that \(a(q)+b(q)=1\). Then
\[
\overline f_q(\alpha)
=
a(q)\,T\!\left(P^{\otimes (n-1)},Q_N^{\otimes (n-1)}\right)(\alpha)
+
\bigl(1-a(q)\bigr)\,\mathrm{ID}(\alpha).
\]

We first show that \(a(q)\) is strictly increasing in \(q\) on \([0,1]\). For \(q\in(0,1)\),
\[
a'(q)
=
\frac{d}{dq}\bigl(q-q(1-q)^N\bigr)
=
1-\frac{d}{dq}\bigl(q(1-q)^N\bigr).
\]
By the product rule,
\[
\frac{d}{dq}\bigl(q(1-q)^N\bigr)
=
(1-q)^N-Nq(1-q)^{N-1},
\]
and therefore
\[
a'(q)
=
1-(1-q)^N+Nq(1-q)^{N-1}.
\]
Since \(q\in(0,1)\) and \(N\ge 1\), we have
\[
1-(1-q)^N>0,
\qquad
Nq(1-q)^{N-1}\ge 0,
\]
which implies
$a'(q)>0$.
Hence \(a(q)\) is strictly increasing on \((0,1)\), and by continuity also increasing on \([0,1]\). In particular, for any \(0\le q_1<q_2\le 1\), if we write
\[
a_1:=a(q_1),
\qquad
a_2:=a(q_2),
\]
then \(a_1<a_2\).

Now compute the difference:
\[
\overline f_{q_1}(\alpha)-\overline f_{q_2}(\alpha)
=
(a_1-a_2)
\Bigl(
T\!\left(P^{\otimes (n-1)},Q_N^{\otimes (n-1)}\right)(\alpha)
-
\mathrm{ID}(\alpha)
\Bigr).
\]
Since \(\mathrm{ID}\) is the maximal trade-off function, we have
\[
T\!\left(P^{\otimes (n-1)},Q_N^{\otimes (n-1)}\right)(\alpha)
\le \mathrm{ID}(\alpha),
\qquad \forall \alpha\in[0,1].
\]
Since \(a_1-a_2<0\), it follows that
\[
\overline f_{q_1}(\alpha)-\overline f_{q_2}(\alpha)\ge 0
\qquad
\text{for all }\alpha\in[0,1].
\]
Therefore,
\[
\overline f_{q_1}(\alpha)\ge \overline f_{q_2}(\alpha)
\qquad
\text{for all }\alpha\in[0,1],
\]
which proves that \(\overline f_q\) is pointwise monotonically decreasing in \(q\in[0,1]\).

Moreover, if
\[
T\!\left(P^{\otimes (n-1)},Q_N^{\otimes (n-1)}\right)(\alpha)<\mathrm{ID}(\alpha)
\]
for some \(\alpha\), then
\[
T\!\left(P^{\otimes (n-1)},Q_N^{\otimes (n-1)}\right)(\alpha)-\mathrm{ID}(\alpha)<0.
\]
Since \(a_1-a_2<0\), we obtain
\[
\overline f_{q_1}(\alpha)-\overline f_{q_2}(\alpha)>0,
\]
and hence
\[
\overline f_{q_1}(\alpha)>\overline f_{q_2}(\alpha).
\]
This proves the strict version.

\subsection{Proof of Theorem \ref{thm:upperbound_estimate}}
\label{App:proof:estimationg}

\begin{proof}
Let \(f_{\ast}=T(P,Q)\) with \(P\sim\mathcal{N}(0,1)\) and \(Q\sim\mathcal{N}(0,\sigma^{2})\) for \(0<\sigma<1\), and define
\[
\overline{f_{\ast}}=\max\{f_{\ast},\,f_{\ast}^{-1}\}.
\]
To prove the theorem, it suffices to show that
\[
\kappa_{3}(\overline{f_{\ast}})
:=\int_{0}^{1}\bigl|\log |\overline{f_{\ast}}'(\alpha)|\bigr|^{3}\,d\alpha
<\infty,
\]
which is the prerequisite for applying Theorem~\ref{thm:estimation}.

Let the trade-off function \(f_{\ast}=T(P,Q)\) be induced by the Neyman--Pearson most-powerful test with rejection region \([-t,t]\). Specifically,
\[
\alpha(t)=\mathbb{P}(|Z|\le t)=2\Phi(t)-1,\qquad
\beta(t)=\mathbb{P}(|X|>t)=2\bigl(1-\Phi(t/\sigma)\bigr),
\]
where \(Z\sim\mathcal{N}(0,1)\), \(X\sim\mathcal{N}(0,\sigma^{2})\), and \(\Phi\) denotes the standard normal CDF. We write \(f_{\ast}(\alpha(t))=\beta(t)\).

\noindent\textbf{Step 1: Proving \(\kappa_{3}(f_{\ast})<\infty\).}
Differentiating with respect to \(t\) gives
\[
\frac{d\alpha}{dt}=2\,\phi(t),\qquad
\frac{d\beta}{dt}=-\,\frac{2}{\sigma}\,\phi\!\left(\frac{t}{\sigma}\right),
\]
where \(\phi(u)=(2\pi)^{-1/2}e^{-u^{2}/2}\) is the standard normal pdf. Consequently,
\[
f_{\ast}'\bigl(\alpha(t)\bigr)
=\frac{d\beta/dt}{d\alpha/dt}
=-\,\frac{1}{\sigma}\,\frac{\phi(t/\sigma)}{\phi(t)}
=-\frac{1}{\sigma}\exp\!\left(-\frac{1-\sigma^2}{2\sigma^2}\,t^2\right).
\]
In particular, \(f_{\ast}'(\alpha(t))<0\) for all \(t\). Moreover,
\[
\frac{d f_{\ast}'(\alpha(t))}{dt}
=\frac{1}{\sigma}\exp\!\left(-\frac{1-\sigma^2}{2\sigma^2}\,t^2\right)
\cdot \frac{1-\sigma^2}{\sigma^2}\,t.
\]
Since \(d\alpha/dt=2\phi(t)>0\), the second derivative satisfies
\[
f''_{\ast}(\alpha(t))
=\frac{\frac{d f_{\ast}'(\alpha(t))}{dt}}{d\alpha/dt}
=\frac{\sqrt{2\pi}}{2}\,\frac{1-\sigma^2}{\sigma^3}\,
t\,\exp\!\left(\frac{2\sigma^2-1}{2\sigma^2}\,t^2\right),
\]
which has the same sign as \(t\). Hence, on the usual parametrization \(t\in[0,\infty)\) (which covers \(\alpha\in[0,1)\)), we have \(f_{\ast}''(\alpha)\ge 0\), i.e.\ \(f_{\ast}\) is convex and strictly decreasing on \((0,1)\).

Using \(\alpha=\alpha(t)\) as a change of variables, we have
\[
d\alpha=2\phi(t)\,dt=\sqrt{\tfrac{2}{\pi}}e^{-t^2/2}\,dt,
\]
and
\[
\log|f_{\ast}'(\alpha(t))|
=\log\!\left(\frac{1}{\sigma}\right)-\frac{1-\sigma^2}{2\sigma^2}\,t^2.
\]
Let
\[
t_0:=\sqrt{\frac{2\sigma^2}{1-\sigma^2}\,\log\!\left(\frac{1}{\sigma}\right)}\in(0,\infty),
\]
so that \(\log|f_{\ast}'(\alpha(t_0))|=0\). Then
\[
\kappa_3(f_{\ast})
=\sqrt{\frac{2}{\pi}}
\int_0^\infty
\Bigl|\log\!\left(\frac{1}{\sigma}\right)-\frac{1-\sigma^2}{2\sigma^2}\,t^2\Bigr|^3
e^{-t^2/2}\,dt.
\]
Split the integral at \(t_0\). On \([0,t_0]\), the integrand is continuous and bounded, hence integrable. On \([t_0,\infty)\), there exists a constant \(C_1>0\) such that
\[
\Bigl|\log\!\left(\frac{1}{\sigma}\right)-c\,t^2\Bigr|^3
\le C_1(1+t^6),
\qquad c=\frac{1-\sigma^2}{2\sigma^2}.
\]
Therefore,
\[
\kappa_3(f_{\ast})
\le
\sqrt{\frac{2}{\pi}}
\int_0^\infty C_1(1+t^6)e^{-t^2/2}\,dt
<\infty,
\]
because a polynomial times a Gaussian tail is integrable \cite{gradshteyn2014table}.

\noindent\textbf{Step 2: Proving \(\kappa_{3}(f_{\ast}^{-1})<\infty\).}
A standard change of variables \cite{rudin1974real} gives
\[
\kappa_3(f_{\ast}^{-1})
=
\int_0^1
\Bigl|\log\frac{1}{|f_{\ast}'(\alpha)|}\Bigr|^3
|f_{\ast}'(\alpha)|\,d\alpha.
\]
Substituting the parametrization above yields
\[
\kappa_3(f_{\ast}^{-1})
=
\sqrt{\frac{2}{\pi}}\,\frac{1}{\sigma}
\int_0^\infty
\Bigl|\log\sigma+\frac{1-\sigma^2}{2\sigma^2}\,t^2\Bigr|^3
\exp\!\left(-\frac{t^2}{2\sigma^2}\right)\,dt.
\]
Again, the integrand is bounded on compact intervals and has the form
\[
\text{(polynomial in \(t\))}\times e^{-t^2/(2\sigma^2)}
\]
for large \(t\), hence it is integrable. Therefore \(\kappa_3(f_{\ast}^{-1})<\infty\).

\noindent\textbf{Step 3: Proving \(\kappa_{3}(\overline{f_{\ast}})<\infty\).}
Since \(\alpha(t)\) is analytic and strictly increasing on \((0,\infty)\), its inverse
\[
t(\alpha)=\Phi^{-1}\!\left(\frac{1+\alpha}{2}\right)
\]
is analytic on \((0,1)\). Hence
\[
f_{\ast}(\alpha)
=
2\Bigl(1-\Phi\!\bigl(t(\alpha)/\sigma\bigr)\Bigr)
\]
is real analytic on \((0,1)\). Because \(f_\ast\) is real analytic on \((0,1)\) and
\(f_\ast'(\alpha)< 0\) for all \(\alpha\in(0,1)\), the real analytic inverse function theorem~\cite[Theorem~1.5.3]{krantz2002primer} implies that \(f_\ast^{-1}\) is real analytic on \((0,1)\).

Now define
\[
h(\alpha):=f_{\ast}(\alpha)-f_{\ast}^{-1}(\alpha).
\]
Then \(h\) is real analytic on \((0,1)\). If \(h\equiv 0\), then \(\overline{f_{\ast}}=f_{\ast}\), and the conclusion follows immediately from Step~1. Otherwise, \(h\not\equiv 0\). Since \(h\) is real analytic, Corollary~1.2.7 of~\cite{krantz2002primer} implies that the zero set of \(h\) has no accumulation point in \((0,1)\). Hence the zero set \[
Z:=\{\alpha\in(0,1): h(\alpha)=0\}
\] is discrete, countable, and of Lebesgue measure zero.

On each connected component of \((0,1)\setminus Z\), the sign of \(h\) is constant. Therefore, on each such component, \(\overline{f_{\ast}}\) coincides either with \(f_{\ast}\) or with \(f_{\ast}^{-1}\). Consequently, for almost every \(\alpha\in(0,1)\),
\[
(\overline{f_{\ast}})'(\alpha)=
\begin{cases}
f_{\ast}'(\alpha), & \text{if } f_{\ast}(\alpha)>f_{\ast}^{-1}(\alpha),\\[4pt]
(f_{\ast}^{-1})'(\alpha), & \text{if } f_{\ast}(\alpha)<f_{\ast}^{-1}(\alpha).
\end{cases}
\]
Since \(Z\) is a null set, we obtain
\[
\kappa_{3}(\overline{f_{\ast}})
=
\int_{\{h>0\}}
\bigl|\log|f_{\ast}'(\alpha)|\bigr|^{3}\,d\alpha
+
\int_{\{h<0\}}
\bigl|\log|(f_{\ast}^{-1})'(\alpha)|\bigr|^{3}\,d\alpha.
\]
Hence
\[
\kappa_{3}(\overline{f_{\ast}})
\le
\kappa_{3}(f_{\ast})+\kappa_{3}(f_{\ast}^{-1})
<\infty.
\]
Thus \(\kappa_{3}(\overline{f_{\ast}})<\infty\), completing the proof.
\end{proof}

\end{document}